\documentclass[oneside]{article}
\usepackage{amsmath, amsthm}
\usepackage{amsfonts}
\usepackage{amssymb}
\usepackage{color}
\usepackage[colorlinks=true, linkcolor=black, citecolor=black,pdfborder={0 0 0}]{hyperref}
\usepackage{multido}
\usepackage{url}
\usepackage{mathrsfs}
\usepackage{enumerate}
\usepackage{verbatim}
\usepackage{setspace}
\usepackage{cleveref}
\crefname{theorem}{Theorem}{Theorems}
\crefname{lemma}{Lemma}{Lemmas}
\crefname{proposition}{Proposition}{Propositions}
\crefname{section}{\S}{\S\S}
\crefname{equation}{}{}
\usepackage[margin=3.25cm]{geometry}
\usepackage{hyperref}
\usepackage{tikz}

\newtheorem{theorem}{Theorem}[section]
\newtheorem{remark}[theorem]{Remark}
\newtheorem{lemma}[theorem]{Lemma}
\newtheorem{proposition}[theorem]{Proposition}
\newtheorem{definition}[theorem]{Definition}
\newtheorem{corollary}[theorem]{Corollary}

\newcommand{\R}{\mathbb{R}}
\newcommand{\N}{\mathbb{N}}
\newcommand{\C}{\mathbb{C}}
\newcommand{\E}{\mathbb{E}}
\newcommand{\Z}{\mathbb{Z}}

\renewcommand{\Pr}{\mathbb{P}}
\newcommand{\B}{\mathbb{B}}
\newcommand{\s}{\sigma}

\renewcommand{\Re}{\operatorname{Re}}

\newcommand{\x}{{\bf x}}

\newcommand{\ket}[1]{|{#1}\rangle}
\newcommand{\bra}[1]{\langle{#1}|}

\newcommand{\LL}{\mathcal{L}}
\newcommand{\WW}{\mathcal{W}}

\newcommand{\textd}{\textrm{d}}
\DeclareMathOperator{\tr}{\mathrm{Tr}}
\newcommand{\Hil}{\mathcal{H}}
\title{Emptiness Formation Probability}
\author{Nicholas Crawford$^{1}$, Stephen Ng$^{2}$ and Shannon Starr$^{3}$\\[10pt]
$^{1}$ Department of Mathematics, The Technion, Haifa, Israel\\
$^{2}$ Department of Mathematics, University of Rochester, Rochester, NY\\
$^{3}$ Department of Mathematics, University of Alabama at Birmingham, Birmingham, AL}
\date{October 14, 2014}
\begin{document}
\newcommand{\Ru}{-- ++(2.5,0) arc (240:-60:1cm) -- ++(2.5,0)}
\newcommand{\Rd}{-- ++(2.5,0) arc (-240:60:1cm) -- ++(2.5,0)}
\newcommand{\Ur}{-- ++(0,2.5) arc(-150:150:1cm) -- ++(0,2.5)}
\newcommand{\Ul}{-- ++(0,2.5) arc (-30:-330:1cm) -- ++(0,2.5)}
\newcommand{\Lu}{-- ++(-2.5,0) arc (-60:240:1cm) -- ++(-2.5,0)}
\newcommand{\Ld}{-- ++(-2.5,0) arc (60:-240:1cm) -- ++(-2.5,0)}
\newcommand{\Dr}{-- ++(0,-2.5) arc(150:-150:1cm) -- ++(0,-2.5)}
\newcommand{\Dl}{-- ++(0,-2.5) arc (-330:-30:1cm) -- ++(0,-2.5)}
\newcommand{\urdl}{\Ru \Ur \Ld \Dl}
\newcommand{\dlur}{\Rd \Ul \Lu \Dr}
\newcommand{\urur}{\Ru \Ur \Lu \Dr}
\newcommand{\dldl}{\Rd \Ul \Ld \Dl}
\newcommand{\ulul}{\Ru \Ul \Lu \Dl}
\newcommand{\drdr}{\Rd \Ur \Ld \Dr}
\newcommand{\Ruhf}{-- ++(2.5,0) arc (240:90:1cm)}
\newcommand{\Rdhf}{-- ++(2.5,0) arc (-240:-90:1cm)}
\newcommand{\Urhf}{-- ++(0,2.5) arc(-150:0:1cm)}
\newcommand{\Ulhf}{-- ++(0,2.5) arc (-30:-180:1cm)}
\newcommand{\Luhf}{-- ++(-2.5,0) arc (-60:90:1cm)}
\newcommand{\Ldhf}{-- ++(-2.5,0) arc (60:-90:1cm)}
\newcommand{\Drhf}{-- ++(0,-2.5) arc(150:0:1cm)}
\newcommand{\Dlhf}{-- ++(0,-2.5) arc (-330:-180:1cm)}
\newcommand{\EFP}{{\mathtt{EFP}}}
\newcommand{\bi}{\mathbf{i}}
\newcommand{\bj}{\mathbf{j}}
\newcommand{\be}{\mathbf{e}}
\newcommand{\T}{\mathbb{T}}

\newcommand{\G}{\mathcal{G}}
\newcommand{\V}{\mathcal{V}}
\newcommand{\Ed}{\mathcal{E}}
\newcommand{\half}{\textstyle{\frac{1}{2}}}

\newcommand{\Ran}{{\textrm{Ran}}}

\maketitle
\begin{abstract}
\setcounter{section}{0}
We present rigorous upper and lower bounds on the 
emptiness formation probability for the ground state of a spin-$1/2$ Heisenberg XXZ quantum spin system.
For a $d$-dimensional system we find a rate of decay of the order $\exp(-c L^{d+1})$ where $L$ is the sidelength of the box in which we ask for the emptiness formation event to occur. In the $d=1$ case this confirms previous predictions made in the integrable systems community, though our bounds do not achieve the precision predicted by Bethe ansatz calculations. On the other hand, our bounds in the case $d
\geq 2$ are new. The main tools we use are reflection positivity and a rigorous path integral expansion which is a variation on those previously introduced by Toth, Aizenman-Nachtergaele and Ueltschi.
\end{abstract}
\section{Introduction and Main Results.}
In this paper we obtain mathematically rigorous bounds for a quantity that physicists have considered
for some time, called the ``emptiness formation probability.''
This is the probability, in the ground state of the quantum Heisenberg antiferromagnet, to find a block of spins ferromagnetically aligned.
In a classical model, such as the Ising model, this probability would be zero in the true ground state.
It is a measure of the quantum nature of the Heisenberg antiferromagnet that this probability is not exactly zero,
even in the ground state.

The expected answer in $d=1$ is that the emptiness formation probability for a block of length $L$
is asymptotically $AL^{\nu} \exp(-cL^{2})$, in the limit $L \to \infty$, where $A,\nu >0$ are independent of L. 
This behavior was determined by physicists for some special Bethe-ansatz solvable models, although
part of their analysis is not rigorous.
We prove, in any dimension $d>0$ of the underlying lattice, there are upper and lower bounds of the form $C_{\pm} \exp(-c_{\pm} L^{d+1})$
for constants $C_+,C_-,c_+,c_- \in (0,\infty)$.
This is certainly an easier explanation than the Bethe ansatz.
In contrast to Bethe ansatz derivations.
 the reasons behind our bounds is transparent.
The exponent scales as $L^{d+1}$ instead of $L^d$ (as in classical statistical mechanical
models at positive temperature) because the extra dimension is {\em imaginary time} in the graphical representation of the quantum model.

The mathematical analysis is still somewhat involved. 
For convenience, we restrict attention to {\em reflection positive} models.
Luckily, many interesting physical models are
reflection positive, including all those considered using the Bethe ansatz.

\subsection{Set-up}
Let $\G = (\V,\Ed)$ be a finite graph.  Using the notation of graph theory, $\Ed$ may be any subset of the collection of all pairs $\{\bi,\bj\}$ such that $\bi,\bj \in \V$, $\bi\neq \bj$.

To define the Heisenberg XXZ models, we begin by introducing its underlying Hilbert space.  In general,
spin-$\frac{1}{2}$ quantum spin systems  the Hilbert space is $\Hil_\V = (\C^2)^{\otimes |\V|}$, one factor for each vertex of the graph.
We denote the usual Pauli spin matrices, normalized by $\frac{1}{2}$, as
\begin{equation}
\label{eq:Pauli}
S^x\, =\, \frac{1}{2}\, \begin{bmatrix} 0 & 1 \\ 1 & 0 \end{bmatrix}\, ,\quad
S^y\, =\, \frac{1}{2}\, \begin{bmatrix} 0 & -i \\ i & 0 \end{bmatrix}\ \text{ and } \
S^z\, =\, \frac{1}{2}\, \begin{bmatrix} 1 & 0 \\ 0 & -1 \end{bmatrix}\, .
\end{equation}
Let $\V \, =\, \{\bi_1,\dots,\bi_{|\V|}\}$ be any enumeration of the vertices in order to specify the spin matrices at the various sites.
The choice of enumeration is immaterial since a re-ordering will just result in a unitarily equivalent representation.
For each $n$, the spin matrices at $\bi_n$ are
$S^{x}_{\bi_n}$, $S^y_{\bi_n}$ and $S^z_{\bi_n}$, where
$S^{x}_{\bi_n}\, =\, ({I}_{\C^2})^{\otimes (n-1)} \otimes S^{x} \otimes ({I}_{\C^2})^{\otimes (|\V|-n)}$,
with similar formulas for $S^y_{\bi_n}$ and $S^z_{\bi_n}$.

There is a real parameter of the model $\Delta \in \R$, called the anisotropy parameter.
With this the XXZ Hamiltonian is a self adjoint operator $H_{\G,\Delta} : \Hil_\V \to \Hil_\V$, defined as
\begin{equation}
\label{eq:PFready}
H_{\G,\Delta}\, =\, -\sum_{\{\bi,\bj\} \in \Ed} (S^x_{\bi} S^x_{\bj} + S^y_{\bi} S^y_{\bj} + \Delta \cdot S^z_{\bi} S^z_{\bj})\, .
\end{equation}
$H_{\G,\Delta}$ is ferromagnetic for $\Delta>0$ and antiferromagnetic for $\Delta<0$.

In this paper, we will restrict to the special case that $\G$ is bipartite: $\exists A \subseteq \V$ such that every edge in $\Ed$ can be written as $\{\bi,\bj\}$ with $\bi \in A$ and $\bj \in \V \setminus A$.
In this case, and for our graphical representations below, let us recall how $H_{\G,\Delta}$ transforms under the unitary $U_A = \prod_{\bi \in A} (2S_{\bi}^z)$--
$$
U_A H_{\G,\Delta} U_A^*\, =\, \sum_{\{\bi,\bj\} \in \Ed} (S^x_{\bi} S^x_{\bj} + S^y_{\bi} S^y_{\bj} - \Delta \cdot S^z_{\bi} S^z_{\bj})\, .
$$
We define the usual thermodynamic quantities: the partition function 
$$
Z_{\G,\Delta}(\beta)\, :=\, \tr(e^{-\beta H_{\G,\Delta}})\, ,
$$
and the equilibrium state
$$
\langle X \rangle_{\G,\Delta,\beta}\, :=\, \frac{\tr(Xe^{-\beta H_{\G,\Delta}})}{Z_{\G,\Delta}(\beta)}\, .
$$

For any $N \in \N$, let $\B_N$ denote the box
\begin{equation}
\label{eq:BNdef}
\B_N\, =\, \{\bi = (i_1,\dots,i_d)\in \Z^d\, :\, -\textstyle{\frac{1}{2}} N < i_1,\dots,i_d \leq \textstyle{\frac{1}{2}}N\}\, =\, \{-\lceil \half N\rceil+1,\dots,\lfloor \half N \rfloor\}^d\, ,
\end{equation}
where $\lceil x \rceil = \min\{k \in \Z\, :\, k\geq x\}$, $\lfloor x \rfloor = \max\{ k \in \Z\, :\, k \leq x\}$.  The parameter $d \in \N$ is the dimension of the underlying box.  Our main line of argument holds for all $d\in \N$ and we will usually leave the dependence on $d$ implicit, in order to simplify the notation.

Given $N \in \N$, let $\mathbb{T}_N$ denote the graph $\G = (\V,\Ed)$ such that $\V = \B_N$ (the box of sidelength $N$) and
\begin{equation}
\label{eq:TNdef}
\Ed\, =\, \Ed_{\mathbb{T}_N}\,
=\, \{\{\bi,\bj\}\, :\, \bi,\bj \in \V\, ,\ \bj - \bi \in \{\be_1,\dots,\be_d,-(N-1)\be_1,\dots,-(N-1)\be_d\}\}\, ,
\end{equation}
where $\be_1,\dots,\be_d$ are the usual canonical basis vectors in $\Z^d$.
This is the discrete torus because of the periodic boundary conditions.
Frequently we will abuse notation by also writing $\T_N$ for the vertex set. In particular, when we write the cardinality $|\T_N|$ this will denote the cardinality of the vertex set, which is $N^d$.  We will always restrict attention to $N$ even, so that $\mathbb{T}_N$ is bipartite.
We write $H_{N,\Delta}$, $Z_{N,\Delta}(\beta)$ and $\langle X \rangle_{N,\Delta,\beta}$ in place of $H_{\mathbb{T}_N,\Delta}$, $Z_{\mathbb{T}_N,\Delta}(\beta)$ and $\langle X \rangle_{\mathbb{T}_N,\Delta,\beta}$.

\subsection{Emptiness Formation Probability}
Recall that the eigenvalues of the spin matrix $S^z_{\bi}$ are $\pm\frac{1}{2}$ so that the two operators $(\frac{1}{2} \pm S^z_{\bi})$ are the projections onto the eigenspaces associated
with the eigenvalues $\pm\frac{1}{2}$.
As long as $N\geq L$, we may view $\B_L$ as a subset of the graph $\T_N$ (whose vertex set is $\B_N$).
We define the projection operator
$$
\mathbf{Q}_{L}\, =\, \prod_{\bi\in \B_L} \left[\frac{1}{2} + S_{\bi}^z\right]\, .
$$
The range of $\mathbf{Q}_{L}$ is the subspace spanned by all spin states having all spins up on the sub-box $\B_L$.
The expectation of $\mathbf{Q}_L$ in the ground state of the XXZ model is called the \textbf{emptiness formation
probability} in the physics literature.

\begin{theorem}
\label{thm:Main1}
Suppose the dimension $d$ is fixed in $\{1,2,\dots\}$.
For each $\Delta \in [-1,1)$, there are
constants $c_1, C_1 \in (0,\infty)$ such that, whenever $L^d \leq N^d/2$
\begin{equation}
\label{E:Ld2}
C_1\exp\left(-c_1 L^{d} \min(L, \beta)\right)\, \leq\, \left \langle \mathbf{Q}_{L} \right \rangle_{N,\Delta,\beta}\, ,
\end{equation}
while if $\Delta\leq 0$ there are constants $c_2,C_2 \in (0,\infty)$ such that,
\begin{equation}
\label{E:Ud2}
\left \langle \mathbf{Q}_{L} \right \rangle_{N,\Delta,\beta}\, \leq\, C_2 \exp\left(-c_2 L^{d} \min(L, \beta)\right)\, .
\end{equation}
\end{theorem}
The lower bound of \cref{E:Ld2} will be proved in \Cref{sec:LBthm1} while the upper bound of \cref{E:Ud2} is the subject of \Cref{S:UB}.

When $d=1$ we may obtain extended results at zero temperature, 
but we must take account  a symmetry of the XXZ model.
For each $\Delta \in \R$, the Hamiltonian $H_{N,\Delta}$ commutes with the operator
$$
S^z_{\mathrm{tot}}\, =\, \sum_{\bi \in \B_N} S_{\bi}^z\, .
$$
The eigenvalues of this operator are $M \in \{-\frac{1}{2}\, N^d,\dots,+\frac{1}{2}\, N^d\}$.
Let $\mathbf{M}_M$ denote the orthogonal projection onto the eigenspace of $S^{z}_{\mathrm{tot}}$ corresponding to eigenvalue $M$.

We recall that $N$ is even.
\begin{theorem}
\label{thm:1d}
Suppose the dimension is $d=1$. For each $\Delta < 1$, there are constants $c_i,C_i \in (0,\infty), \; i=1,2$ such that, whenever $L\leq N/2$
\begin{equation}
\label{E:Ld3}
C_1\exp\left(-c_1 L^2\right)\, \leq\, \min_{M \in \{-\frac{1}{2}N,\dots,\frac{1}{2}N\}}\, \lim_{\beta \to \infty}
\frac{\left \langle \mathbf{M}_{M} \cdot \mathbf{Q}_{L} \right \rangle_{N,\Delta,\beta}}
{\left \langle \mathbf{M}_M \right \rangle_{N,\Delta,\beta}}\, \text{ and }\lim_{\beta \to \infty} 
\frac{\left \langle \mathbf{M}_{0} \cdot \mathbf{Q}_{L} \right \rangle_{N,\Delta,\beta}}
{\left \langle \mathbf{M}_0 \right \rangle_{N,\Delta,\beta}}\, \leq\, C_2 \exp\left(-c_2 L ^2\right)\, .
\end{equation}
\end{theorem}
The reason we have a stronger result in one dimension is that there its groundstate maps to the six vertex model, a point described in \cref{sec:1d6vtx}.
The upper bound of \cref{E:Ld3} will be proved in \cref{sec:UB6} while the lower bound of \cref{E:Ld3} will be proved in \cref{sec:LBthm2}.

Finally, at positive temperatures, if we take sufficiently large $L$ then the lower bound holds with no restrictions on $\Delta$.
\begin{theorem}
For any fixed $d \in \N$ and any $\Delta \in \R$, there are constants $c$ and $C$ such that whenever $0\leq \beta\leq 4L$
$$
C\exp\left(-c L^{d} \beta\right)\, \leq\, \left \langle \mathbf{Q}_{L} \right \rangle_{N,\Delta,\beta}\, .
$$
\end{theorem}

\subsection{Background and Motivation}
The background motivation for our investigation stems from a few sources. For the rest of this section, we assume $d=1$.
The name "emptiness formation probability" ($\EFP_L$) seems to come from the computation of (say) density-density correlation functions in the $1$-dimensional hardcore Bose gas. 
The generating functional for this and other correlation functions is $\langle e^{\alpha Q(x)} \rangle$ where 
$$
Q(x)\, =\, \int_0^x \textrm{d} y \Psi^{\dag}(y) \Psi(y)\, .
$$ 
Here $ \Psi^{\dag}(y), \Psi(y)$ are field operators for the Bose gas and $\alpha \in \C$. As $\Re \alpha \rightarrow - \infty$, $e^{\alpha Q(x)}$ converges (weakly say) to the projection operator onto the subspace with no particles present in $[0, x]$. It turns out that $\langle e^{\alpha Q(x)} \rangle$ is more easily computed using Bethe ansatz techniques than various other correlation functions \cite{Korepin-Book}.
More recently, it was argued that $\EFP_L$ is of primary importance for the ground state correlation structure of $XXZ$ chains \cite{BK}.  In fact computing exactly, in the thermodynamic limit, $\EFP_L$ for all $L$ would allow to compute many other correlation functions as well. This is prohibitively complicated once $L >5$, even by the standards of the Bethe ansatz , and in \cite{KLNS} the authors focused instead on the asymptotic behavior of $\EFP_L$ in $L$. 

Further (nonrigorous) work appears in \cite{KMST,Stephan} to name just a few articles.
In fact, there \textit{are} at least two cases where the asymptotic computations may be made rigorous: when $\Delta \in \{0, 1/2\}$, \cite{KMST0,STN}. The latter is special as it corresponds to the uniform measure on $6$-vertex configurations (a short explanation is given below). The former is special because, via the Jordan-Wigner transformation, it is the equivalent to the model of free fermions. One can therefore write all eigenvectors of $-H_{\Delta=0}$ as Slater determinants of $1$-particle eigenfunctions of the discrete Laplace operator on a one
dimensional torus. In particular the groundstate is explicit.

To make a connection with currently fashionable mathematical phyics let us restate the quantum $XY$ model (that is \cref{eq:PFready} with $\Delta=0$) correspondence to free fermions in a probabilistic language. Consider for a moment a collection of $k$ independent continuous time simple random walkers on a discrete circle of $N$ vertices, conditioned not to collide. Then via the Karlin-Mcgregor formula \cite{KarlinMcGregor}, the quasi-stationary measure for this process is exactly the square of the amplitude of the groundstate wave function of the free fermion model restricted to its $k$ particle sector.
In this language, computing the asymptotics of $\EFP_L$ translates to computing the large deviation rate of decay for large gaps in a "Dyson" random walk, with density $\frac 12$ of particles. From this perspective, the generalization from $1$ to $d$ dimensions is quite natural and forms the basic heuristic explanation of the $e^{-cL^{d+1}}$ rate of decay. 

We originally learned of this problem from O. Zeitouni \cite{Zeitouni}, whose interest was piqued by the resemblance between the multiple integral representation of XXZ correlation functions arising from the Bethe ansatz and certain computations from random matrix theory.
A third motivation for our work, which should be contrasted with the (non-rigorous) Bethe ansatz methods mentioned above is to give an explanation for these very strong rates of decay which is robust to perturbations of the model and makes intuitive sense physically. As is explained below, our proof for the second half of of \Cref{thm:Main1} achieves this goal; for example it is possible changing the range of interaction in $-H_{\Delta}$ and obtain the same results, as long as the interaction remains reflection positive and the interaction strength decays exponentially with distance.

Finally, we wish to draw the reader's attention to related work of Gallavotti, Lebowitz and Mastropietro \cite{GalLebMas}.
In that paper they proved a large deviation principle for quantum gases, either Bosonic or Fermionic, under
the condition of sufficiently low densities.
More precisely, they exhibit the usual large deviation type of behavior probability  to find density $\rho$ in a prescribed large box of volume $L^d$: the probability decays as $\exp(-\beta \Delta F(\beta,\rho,\rho_0)L^d)$.  The function $\Delta F(\beta,\rho,\rho_0)$ represents the change in free energy density at inverse-temperature $\beta$ between the prescribed density $\rho$ and the thermodynamic density $\rho_0$.

\subsection*{Acknowledgements}
We thank Dmitry Ioffe for discussions at various points in the work and an anonymous referee for suggestions which improved the presentation. NC is supported by ISF grant number 915/12.

\section{The Lower Bound \cref{E:Ld2} in \Cref{thm:Main1}.}
\label{sec:LBthm1}

\subsection{The Aizenman-Nachtergaele-Toth-Ueltschi representation.}
\label{S:graphical}

The origins of this graphical representation we present below are \cite{AN, Toth}, but the synthesis of ferromagnetic and antiferromagnetic loops into the same representation only appeared in the recent paper \cite{U}, see Section $3$.
The construction works for any bipartite finite graph $\G = (\V,\Ed)$.
For inverse-temperature $\beta \geq 0$, the underlying space for this construction is $\V_{\beta} = \V \times [-\frac{1}{2}\beta,\frac{1}{2}\beta]$.


For each edge $\{\bi,\bj\}$ in $\Ed$, 
we consider two independent Poisson point processes on $\R$,
which we label as random (nonnegative integer valued-) measures $\textd \nu_{\bi\bj}^F(\omega)$,
$\textd \nu_{\bi\bj}^{AF}(\omega)$ with respective rates $\frac{1}{2}(1-u),  \frac{1}{2}u$.  associated with the family of Poisson point processes $\bigotimes_{\{\bi,\bj\} \in E}[\textd \nu_{\bi\bj}^F(\omega) \otimes \textd \nu_{\bi\bj}^{AF}(\omega)]$. 
All these Poisson point processes are taken to be independent of one another.  Let $(\Omega, \mathcal F, \mathbb P_u(\cdot))$ be a probability space supporting these processes and let $\E_u[\cdot]$ denote the corresponding e{x}pectation.  We shall regard $\omega \in \Omega$ as represented a configuration of edges on $\V_{\beta}$ as follows.
If $t$ is an arrival time of $\textd \nu_{\bi\bj}^F(\omega)$ then we create
an \textbf{overpass edge},
while for an arrival time of $\textd \nu_{\bi\bj}^{AF}(\omega)$ then we create
a \textbf{cul-de-sac edge}.
An example of this is shown in Figure \ref{fig:ANTU1}
\begin{figure}[h]
\begin{tikzpicture}[xscale=0.85,yscale=0.65, very thick]
\foreach \x in {1,...,6}
\draw[thin] (\x,1) -- (\x,7);
\draw[<->] (0.75,0.5) -- (6.25,0.5);
\draw (3.5,0.5) node[below] {$G = \mathbb{T}_6$, $d=1$};
\draw[->] (-0.25,0.25) -- (-0.25,7.5) node[above] {$t$};
\draw[thick] (0,1) -- (-0.5,1) node[left] {$-1$};
\draw[thick] (0,7) -- (-0.5,7) node[left] {$1$};
\draw (-1.5,4) node[] {$\beta=2$};
\fill[white] (0.9,1.7) rectangle (2.1,1.9); \draw (1,1.7) -- (2,1.7); \draw (1,1.9) -- (2,1.9);
\fill[white] (1.9,4.7) rectangle (3.1,4.9); \draw (2,4.7) -- (3,4.7); \draw (2,4.9) -- (3,4.9);
\fill[white] (1.9,5.4) rectangle (3.1,5.6); \draw (2,5.4) -- (3,5.4); \draw (2,5.6) -- (3,5.6);
\fill[white] (2.9,4.1) rectangle (4.1,4.3); \draw (3,4.1) -- (4,4.1); \draw (3,4.3) -- (4,4.3);
\fill[white] (3.9,6.5) rectangle (5.1,6.7); \draw (4,6.5) -- (5,6.5); \draw (4,6.7) -- (5,6.7);
\fill[white] (5.9,6.6) rectangle (6.5,6.8); \draw[->] (6,6.6) -- (6.5,6.6); \draw[->] (6,6.8) -- (6.5,6.8);
\fill[white] (1.1,6.6) rectangle (0.5,6.8); \draw[->] (1,6.6) -- (0.5,6.6); \draw[->] (1,6.8) -- (0.5,6.8);
\fill[white] (0.9,2.9) rectangle (2.1,3.1); \draw (1,2.9) -- (2,3.1); \draw[line width=1mm,white] (1,3.1) -- (2,2.9); \draw (1,3.1) -- (2,2.9);
\fill[white] (0.9,3.6) rectangle (2.1,3.8); \draw (1,3.6) -- (2,3.8); \draw[line width=1mm,white] (1,3.8) -- (2,3.6); \draw (1,3.8) -- (2,3.6);
\fill[white] (1.9,2.4) rectangle (3.1,2.6); \draw (2,2.4) -- (3,2.6); \draw[line width=1mm,white] (2,2.6) -- (3,2.4); \draw (2,2.6) -- (3,2.4);
\fill[white] (1.9,3.2) rectangle (3.1,3.4); \draw (2,3.2) -- (3,3.4); \draw[line width=1mm,white] (2,3.4) -- (3,3.2); \draw (2,3.4) -- (3,3.2);
\fill[white] (1.9,6.5) rectangle (3.1,6.7); \draw (2,6.5) -- (3,6.7); \draw[line width=1mm,white] (2,6.7) -- (3,6.5); \draw (2,6.7) -- (3,6.5);
\fill[white] (3.9,3.4) rectangle (5.1,3.6); \draw (4,3.4) -- (5,3.6); \draw[line width=1mm,white] (4,3.6) -- (5,3.4); \draw (4,3.6) -- (5,3.4);
\fill[white] (3.9,6.0) rectangle (5.1,6.2); \draw (4,6.0) -- (5,6.2); \draw[line width=1mm,white] (4,6.2) -- (5,6.0); \draw (4,6.2) -- (5,6.0);
\fill[white] (4.9,4.7) rectangle (6.1,4.9); \draw (5,4.7) -- (6,4.9); \draw[line width=1mm,white] (5,4.9) -- (6,4.7); \draw (5,4.9) -- (6,4.7);
\draw[semitransparent, line width=1.5mm, gray, rounded corners] (1,1) -- (1,1.7) -- (2,1.7) -- (2,1)
(2,7) -- (2,6.7) -- (3,6.5) -- (3,5.6) -- (2,5.6) -- (2,6.5) -- (3,6.7) -- (3,7) (3,1) -- (3,2.4) -- (2,2.6) -- (2,2.9) -- (1,3.1) -- (1,3.6) -- (2,3.8) -- (2,4.7) -- (3,4.7) -- (3,4.3) -- (4,4.3) -- (4,6.0) -- (5,6.2) -- (5,6.5) -- (4,6.5) -- (4,6.2) -- (5,6.0) -- (5,4.9) -- (6,4.7) -- (6,1) (6,7) -- (6,6.8) -- (6.5,6.8) (0.5,6.8) -- (1,6.8) -- (1,7);
\draw(7.5,3.75) node[right] {\begin{minipage}{10cm}\small 
Overpass edges:\quad
\raisebox{0cm}{\begin{minipage}{1.25cm}
\begin{tikzpicture}[xscale=0.85,yscale=0.65, very thick]
\foreach \x in {1,2}
\draw[thin] (\x,2.5) -- (\x,3.5);
\fill[white] (0.9,2.9) rectangle (2.1,3.1); \draw (1,2.9) -- (2,3.1); \draw[line width=1mm,white] (1,3.1) -- (2,2.9); \draw (1,3.1) -- (2,2.9);
\end{tikzpicture}
\end{minipage}}
\\[0.5cm]
Cul-de-sac edges:\quad
\raisebox{0cm}{\begin{minipage}{1.25cm}
\begin{tikzpicture}[xscale=0.85,yscale=0.65, very thick]
\foreach \x in {1,2}
\draw[thin] (\x,1.3) -- (\x,2.3);
\fill[white] (0.9,1.7) rectangle (2.1,1.9); \draw (1,1.7) -- (2,1.7); \draw (1,1.9) -- (2,1.9);
\end{tikzpicture}
\end{minipage}}
\end{minipage}};
\end{tikzpicture}
\caption{
An instance of  labeled edges for the graph $G = \mathbb{T}_6$, when $d=1$, and a 
highlighted loop.
\label{fig:ANTU1}
}
\end{figure}
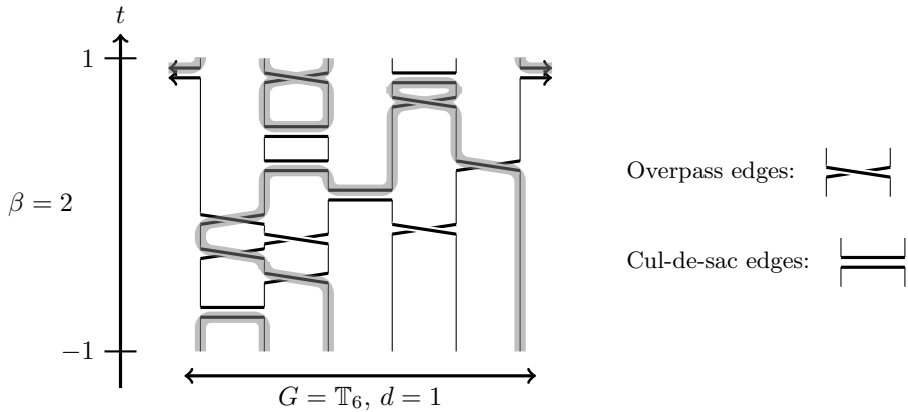

The nature of the edge-types is manifested in the rules assigning labels to $\V_{\beta}$, which we describe next.
Let $\Sigma_{\V}=\{-1, 1\}^{\V}$ and let $\Sigma_{\V,\beta}$ be the set of all piecewise constant functions 
$\sigma(\cdot) : [-\frac{1}{2}\beta,\frac{1}{2}\beta] \to \Sigma_{\V}$.

Let $\Sigma_{\V,\beta}(\omega)\subset \Sigma_{\V, \beta}$ denote the subset of functions which satisfy the following rules:
\begin{itemize}
\item
For any $\bi \in \V$, consider the union of the set of arrival times of all the Poisson processes
$\textd \nu_{\bi\bj}^F(\omega)$ 
and $\textd \nu_{\bi\bj}^{AF}(\omega)$ for all
$\bj$ such that $\{\bi,\bj\} \in E$. If the set of all these times is disjoint from 
$[t-\epsilon,t+\epsilon)$  then $\sigma_{\bi}(\cdot)$ is constant on this time interval.
\item If $t$ is an arrival time of $\textd \nu_{\bi\bj}^F(\omega)$, then 
$$
\sigma_{\bi}(t)\, =\, \sigma_{\bj}(t-)
\quad \text{ and }\quad
\sigma_{\bj}(t)\, =\, \sigma_{\bi}(t-)\,  .
$$
\item If $t$ is an arrival time of $\textd \nu_{\bi\bj}^{AF}(\omega)$, then
$$
\sigma_{\bi}(t-)\, =\, -\sigma_{\bj}(t-)
\quad \text{ and }\quad
\sigma_{\bi}(t)\, =\, -\sigma_{\bi}(t)\,  .
$$
\end{itemize}
Finally $\Sigma_{G,\beta}^{\mathrm{per}}(\omega)\subset \Sigma_{G,\beta}(\omega)$
consists of those labelings such that 
$\sigma_{\bi}(\frac{\beta}{2}-) = \sigma_{\bi}(-\frac{\beta}{2})$
for each $\bi \in V$, i.e. the periodic configurations in (imaginary) time.
We note that this is equivalent to identifying $t=\frac{\beta}{2}$ with $t=-\frac{\beta}{2}$.
In the periodic setting, the graph may be decomposed entirely into disjoint loops $\LL_{\beta}$. We have highlighted
one loop in $\LL_{\beta}$ in Figure \ref{fig:ANTU1}.
For each loop, once $\sigma_{\bi_0}(t_0)$ has been specified at a single space-time point $(\bi_0,t_0) \in \mathcal{V}_{\beta}$ on the loop, the rules
above prescribe it uniquely at each other space-time point on that loop.
Therefore $|\Sigma_{G,\beta}^{\mathrm{per}}(\omega)| = 2^{|\LL_{\beta}(\omega)|}$.

We will now specialize to $\G = \T_N$.
Suppose $\omega$ is a configuration of edges.
Given $L\leq N$ and given
$\tau \in \{-1, 1\}^{\mathbb{B}_L}$
let us define the event
\begin{equation}
\label{eq:EventDef}
E_{N,L,\beta}(\tau,\omega)\, =\, 
\{\sigma(\cdot) \in \Sigma_{N,\beta}^{\mathrm{per}}(\omega)\, :\, 
\sigma_{\bi}(0) =\tau_{\bi} \text{ for all } \bi \in \mathbb{B}_L\}.
\end{equation}
The reason we introduced the stochastic process described above is that it allows us to express various quantum spin system correlation functions in a 
more amenable probabilistic language. 

For $\C^2$, let $\psi^+$ and $\psi^-$ be the standard orthonormal basis, such that the spin
matrices in this basis have the form (\ref{eq:Pauli}).
Given any $\sigma \in \Sigma_{\mathcal{V}}$, define the Ising basis vector
\begin{equation}
\label{eq:PsiSigm}
\Psi_{\V}(\sigma)\, =\, \bigotimes_{\bi \in \mathcal{V}} \psi_{\bi}^{\sigma_\bi}\, .
\end{equation}
As usual, we write $\Psi_N(\sigma)$ as a short-hand for $\Psi_{\V}(\sigma)$ in the special case that $\V = \mathbb{T}_N$
(which will be used mainly in \S 3).

\begin{proposition}
\label{prop:Ueltschi}
For any $\s,\tau \in \Sigma_{\mathcal{V}}$,
\begin{equation}
\label{eq:Eform}
\big\langle \Psi_{\V}(\tau), e^{-\beta H_{\mathcal{G},\Delta}} \Psi_{\V}(\s) \big\rangle\, 
=\, e^{\beta |\mathcal{E}|/4}
\E_u\Bigg[\sum_{\sigma(\cdot) \in \Sigma_{\mathcal{V},\beta}(\omega)}
\mathbf{1}_{\{\sigma(-\beta/2) = \sigma\}} \mathbf{1}_{\{\sigma(\beta/2) = \tau\}}\Bigg]\, ,
\end{equation}
with the choice $u = (1+\Delta)/2$.
In particular, this means
\begin{align}
\label{EPar} & Z_{N,\Delta}(\beta)\, =\, e^{\beta |\mathcal{E}(\mathbb{T}_N)|/4} \E_u[2^{|\LL_\beta(\omega)|}], \\
\label{EEFP} & \EFP_L(N,\beta)\, =\, \frac{\E_u[|E_{N,L,\beta}(\mathbf{1}_L)|]}{ \E_u[2^{|\LL_{\beta}(\omega)|}]}\, ,
\end{align}
for $\mathbf{1}_L$ being the configuration with all $1$'s on $\Lambda_L$.
\end{proposition}
\begin{remark}
The fact that the weight factor in \cref{EPar} is  $2^{|\LL_{\beta}(\omega)|}$ has important implications.
The number of loops changes by at most 1 in absolute value if we add or subtract 
an arrival to one of the Poisson processes.
Therefore, $|\LL_{\beta}(\omega)|$ has a Lipschitz property with respect to the number of arrivals of $\omega$.
This is a useful property for obtaining large deviation type bounds.
\end{remark}
\begin{proof}
An equivalent result is proved in \cite{U}, Section 3.
This follows by considering the infinitesimal generator.
Given $\sigma(\cdot)$
in a small increment of time $0<\Delta t\ll 1$, there is a probability $\frac{1}{2} u \Delta t (1 + o(1))$ to have a cul-de-sac edge at $\{\bi,\bj\}$.
When there is a cul-de-sac edge, that is represented by the operator
$$
1 + 2Q_{\{\bi,\bj\}}\,
=\, 2\Big(\frac{1}{4} + \big(S_{\bi}^{x} S_{\bj}^{x} +  S_{\bi}^{y} S_{\bj}^{y}
- S_{\bi}^{z} S_{\bj}^{z}\big)\Big)
$$
In the same time increment there is a probability $(1-u) \Delta t(1+o(1))$
to have an overpass edge at $\{\bi,\bj\}$.
Then there is an operator
$$
1 + 2T_{\{\bi,\bj\}}\,
=\, 2\Big(\frac{1}{4} + \big(S_{\bi}^{x} S_{\bj}^{x} +  S_{\bi}^{y} S_{\bj}^{y}
+ S_{\bi}^{z} S_{\bj}^{z}\big)\Big)\, .
$$
There is a probability $1 - \Delta t(1+o(1))$ that there is no edge which is represented by
the identity operator $1$.
The reason we have shifted the operators above by 1 is for this reason.
Similarly, we scaled based on the fact that we chose the rates to be $1/2$ as large.
So we have
$$
Q_{\{\bi,\bj\}}\,
=\, \big(S_{\bi}^{x} S_{\bj}^{x} +  S_{\bi}^{y} S_{\bj}^{y}
- S_{\bi}^{z} S_{\bj}^{z}\big)-\frac{1}{4}\, ,\qquad
T_{\{\bi,\bj\}}\,
=\, \big(S_{\bi}^{x} S_{\bj}^{x} +  S_{\bi}^{y} S_{\bj}^{y}
+ S_{\bi}^{z} S_{\bj}^{z}\big)-\frac{1}{4}\, .
$$
From this, we see that in time increment $\Delta t$ the operator representing
all these effects is, to leading order, equal to
$$
\exp\Bigg(\Delta t \sum_{\{\bi,\bj\} \in \mathcal{E}} [u Q_{\{\bi,\bj\}} + (1-u) T_{\{\bi,\bj\}}]\Bigg)\, .
$$
Noting the relationship between $T_{\{\bi,\bj\}}$, $Q_{\{\bi,\bj\}}$ and
the XXZ interaction for $\Delta=1$, $\Delta=-1$ proves the lemma.
\end{proof}

\subsection{Derivation of the Lower Bound}
In this section, we present the lower bound of $C\exp(-c L^{d+1})$ for $d\geq 1$. This lower bound does NOT give the correct $L$ dependence for $\Delta \geq 1$. We do expect this bound to hold if $\Delta< -1$, but are currently unable to prove it. We refer the reader to \Cref{S:func} for a partial result in this direction. We will derive an upper bound of $C\exp(-c L^{d+1})$ in \Cref{S:UB}.
\subsubsection{The Case $\Delta \in [-1, 1]$.}
Let $\Delta \in [-1,1]$ be fixed. As before in \Cref{prop:Ueltschi}, we also fix $u = (1+\Delta)/2$.
For any $\tau \in \Sigma_{\Lambda_L}$, set
\[
Z_{N,L,\beta}(\tau)= \E_u[|E_{N,L,\beta}(\tau)|].
\]
where $E_{N,L,\beta}(\tau,\omega)$ follows the definition \cref{eq:EventDef}.
Since $u$ is fixed, we will write $\Pr$ and $\E$ instead of $\Pr_u$ and $\E_u$.
Because of \cref{EPar,EEFP}, the following lemma immediately gives the lower bound \cref{E:Ld2}.
\begin{lemma}
There are constants $C_1, c_1>0$ so that, assuming $\Delta \in [-1, 1]$, and $L\leq \frac{1}{2}N$,
\[
\frac{Z_{N,L,\beta}(\tau)}{ \E[2^{|\LL_{\beta}(\omega)|}]}\,
\geq\, C_1\exp(-c_1 L^{d} \min(\beta, L))
\]
for any $\tau \in \{-1, 1\}^{\Lambda_L}$.
\end{lemma}
\begin{proof}
We consider the case $d=1$ for ease of e{x}position and only the most relevant case $\tau= \mathbf{1}_L$. The entire argument e{x}tends without difficulty to the case $d>1$ and other choices of $\tau$.
We first consider the case $\beta\geq L$. We will comment on the case $\beta< L$ at the end of the proof.
We first observe the following: if $\omega, \omega_1$ are two configurations of edges which differ by $k$ space-time arrivals. Then
\begin{equation}
\label{E:Finite}
2^{-k} \leq 2^{|\LL_\beta(\omega)|-|\LL_\beta(\omega_1)|} \leq 2^k.
\end{equation}
Now let $\WW_L$ denote the space time window $\{-L+1,\dots,L\} \times [-L, L]$ and
let
\[
F=\{\omega\, :\, \text{there are no arrivals of $\omega$ in $\WW_{L}$}\}.
\]
As a simple consequence of \cref{E:Finite} and tail bounds for large numbers of Poisson arrivals, we have
constants $C, c>0$ so that
\[
\E[2^{|\LL_{\beta}(\omega)|}] \leq Ce^{c L^{2}} \E[2^{|\LL_{\beta}(\omega)|} \mathbf 1\{ F\} ].
\]
Ne{x}t, we bound $\E[2^{|\LL_{\beta}(\omega)|} \mathbf 1\{ F\} ]$ by $\E[|E_{N,L,\beta}(\mathbf 1_L)|]$.
Consider the event $G$ consisting of configurations of edges satisfying the following properties (see \Cref{fig:FandG} for an illustration):
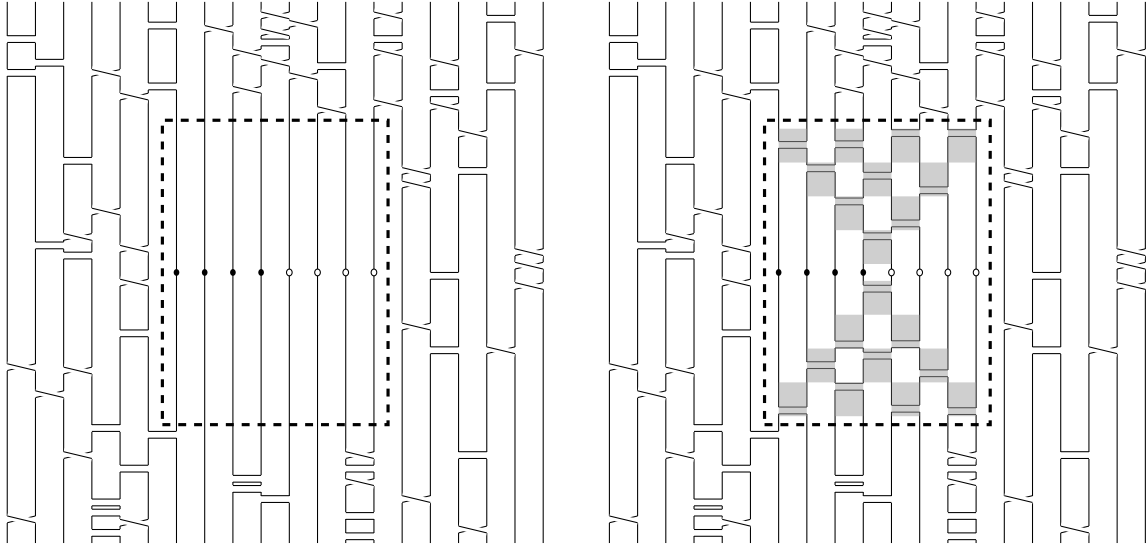
\begin{figure}
\begin{tikzpicture}[xscale=0.375,yscale=0.45]
\foreach \x in {1,...,20}
{
\draw (\x,-8) -- (\x,8);
}
\foreach \x in {7,...,10}
{
\fill (\x,0) circle (1mm);
}
\foreach \x in {11,...,14}
{
\filldraw[fill=white] (\x,0) circle (1mm);
}
\draw[very thick, dashed] (6.5,-4.5) rectangle (14.5,4.5);
\fill[white] (0.9,-7.4) rectangle (2.1,-7.2); \draw (1,-7.4) -- (2,-7.2); \draw[line width=1mm,white] (1,-7.2) -- (2,-7.4); \draw (1,-7.2) -- (2,-7.4);
\fill[white] (0.9,-2.9) rectangle (2.1,-2.7); \draw (1,-2.9) -- (2,-2.7); \draw[line width=1mm,white] (1,-2.7) -- (2,-2.9); \draw (1,-2.7) -- (2,-2.9);
\fill[white] (0.9,5.8) rectangle (2.1,6.0); \draw (1,5.8) -- (2,5.8); \draw[] (1,6.0) -- (2,6.0);
\fill[white] (0.9,6.8) rectangle (2.1,7.0); \draw (1,6.8) -- (2,6.8); \draw[] (1,7.0) -- (2,7.0);
\fill[white] (1.9,-3.7) rectangle (3.1,-3.5); \draw (2,-3.7) -- (3,-3.5); \draw[line width=1mm,white] (2,-3.5) -- (3,-3.7); \draw (2,-3.5) -- (3,-3.7);
\fill[white] (1.9,0.7) rectangle (3.1,0.9); \draw (2,0.7) -- (3,0.7); \draw[] (2,0.9) -- (3,0.9);
\fill[white] (1.9,6.1) rectangle (3.1,6.3); \draw (2,6.1) -- (3,6.1); \draw[] (2,6.3) -- (3,6.3);
\fill[white] (2.9,-5.5) rectangle (4.1,-5.3); \draw (3,-5.5) -- (4,-5.3); \draw[line width=1mm,white] (3,-5.3) -- (4,-5.5); \draw (3,-5.3) -- (4,-5.5);
\fill[white] (2.9,-3.0) rectangle (4.1,-2.8); \draw (3,-3.0) -- (4,-2.8); \draw[line width=1mm,white] (3,-2.8) -- (4,-3.0); \draw (3,-2.8) -- (4,-3.0);
\fill[white] (2.9,0.4) rectangle (4.1,0.6); \draw (3,0.4) -- (4,0.4); \draw[] (3,0.6) -- (4,0.6);
\fill[white] (2.9,0.95) rectangle (4.1,1.15); \draw (3,0.95) -- (4,1.15); \draw[line width=1mm,white] (3,1.15) -- (4,0.95); \draw (3,1.15) -- (4,0.95);
\fill[white] (2.9,3.2) rectangle (4.1,3.4); \draw (3,3.2) -- (4,3.2); \draw[] (3,3.4) -- (4,3.4);
\fill[white] (3.9,-7.8) rectangle (5.1,-7.6); \draw (4,-7.8) -- (5,-7.8); \draw[] (4,-7.6) -- (5,-7.6);
\fill[white] (3.9,-7.2) rectangle (5.1,-7.0); \draw (4,-7.2) -- (5,-7.2); \draw[] (4,-7.0) -- (5,-7.0);
\fill[white] (3.9,-6.9) rectangle (5.1,-6.7); \draw (4,-6.9) -- (5,-6.9); \draw[] (4,-6.7) -- (5,-6.7);
\fill[white] (3.9,-4.7) rectangle (5.1,-4.5); \draw (4,-4.7) -- (5,-4.7); \draw[] (4,-4.5) -- (5,-4.5);
\fill[white] (3.9,1.7) rectangle (5.1,1.9); \draw (4,1.7) -- (5,1.9); \draw[line width=1mm,white] (4,1.9) -- (5,1.7); \draw (4,1.9) -- (5,1.7);
\fill[white] (3.9,5.8) rectangle (5.1,6.0); \draw (4,5.8) -- (5,6.0); \draw[line width=1mm,white] (4,6.0) -- (5,5.8); \draw (4,6.0) -- (5,5.8);
\fill[white] (4.9,-7.5) rectangle (6.1,-7.3); \draw (5,-7.5) -- (6,-7.3); \draw[line width=1mm,white] (5,-7.3) -- (6,-7.5); \draw (5,-7.3) -- (6,-7.5);
\fill[white] (4.9,-5.9) rectangle (6.1,-5.7); \draw (5,-5.9) -- (6,-5.9); \draw[] (5,-5.7) -- (6,-5.7);
\fill[white] (4.9,-1.9) rectangle (6.1,-1.7); \draw (5,-1.9) -- (6,-1.9); \draw[] (5,-1.7) -- (6,-1.7);
\fill[white] (4.9,-0.4) rectangle (6.1,-0.2); \draw (5,-0.4) -- (6,-0.2); \draw[line width=1mm,white] (5,-0.2) -- (6,-0.4); \draw (5,-0.2) -- (6,-0.4);
\fill[white] (4.9,0.6) rectangle (6.1,0.8); \draw (5,0.6) -- (6,0.8); \draw[line width=1mm,white] (5,0.8) -- (6,0.6); \draw (5,0.8) -- (6,0.6);
\fill[white] (4.9,5.1) rectangle (6.1,5.3); \draw (5,5.1) -- (6,5.3); \draw[line width=1mm,white] (5,5.3) -- (6,5.1); \draw (5,5.3) -- (6,5.1);
\fill[white] (5.9,-4.9) rectangle (7.1,-4.7); \draw (6,-4.9) -- (7,-4.9); \draw[] (6,-4.7) -- (7,-4.7);
\fill[white] (5.9,5.4) rectangle (7.1,5.6); \draw (6,5.4) -- (7,5.4); \draw[] (6,5.6) -- (7,5.6);
\fill[white] (5.9,7.2) rectangle (7.1,7.4); \draw (6,7.2) -- (7,7.2); \draw[] (6,7.4) -- (7,7.4);
\fill[white] (7.9,7.1) rectangle (9.1,7.3); \draw (8,7.1) -- (9,7.3); \draw[line width=1mm,white] (8,7.3) -- (9,7.1); \draw (8,7.3) -- (9,7.1);
\fill[white] (8.9,-6.5) rectangle (10.1,-6.3); \draw (9,-6.5) -- (10,-6.5); \draw[] (9,-6.3) -- (10,-6.3);
\fill[white] (8.9,-6.2) rectangle (10.1,-6.0); \draw (9,-6.2) -- (10,-6.2); \draw[] (9,-6.0) -- (10,-6.0);
\fill[white] (8.9,5.3) rectangle (10.1,5.5); \draw (9,5.3) -- (10,5.5); \draw[line width=1mm,white] (9,5.5) -- (10,5.3); \draw (9,5.5) -- (10,5.3);
\fill[white] (8.9,6.4) rectangle (10.1,6.6); \draw (9,6.4) -- (10,6.6); \draw[line width=1mm,white] (9,6.6) -- (10,6.4); \draw (9,6.6) -- (10,6.4);
\fill[white] (9.9,-6.8) rectangle (11.1,-6.6); \draw (10,-6.8) -- (11,-6.8); \draw[] (10,-6.6) -- (11,-6.6);
\fill[white] (9.9,6.1) rectangle (11.1,6.3); \draw (10,6.1) -- (11,6.3); \draw[line width=1mm,white] (10,6.3) -- (11,6.1); \draw (10,6.3) -- (11,6.1);
\fill[white] (9.9,6.7) rectangle (11.1,6.9); \draw (10,6.7) -- (11,6.7); \draw[] (10,6.9) -- (11,6.9);
\fill[white] (9.9,7.0) rectangle (11.1,7.2); \draw (10,7.0) -- (11,7.2); \draw[line width=1mm,white] (10,7.2) -- (11,7.0); \draw (10,7.2) -- (11,7.0);
\fill[white] (9.9,7.5) rectangle (11.1,7.7); \draw (10,7.5) -- (11,7.7); \draw[line width=1mm,white] (10,7.7) -- (11,7.5); \draw (10,7.7) -- (11,7.5);
\fill[white] (10.9,5.7) rectangle (12.1,5.9); \draw (11,5.7) -- (12,5.9); \draw[line width=1mm,white] (11,5.9) -- (12,5.7); \draw (11,5.9) -- (12,5.7);
\fill[white] (10.9,7.25) rectangle (12.1,7.45); \draw (11,7.25) -- (12,7.45); \draw[line width=1mm,white] (11,7.45) -- (12,7.25); \draw (11,7.45) -- (12,7.25);
\fill[white] (11.9,4.7) rectangle (13.1,4.9); \draw (12,4.7) -- (13,4.9); \draw[line width=1mm,white] (12,4.9) -- (13,4.7); \draw (12,4.9) -- (13,4.7);
\fill[white] (11.9,6.0) rectangle (13.1,6.2); \draw (12,6.0) -- (13,6.0); \draw[] (12,6.2) -- (13,6.2);
\fill[white] (12.9,-7.9) rectangle (14.1,-7.7); \draw (13,-7.9) -- (14,-7.9); \draw[] (13,-7.7) -- (14,-7.7);
\fill[white] (12.9,-7.3) rectangle (14.1,-7.1); \draw (13,-7.3) -- (14,-7.3); \draw[] (13,-7.1) -- (14,-7.1);
\fill[white] (12.9,-6.3) rectangle (14.1,-6.1); \draw (13,-6.3) -- (14,-6.1); \draw[line width=1mm,white] (13,-6.1) -- (14,-6.3); \draw (13,-6.1) -- (14,-6.3);
\fill[white] (12.9,-5.9) rectangle (14.1,-5.7); \draw (13,-5.9) -- (14,-5.9); \draw[] (13,-5.7) -- (14,-5.7);
\fill[white] (12.9,-5.5) rectangle (14.1,-5.3); \draw (13,-5.5) -- (14,-5.3); \draw[line width=1mm,white] (13,-5.3) -- (14,-5.5); \draw (13,-5.3) -- (14,-5.5);
\fill[white] (13.9,4.9) rectangle (15.1,5.1); \draw (14,4.9) -- (15,5.1); \draw[line width=1mm,white] (14,5.1) -- (15,4.9); \draw (14,5.1) -- (15,4.9);
\fill[white] (13.9,6.4) rectangle (15.1,6.6); \draw (14,6.4) -- (15,6.6); \draw[line width=1mm,white] (14,6.6) -- (15,6.4); \draw (14,6.6) -- (15,6.4);
\fill[white] (13.9,6.8) rectangle (15.1,7.0); \draw (14,6.8) -- (15,6.8); \draw[] (14,7.0) -- (15,7.0);
\fill[white] (13.9,7.5) rectangle (15.1,7.7); \draw (14,7.5) -- (15,7.7); \draw[line width=1mm,white] (14,7.7) -- (15,7.5); \draw (14,7.7) -- (15,7.5);
\fill[white] (14.9,-6.8) rectangle (16.1,-6.6); \draw (15,-6.8) -- (16,-6.6); \draw[line width=1mm,white] (15,-6.6) -- (16,-6.8); \draw (15,-6.6) -- (16,-6.8);
\fill[white] (14.9,-1.7) rectangle (16.1,-1.5); \draw (15,-1.7) -- (16,-1.5); \draw[line width=1mm,white] (15,-1.5) -- (16,-1.7); \draw (15,-1.5) -- (16,-1.7);
\fill[white] (14.9,2.5) rectangle (16.1,2.7); \draw (15,2.5) -- (16,2.7); \draw[line width=1mm,white] (15,2.7) -- (16,2.5); \draw (15,2.7) -- (16,2.5);
\fill[white] (14.9,2.9) rectangle (16.1,3.1); \draw (15,2.9) -- (16,3.1); \draw[line width=1mm,white] (15,3.1) -- (16,2.9); \draw (15,3.1) -- (16,2.9);
\fill[white] (15.9,-2.4) rectangle (17.1,-2.2); \draw (16,-2.4) -- (17,-2.4); \draw[] (16,-2.2) -- (17,-2.2);
\fill[white] (15.9,-0.2) rectangle (17.1,-0.0); \draw (16,-0.2) -- (17,-0.2); \draw[] (16,-0.0) -- (17,-0.0);
\fill[white] (15.9,4.8) rectangle (17.1,5.0); \draw (16,4.8) -- (17,5.0); \draw[line width=1mm,white] (16,5.0) -- (17,4.8); \draw (16,5.0) -- (17,4.8);
\fill[white] (15.9,5.2) rectangle (17.1,5.4); \draw (16,5.2) -- (17,5.2); \draw[] (16,5.4) -- (17,5.4);
\fill[white] (15.9,6.1) rectangle (17.1,6.3); \draw (16,6.1) -- (17,6.3); \draw[line width=1mm,white] (16,6.3) -- (17,6.1); \draw (16,6.3) -- (17,6.1);
\fill[white] (16.9,-7.7) rectangle (18.1,-7.5); \draw (17,-7.7) -- (18,-7.5); \draw[line width=1mm,white] (17,-7.5) -- (18,-7.7); \draw (17,-7.5) -- (18,-7.7);
\fill[white] (16.9,-5.5) rectangle (18.1,-5.3); \draw (17,-5.5) -- (18,-5.5); \draw[] (17,-5.3) -- (18,-5.3);
\fill[white] (16.9,2.9) rectangle (18.1,3.1); \draw (17,2.9) -- (18,2.9); \draw[] (17,3.1) -- (18,3.1);
\fill[white] (16.9,4.0) rectangle (18.1,4.2); \draw (17,4.0) -- (18,4.2); \draw[line width=1mm,white] (17,4.2) -- (18,4.0); \draw (17,4.2) -- (18,4.0);
\fill[white] (17.9,-2.4) rectangle (19.1,-2.2); \draw (18,-2.4) -- (19,-2.2); \draw[line width=1mm,white] (18,-2.2) -- (19,-2.4); \draw (18,-2.2) -- (19,-2.4);
\fill[white] (17.9,5.4) rectangle (19.1,5.6); \draw (18,5.4) -- (19,5.4); \draw[] (18,5.6) -- (19,5.6);
\fill[white] (17.9,7.4) rectangle (19.1,7.6); \draw (18,7.4) -- (19,7.4); \draw[] (18,7.6) -- (19,7.6);
\fill[white] (18.9,-0.5) rectangle (20.1,-0.3); \draw (19,-0.5) -- (20,-0.3); \draw[line width=1mm,white] (19,-0.3) -- (20,-0.5); \draw (19,-0.3) -- (20,-0.5);
\fill[white] (18.9,0.1) rectangle (20.1,0.3); \draw (19,0.1) -- (20,0.3); \draw[line width=1mm,white] (19,0.3) -- (20,0.1); \draw (19,0.3) -- (20,0.1);
\fill[white] (18.9,0.5) rectangle (20.1,0.7); \draw (19,0.5) -- (20,0.7); \draw[line width=1mm,white] (19,0.7) -- (20,0.5); \draw (19,0.7) -- (20,0.5);
\fill[white] (18.9,6.4) rectangle (20.1,6.6); \draw (19,6.4) -- (20,6.6); \draw[line width=1mm,white] (19,6.6) -- (20,6.4); \draw (19,6.6) -- (20,6.4);
\end{tikzpicture}
\qquad
\begin{tikzpicture}[xscale=0.375,yscale=0.45]
\foreach \x in {1,...,20}
{
\draw (\x,-8) -- (\x,8);
}
\foreach \x in {7,...,10}
{
\fill (\x,0) circle (1mm);
}
\foreach \x in {11,...,14}
{
\filldraw[fill=white] (\x,0) circle (1mm);
}
\draw[very thick, dashed] (6.5,-4.5) rectangle (14.5,4.5);
\fill[white] (0.9,-7.4) rectangle (2.1,-7.2); \draw (1,-7.4) -- (2,-7.2); \draw[line width=1mm,white] (1,-7.2) -- (2,-7.4); \draw (1,-7.2) -- (2,-7.4);
\fill[white] (0.9,-2.9) rectangle (2.1,-2.7); \draw (1,-2.9) -- (2,-2.7); \draw[line width=1mm,white] (1,-2.7) -- (2,-2.9); \draw (1,-2.7) -- (2,-2.9);
\fill[white] (0.9,5.8) rectangle (2.1,6.0); \draw (1,5.8) -- (2,5.8); \draw[] (1,6.0) -- (2,6.0);
\fill[white] (0.9,6.8) rectangle (2.1,7.0); \draw (1,6.8) -- (2,6.8); \draw[] (1,7.0) -- (2,7.0);
\fill[white] (1.9,-3.7) rectangle (3.1,-3.5); \draw (2,-3.7) -- (3,-3.5); \draw[line width=1mm,white] (2,-3.5) -- (3,-3.7); \draw (2,-3.5) -- (3,-3.7);
\fill[white] (1.9,0.7) rectangle (3.1,0.9); \draw (2,0.7) -- (3,0.7); \draw[] (2,0.9) -- (3,0.9);
\fill[white] (1.9,6.1) rectangle (3.1,6.3); \draw (2,6.1) -- (3,6.1); \draw[] (2,6.3) -- (3,6.3);
\fill[white] (2.9,-5.5) rectangle (4.1,-5.3); \draw (3,-5.5) -- (4,-5.3); \draw[line width=1mm,white] (3,-5.3) -- (4,-5.5); \draw (3,-5.3) -- (4,-5.5);
\fill[white] (2.9,-3.0) rectangle (4.1,-2.8); \draw (3,-3.0) -- (4,-2.8); \draw[line width=1mm,white] (3,-2.8) -- (4,-3.0); \draw (3,-2.8) -- (4,-3.0);
\fill[white] (2.9,0.4) rectangle (4.1,0.6); \draw (3,0.4) -- (4,0.4); \draw[] (3,0.6) -- (4,0.6);
\fill[white] (2.9,0.95) rectangle (4.1,1.15); \draw (3,0.95) -- (4,1.15); \draw[line width=1mm,white] (3,1.15) -- (4,0.95); \draw (3,1.15) -- (4,0.95);
\fill[white] (2.9,3.2) rectangle (4.1,3.4); \draw (3,3.2) -- (4,3.2); \draw[] (3,3.4) -- (4,3.4);
\fill[white] (3.9,-7.8) rectangle (5.1,-7.6); \draw (4,-7.8) -- (5,-7.8); \draw[] (4,-7.6) -- (5,-7.6);
\fill[white] (3.9,-7.2) rectangle (5.1,-7.0); \draw (4,-7.2) -- (5,-7.2); \draw[] (4,-7.0) -- (5,-7.0);
\fill[white] (3.9,-6.9) rectangle (5.1,-6.7); \draw (4,-6.9) -- (5,-6.9); \draw[] (4,-6.7) -- (5,-6.7);
\fill[white] (3.9,-4.7) rectangle (5.1,-4.5); \draw (4,-4.7) -- (5,-4.7); \draw[] (4,-4.5) -- (5,-4.5);
\fill[white] (3.9,1.7) rectangle (5.1,1.9); \draw (4,1.7) -- (5,1.9); \draw[line width=1mm,white] (4,1.9) -- (5,1.7); \draw (4,1.9) -- (5,1.7);
\fill[white] (3.9,5.8) rectangle (5.1,6.0); \draw (4,5.8) -- (5,6.0); \draw[line width=1mm,white] (4,6.0) -- (5,5.8); \draw (4,6.0) -- (5,5.8);
\fill[white] (4.9,-7.5) rectangle (6.1,-7.3); \draw (5,-7.5) -- (6,-7.3); \draw[line width=1mm,white] (5,-7.3) -- (6,-7.5); \draw (5,-7.3) -- (6,-7.5);
\fill[white] (4.9,-5.9) rectangle (6.1,-5.7); \draw (5,-5.9) -- (6,-5.9); \draw[] (5,-5.7) -- (6,-5.7);
\fill[white] (4.9,-1.9) rectangle (6.1,-1.7); \draw (5,-1.9) -- (6,-1.9); \draw[] (5,-1.7) -- (6,-1.7);
\fill[white] (4.9,-0.4) rectangle (6.1,-0.2); \draw (5,-0.4) -- (6,-0.2); \draw[line width=1mm,white] (5,-0.2) -- (6,-0.4); \draw (5,-0.2) -- (6,-0.4);
\fill[white] (4.9,0.6) rectangle (6.1,0.8); \draw (5,0.6) -- (6,0.8); \draw[line width=1mm,white] (5,0.8) -- (6,0.6); \draw (5,0.8) -- (6,0.6);
\fill[white] (4.9,5.1) rectangle (6.1,5.3); \draw (5,5.1) -- (6,5.3); \draw[line width=1mm,white] (5,5.3) -- (6,5.1); \draw (5,5.3) -- (6,5.1);
\fill[white] (5.9,-4.9) rectangle (7.1,-4.7); \draw (6,-4.9) -- (7,-4.9); \draw[] (6,-4.7) -- (7,-4.7);
\fill[white] (5.9,5.4) rectangle (7.1,5.6); \draw (6,5.4) -- (7,5.4); \draw[] (6,5.6) -- (7,5.6);
\fill[white] (5.9,7.2) rectangle (7.1,7.4); \draw (6,7.2) -- (7,7.2); \draw[] (6,7.4) -- (7,7.4);
\fill[white] (7.9,7.1) rectangle (9.1,7.3); \draw (8,7.1) -- (9,7.3); \draw[line width=1mm,white] (8,7.3) -- (9,7.1); \draw (8,7.3) -- (9,7.1);
\fill[white] (8.9,-6.5) rectangle (10.1,-6.3); \draw (9,-6.5) -- (10,-6.5); \draw[] (9,-6.3) -- (10,-6.3);
\fill[white] (8.9,-6.2) rectangle (10.1,-6.0); \draw (9,-6.2) -- (10,-6.2); \draw[] (9,-6.0) -- (10,-6.0);
\fill[white] (8.9,5.3) rectangle (10.1,5.5); \draw (9,5.3) -- (10,5.5); \draw[line width=1mm,white] (9,5.5) -- (10,5.3); \draw (9,5.5) -- (10,5.3);
\fill[white] (8.9,6.4) rectangle (10.1,6.6); \draw (9,6.4) -- (10,6.6); \draw[line width=1mm,white] (9,6.6) -- (10,6.4); \draw (9,6.6) -- (10,6.4);
\fill[white] (9.9,-6.8) rectangle (11.1,-6.6); \draw (10,-6.8) -- (11,-6.8); \draw[] (10,-6.6) -- (11,-6.6);
\fill[white] (9.9,6.1) rectangle (11.1,6.3); \draw (10,6.1) -- (11,6.3); \draw[line width=1mm,white] (10,6.3) -- (11,6.1); \draw (10,6.3) -- (11,6.1);
\fill[white] (9.9,6.7) rectangle (11.1,6.9); \draw (10,6.7) -- (11,6.7); \draw[] (10,6.9) -- (11,6.9);
\fill[white] (9.9,7.0) rectangle (11.1,7.2); \draw (10,7.0) -- (11,7.2); \draw[line width=1mm,white] (10,7.2) -- (11,7.0); \draw (10,7.2) -- (11,7.0);
\fill[white] (9.9,7.5) rectangle (11.1,7.7); \draw (10,7.5) -- (11,7.7); \draw[line width=1mm,white] (10,7.7) -- (11,7.5); \draw (10,7.7) -- (11,7.5);
\fill[white] (10.9,5.7) rectangle (12.1,5.9); \draw (11,5.7) -- (12,5.9); \draw[line width=1mm,white] (11,5.9) -- (12,5.7); \draw (11,5.9) -- (12,5.7);
\fill[white] (10.9,7.25) rectangle (12.1,7.45); \draw (11,7.25) -- (12,7.45); \draw[line width=1mm,white] (11,7.45) -- (12,7.25); \draw (11,7.45) -- (12,7.25);
\fill[white] (11.9,4.7) rectangle (13.1,4.9); \draw (12,4.7) -- (13,4.9); \draw[line width=1mm,white] (12,4.9) -- (13,4.7); \draw (12,4.9) -- (13,4.7);
\fill[white] (11.9,6.0) rectangle (13.1,6.2); \draw (12,6.0) -- (13,6.0); \draw[] (12,6.2) -- (13,6.2);
\fill[white] (12.9,-7.9) rectangle (14.1,-7.7); \draw (13,-7.9) -- (14,-7.9); \draw[] (13,-7.7) -- (14,-7.7);
\fill[white] (12.9,-7.3) rectangle (14.1,-7.1); \draw (13,-7.3) -- (14,-7.3); \draw[] (13,-7.1) -- (14,-7.1);
\fill[white] (12.9,-6.3) rectangle (14.1,-6.1); \draw (13,-6.3) -- (14,-6.1); \draw[line width=1mm,white] (13,-6.1) -- (14,-6.3); \draw (13,-6.1) -- (14,-6.3);
\fill[white] (12.9,-5.9) rectangle (14.1,-5.7); \draw (13,-5.9) -- (14,-5.9); \draw[] (13,-5.7) -- (14,-5.7);
\fill[white] (12.9,-5.5) rectangle (14.1,-5.3); \draw (13,-5.5) -- (14,-5.3); \draw[line width=1mm,white] (13,-5.3) -- (14,-5.5); \draw (13,-5.3) -- (14,-5.5);
\fill[white] (13.9,4.9) rectangle (15.1,5.1); \draw (14,4.9) -- (15,5.1); \draw[line width=1mm,white] (14,5.1) -- (15,4.9); \draw (14,5.1) -- (15,4.9);
\fill[white] (13.9,6.4) rectangle (15.1,6.6); \draw (14,6.4) -- (15,6.6); \draw[line width=1mm,white] (14,6.6) -- (15,6.4); \draw (14,6.6) -- (15,6.4);
\fill[white] (13.9,6.8) rectangle (15.1,7.0); \draw (14,6.8) -- (15,6.8); \draw[] (14,7.0) -- (15,7.0);
\fill[white] (13.9,7.5) rectangle (15.1,7.7); \draw (14,7.5) -- (15,7.7); \draw[line width=1mm,white] (14,7.7) -- (15,7.5); \draw (14,7.7) -- (15,7.5);
\fill[white] (14.9,-6.8) rectangle (16.1,-6.6); \draw (15,-6.8) -- (16,-6.6); \draw[line width=1mm,white] (15,-6.6) -- (16,-6.8); \draw (15,-6.6) -- (16,-6.8);
\fill[white] (14.9,-1.7) rectangle (16.1,-1.5); \draw (15,-1.7) -- (16,-1.5); \draw[line width=1mm,white] (15,-1.5) -- (16,-1.7); \draw (15,-1.5) -- (16,-1.7);
\fill[white] (14.9,2.5) rectangle (16.1,2.7); \draw (15,2.5) -- (16,2.7); \draw[line width=1mm,white] (15,2.7) -- (16,2.5); \draw (15,2.7) -- (16,2.5);
\fill[white] (14.9,2.9) rectangle (16.1,3.1); \draw (15,2.9) -- (16,3.1); \draw[line width=1mm,white] (15,3.1) -- (16,2.9); \draw (15,3.1) -- (16,2.9);
\fill[white] (15.9,-2.4) rectangle (17.1,-2.2); \draw (16,-2.4) -- (17,-2.4); \draw[] (16,-2.2) -- (17,-2.2);
\fill[white] (15.9,-0.2) rectangle (17.1,-0.0); \draw (16,-0.2) -- (17,-0.2); \draw[] (16,-0.0) -- (17,-0.0);
\fill[white] (15.9,4.8) rectangle (17.1,5.0); \draw (16,4.8) -- (17,5.0); \draw[line width=1mm,white] (16,5.0) -- (17,4.8); \draw (16,5.0) -- (17,4.8);
\fill[white] (15.9,5.2) rectangle (17.1,5.4); \draw (16,5.2) -- (17,5.2); \draw[] (16,5.4) -- (17,5.4);
\fill[white] (15.9,6.1) rectangle (17.1,6.3); \draw (16,6.1) -- (17,6.3); \draw[line width=1mm,white] (16,6.3) -- (17,6.1); \draw (16,6.3) -- (17,6.1);
\fill[white] (16.9,-7.7) rectangle (18.1,-7.5); \draw (17,-7.7) -- (18,-7.5); \draw[line width=1mm,white] (17,-7.5) -- (18,-7.7); \draw (17,-7.5) -- (18,-7.7);
\fill[white] (16.9,-5.5) rectangle (18.1,-5.3); \draw (17,-5.5) -- (18,-5.5); \draw[] (17,-5.3) -- (18,-5.3);
\fill[white] (16.9,2.9) rectangle (18.1,3.1); \draw (17,2.9) -- (18,2.9); \draw[] (17,3.1) -- (18,3.1);
\fill[white] (16.9,4.0) rectangle (18.1,4.2); \draw (17,4.0) -- (18,4.2); \draw[line width=1mm,white] (17,4.2) -- (18,4.0); \draw (17,4.2) -- (18,4.0);
\fill[white] (17.9,-2.4) rectangle (19.1,-2.2); \draw (18,-2.4) -- (19,-2.2); \draw[line width=1mm,white] (18,-2.2) -- (19,-2.4); \draw (18,-2.2) -- (19,-2.4);
\fill[white] (17.9,5.4) rectangle (19.1,5.6); \draw (18,5.4) -- (19,5.4); \draw[] (18,5.6) -- (19,5.6);
\fill[white] (17.9,7.4) rectangle (19.1,7.6); \draw (18,7.4) -- (19,7.4); \draw[] (18,7.6) -- (19,7.6);
\fill[white] (18.9,-0.5) rectangle (20.1,-0.3); \draw (19,-0.5) -- (20,-0.3); \draw[line width=1mm,white] (19,-0.3) -- (20,-0.5); \draw (19,-0.3) -- (20,-0.5);
\fill[white] (18.9,0.1) rectangle (20.1,0.3); \draw (19,0.1) -- (20,0.3); \draw[line width=1mm,white] (19,0.3) -- (20,0.1); \draw (19,0.3) -- (20,0.1);
\fill[white] (18.9,0.5) rectangle (20.1,0.7); \draw (19,0.5) -- (20,0.7); \draw[line width=1mm,white] (19,0.7) -- (20,0.5); \draw (19,0.7) -- (20,0.5);
\fill[white] (18.9,6.4) rectangle (20.1,6.6); \draw (19,6.4) -- (20,6.6); \draw[line width=1mm,white] (19,6.6) -- (20,6.4); \draw (19,6.6) -- (20,6.4);
\fill[white] (9.9,0.97) rectangle (11.1,1.17); \draw (10,0.97) -- (11,0.97); \draw (10,1.17) -- (11,1.17);
\fill[white] (8.9,2.00) rectangle (10.1,2.20); \draw (9,2.00) -- (10,2.00); \draw (9,2.20) -- (10,2.20);
\fill[white] (10.9,1.35) rectangle (12.1,1.55); \draw (11,1.35) -- (12,1.35); \draw (11,1.55) -- (12,1.55);
\fill[white] (7.9,2.98) rectangle (9.1,3.18); \draw (8,2.98) -- (9,2.98); \draw (8,3.18) -- (9,3.18);
\fill[white] (9.9,2.76) rectangle (11.1,2.96); \draw (10,2.76) -- (11,2.76); \draw (10,2.96) -- (11,2.96);
\fill[white] (11.9,2.33) rectangle (13.1,2.53); \draw (12,2.33) -- (13,2.33); \draw (12,2.53) -- (13,2.53);
\fill[white] (6.9,3.67) rectangle (8.1,3.87); \draw (7,3.67) -- (8,3.67); \draw (7,3.87) -- (8,3.87);
\fill[white] (8.9,3.69) rectangle (10.1,3.89); \draw (9,3.69) -- (10,3.69); \draw (9,3.89) -- (10,3.89);
\fill[white] (10.9,4.02) rectangle (12.1,4.22); \draw (11,4.02) -- (12,4.02); \draw (11,4.22) -- (12,4.22);
\fill[white] (12.9,4.02) rectangle (14.1,4.22); \draw (13,4.02) -- (14,4.02); \draw (13,4.22) -- (14,4.22);
\fill[semitransparent,gray!75!white] (10,0.25) rectangle (11,1.25) (9,1.25) rectangle (10,2.25) (11,1.25) rectangle (12,2.25) (8,2.25) rectangle (9,3.25) (10,2.25) rectangle (11,3.25) (12,2.25) rectangle (13,3.25)
(7,3.25) rectangle (8,4.25) (9,3.25) rectangle (10,4.25) (11,3.25) rectangle (12,4.25) (13,3.25) rectangle (14,4.25);
\begin{scope}[yscale=-1]
\fill[white] (9.9,0.38) rectangle (11.1,0.58); \draw (10,0.38) -- (11,0.38); \draw (10,0.58) -- (11,0.58);
\fill[white] (8.9,2.03) rectangle (10.1,2.23); \draw (9,2.03) -- (10,2.03); \draw (9,2.23) -- (10,2.23);
\fill[white] (10.9,2.02) rectangle (12.1,2.22); \draw (11,2.02) -- (12,2.02); \draw (11,2.22) -- (12,2.22);
\fill[white] (7.9,2.64) rectangle (9.1,2.84); \draw (8,2.64) -- (9,2.64); \draw (8,2.84) -- (9,2.84);
\fill[white] (9.9,2.36) rectangle (11.1,2.56); \draw (10,2.36) -- (11,2.36); \draw (10,2.56) -- (11,2.56);
\fill[white] (11.9,2.88) rectangle (13.1,3.08); \draw (12,2.88) -- (13,2.88); \draw (12,3.08) -- (13,3.08);
\fill[white] (6.9,3.98) rectangle (8.1,4.18); \draw (7,3.98) -- (8,3.98); \draw (7,4.18) -- (8,4.18);
\fill[white] (8.9,3.28) rectangle (10.1,3.48); \draw (9,3.28) -- (10,3.28); \draw (9,3.48) -- (10,3.48);
\fill[white] (10.9,3.93) rectangle (12.1,4.13); \draw (11,3.93) -- (12,3.93); \draw (11,4.13) -- (12,4.13);
\fill[white] (12.9,4.00) rectangle (14.1,4.20); \draw (13,4.00) -- (14,4.00); \draw (13,4.20) -- (14,4.20);
\fill[semitransparent,gray!75!white] (10,0.25) rectangle (11,1.25) (9,1.25) rectangle (10,2.25) (11,1.25) rectangle (12,2.25) (8,2.25) rectangle (9,3.25) (10,2.25) rectangle (11,3.25) (12,2.25) rectangle (13,3.25)
(7,3.25) rectangle (8,4.25) (9,3.25) rectangle (10,4.25) (11,3.25) rectangle (12,4.25) (13,3.25) rectangle (14,4.25);
\end{scope}
\end{tikzpicture}
\caption{
On the left is an example of an $\omega \in F$, namely, there is a space-time box conditioned to have no arrivals.
On the right is an example of an $\omega \in G$. We condition on certain arrivals of cul-de-sac edges, on in each gray box, arranged in a pattern to enforce a dipole picture, 
creating loops joining  spins on the left (colored black) and those on the right (colored white) for an interval of length $2L$.
\label{fig:FandG}
}
\end{figure}
\begin{itemize}
\item There are no overpass edges, i.e., no arrivals of $\textd \nu_{\bi\bj}^{F}(\omega)$, in $\WW_L$.
\item Let $\{\bi,\bj\}= \{ k-1, k \}$ be an edge with $k \in\{-L+2,\dots, L\}$ and $k$ even. For $t \in \Z$ satisfying $|k| \leq 2t \leq L-1$, f
there is e{x}actly one cul-de-sac edge, i.e., one arrival of $\textd \nu_{\bi\bj}^{AF}(\omega)$ in $(2t, 2t+1)$,
and there is e{x}actly one arrival of $\textd \nu_{\bi\bj}^{AF}(\omega)$ in $(-2t-1, -2t)$
\item Similarly, let $\{\bi,\bj\}= \{ k-1, k \}$ be an edge with $k \in\{-L+2,\dots, L\}$  and $k$ odd. 
For $t \in \Z$ satisfying $|k| \leq 2t-1 \leq L-1$, 
there is e{x}actly one cul-de-sac edge, i.e., one arrival of $\textd \nu_{\bi\bj}^{AF}(\omega)$ in $(2t-1, 2t)$,
and there is e{x}actly one arrival of $\textd \nu_{\bi\bj}^{AF}(\omega)$ in $(-2t, -2t+1)$
\item There are no other arrivals of $\textd \nu_{\bi\bj}^{AF}(\omega)$ in $\WW_L$.
\end{itemize}
In the case where $d>1$ then one merely  and makes this construction in the first coordinate direction for each cross section of $\B_L$.
By spatial independence of arrivals of the Poisson point processes, there are constants $C_i, c_i>0 , \: i=1,2 $ depending only on $u$ so that
\[
C_1 e^{-c_1 L^{d+1}} \leq \frac{\mathbb P(F)}{\mathbb P(G)} \leq C_2 e^{c_2 L^{d+1}}.
\]
Since the event $G$ only adds $O(L^{d+1})$ edges,
\[
\E[2^{|\LL_{\beta}(\omega)|} \mathbf 1\{ F\} ] \leq C_3 e^{c_3 L^{d+1}} \E[2^{|\LL_{\beta}(\omega)|} \mathbf 1\{ G\}].
\]
Finally, it is clear that
\[
\E[2^{|\LL_{\beta}(\omega)|} \mathbf 1\{ G\}] \leq 4^{L^d} \E[|E_{\beta}(\mathbf 1)|]
\]
since the loops induced by the event $G$ support a labeling such that $\sigma_{\bi}(0)=1$ $ \forall \bi \in \mathbb{B}_L$. The factor $4^{L^d}$ appears in the upper bound due to the total number of (admissible and non-admissible) ways of labeling $\B_L \times \{0\}$ with a spin configuration if we do not demand that all spins be $+1$. Collecting the estimates together finishes the proof in the case $\beta\geq L$.

In the case $\beta< L$, running through the first step of the above proof we have
the e{x}istence of constants $C, c$ so that
\[
\E[2^{|\LL_{\beta}(\omega)|}] \leq C\exp(c L^{d} \beta) \E[2^{|\LL_{\beta}(\omega)|} \mathbf 1\{ F_\beta\} ].
\]
where $F_\beta$ is the event that there are no edges in $\{-L+1,\dots,L\} \times \left[-\frac \beta2, \frac \beta2 \right)$. Now, because of the periodic boundary conditions, a configuration of edges $\omega \in F_\beta$ supports all choices of labelings $\tau$ on $\mathbb{B}_L \times \{0\}$, in particular $\tau= \mathbf{1}_L$. Therefore
\[
\E[2^{|\LL_{\beta}(\omega)|} \mathbf 1\{ F_\beta\} ] \leq 4^{L^d} \E[|E_{\beta}(\mathbf 1)|].
\]
Gathering the estimates together, the lemma is proved in this case.
\end{proof}
For more general configurations $\tau$ than $\mathbf{1}_L$, one merely reflects $-\tau$ and connects $\tau$ on $\mathbb{B}_L$ to $-\tau$ on a reflection of $\mathbb{B}_L$
across one of the faces. It does not matter which dimension one chooses to reflect in if $d>1$.
In fact, it seems reasonable that one can make an even better construction in dimensions $d>1$ (reflecting something more like the true higher dimensional analogue of the dipole picture)
so that one need not assume $L\leq \frac{1}{2}N$, but instead just that $L^d \leq \frac{1}{2} N^d$.
But we will not pursue this here.
\subsubsection{The Case $\Delta \notin [-1, 1]$ and $\beta< 4L$.}
\label{S:func}
The above argument does not work for $\Delta \notin [-1, 1]$. Nevertheless, if $\beta< 4L$ we can use a more direct functional analysis argument, working with the operators in the expression
\begin{equation}
\label{E:EFP1}
\EFP_{L}(N,\beta)= \frac{\textrm{Tr}\left( \prod_{\bi \in \mathbb{B}_L} \left[\frac{1}{2} + S_{\bi}^z\right] e^{-\beta H'_{N,\Delta}} \right)
}{\textrm{ Tr} (e^{-\beta H'_{N,\Delta}})}.
\end{equation}
where
\begin{equation}
\label{E:Ham}
-H'_{N, \Delta}\, =\, -b(\Lambda)+ H_{N,\Delta}\, ,
\end{equation}
where $b(\Lambda)=\frac{1}{4} |\mathcal{E}(\T_N)|$, so that $\operatorname{Tr}[\exp(-\beta H'_{N,\Delta})] = \E[2^{|\LL_{\beta}(\omega)|}]$
 by \cref{EPar}. 
\begin{lemma}
For any $\Delta \in \R$, there are constants $c, C>0$ so that the following holds. If $|\beta| \leq 4L$,
\[
\EFP_{L}(N,\beta) \geq Ce^{-c L^d \beta}.
\]
\end{lemma}
\begin{proof}
The simple and brutal idea is to compare to the system in which interactions between $\mathbb{T}_N\backslash \mathbb{B}_L$ and $\mathbb{B}_L$ have been turned off. To this end, let us introduce the interpolating Hamiltonians
\[
-H'_{N,\Delta}(a)= -H'_{\mathbb{T}_N \backslash \mathbb{B}_L,\Delta} - H'_{\mathbb{B}_L,\Delta} 
- a \underbrace{\sum_{\substack{\{ \bi\bj\} \in E(\mathbb{T}_N)\\ \bi \in \mathbb{B}_L,\, \bj \in \mathbb{T}_N \backslash \mathbb{B}_L}} 
S^x_{\bi}S^x_{\bj} +S^y_{\bi}S^y_{\bj}+ \Delta S^z_{\bi}S^z_{\bj}- \frac{|\partial \mathbb{B}_L|}{4}}_{\mathcal I(\mathbb{T}_N\backslash \mathbb{B}_L, \mathbb{B}_L)}
\]
for $a \in [0, 1]$. Here $\partial \mathbb{B}_L$ denotes the edge boundary of $\mathbb{B}_L$ -- the collection of edges with exactly one endpoint in $\mathbb{B}_L$.
Further set
\[
Z(a) = \textrm{ Tr} (e^{-\beta H_{N,\Delta}'(
a)}).
\]
Then $|\log Z(1)- \log Z(0)| = \left|\int_0^1 \frac{\textd}{\textd a} \log Z(a) \right|$
and
\[
\frac{\textd}{\textd a} \log Z(a)= \frac{\beta \textrm{Tr} (\mathcal I(\mathbb{T}_N\backslash \mathbb{B}_L, \mathbb{B}_L) e^{-\beta H'_{N,\Delta}})}{Z(a)}.
\]
Because the summands of $\mathcal I(\mathbb{T}_N\backslash \mathbb{B}_L, \mathbb{B}_L)$ are uniformly bounded the $\ell^1$-$\ell^\infty$ H\"{o}lder inequality for operators implies
$\frac{\textd}{\textd a} \log Z(a) \leq c\beta L^{d-1}$ and thus
\[
e^{-c \beta L^{d-1}} \leq \frac{Z(1)}{Z(0)} \leq e^{c \beta L^{d-1}}.
\]
To handle the numerator in \cref{E:EFP1}, we use a modified version of the graphical representation summarized in \Cref{S:graphical} (see \Cref{S:UB} for more details since it is also important there). For $\Delta$ fixed, we may write $-H'_{N,\Delta}= -H'_{N,{1}} + \sum_{\{\bi,\bj\} \in E(\mathbb{T}_N)}
(\Delta -1)S^z_{\bi} S^z_{\bj}$. The term $\sum_{\{\bi,\bj\}}
(\Delta -1)S^z_{\bi} S^z_{\bj}$ acts as a potential over the configuration space of labeled graphs $\Sigma_{N,\beta}^{\mathrm{per}}(\omega)$ determined by the graphical representation for $-H'_{N,{1}}$. 
(This is much the same as the outcome of the Feynman-Kac expansion for Schr\"{o}dinger operators).
Using this representation we have:
there is a constant $C_0$ depending only on $u, \Lambda_N$ and $\beta$ so that
\begin{align}
& Z_{N,\Delta}(\beta)= C_0 \E_{1}\left[ \sum_{\sigma \in \Sigma^{\mathrm{per}}_{N,\beta}(\omega)} e^{\frac{\Delta-1}{4} V_\beta(\sigma)} \right], \\
\label{E:Pot}
& Z_{N,\Delta}(\beta) \EFP_L(N,\beta) = C_0 \E_{1} \left[ \sum_{\sigma \in E_{N,L,\beta}(\mathbf 1_L)} e^{\frac{\Delta-1}{4} V_\beta(\sigma)} \right]
\end{align}
where $\E_{1}[\cdot]$ denotes the expectation associated with respect to the Poisson process for the graphical representation of $-H'_{N,{1}}$ and 
$-V_\beta(\sigma)= \int_{-\beta/2}^{\beta/2} \textrm{d} t 
\sum_{\{\bi,\bj\}} \sigma_{\bi}(t) \sigma_{\bj}(t)$.
We use, in particular, \cref{E:Pot}. Let $F_L$ be the event that there are no bonds connecting $\mathbb{B}_L$ to $\mathbb{T}_N\backslash \mathbb{B}_L$. Then, of course,
\[
\E_{1} \left[ \sum_{\sigma \in E_{N,L,\beta}(\mathbf 1_L)} e^{\frac{\Delta-1}{4} V_\beta(\sigma)} \right]
\geq \E_{1} \left[ \mathbf 1\{F_L\}\sum_{\sigma \in E_{N,L,\beta}(\mathbf 1_L)} e^{\frac{\Delta-1}{4} V_\beta(\sigma)} \right].
\]
Combined with the bound above on the ratio of partition functions this implies (for a second pair of constants $c_1, C_1$)
\begin{align*}
\EFP_L(N,\beta)\, 
&\geq\, e^{-c \beta L^{d-1}} \frac{\E_{1} \left[ \mathbf 1\{F_L\} \sum_{\sigma \in E_{N,L,\beta}(\mathbf 1_L)} e^{\frac{\Delta-1}{4} V_\beta(\sigma)}\right]}{Z(0)}\\ 
&\geq\, C_1 e^{-c_1 \beta L^{d-1}} 
\frac{\textrm{ Tr} \left( \prod_{\bi \in \mathbb{B}_L} \left[\frac 12 + S_{\bi}^z\right] \exp\left(-\beta H'_{\mathbb{B}_L,\Delta}\right)\right)}{\textrm{ Tr} (\exp\left(-\beta H'_{\mathbb{B}_L,\Delta}\right))}.
\end{align*}
Part of the reason for the factor $e^{-c_1 \beta L^{d-1}}$ is the potential terms spanning the boundary of $\partial \mathbb{B}_L$.
Since the numerator on the righthand side of this string of inequalities is
\[
\textrm{ Tr} (\exp\left(-\beta H'_{\mathbb{B}_L,\Delta}\right))\, =\, \langle  \Psi_{\mathbb{B}_L}(\mathbf{1}_L), e^{-\beta H'_{\mathbb{B}_L,\Delta}}  \Psi_{\mathbb{B}_L}(\mathbf{1}_L) \rangle\, ,
\]
and the denominator on the righthand side of this string of inequalities is
\[
\sum_{\sigma \in\Sigma_{\mathbb{B}_L}} \langle \Psi_{\mathbb{B}_L}(\sigma), \exp\left(-\beta H'_{\mathbb{B}_L,\Delta}\right)\Psi_{\mathbb{B}_L}(\sigma)\rangle\, ,
\]
we have that
\[
\frac{\textrm{ Tr} \left[ \prod_{\bi \in \mathbb{B}_L} \left[\frac 12 + S_{\bi}^z\right] \exp\left(-\beta H'_{\mathbb{B}_L,\Delta}\right)\right]}
{\textrm{ Tr} \left[\exp\left(-\beta H'_{\mathbb{B}_L,\Delta}\right)\right]} \geq C_2 \exp^{-c_2 \beta L^d},
\]
and the lemma follows.
\end{proof}
\section{The Upper Bound \cref{E:Ud2} in \Cref{thm:Main1}}
\label{S:UB}

The idea of the upper bound is that the projector $\mathbf{Q}_L$ imposes higher energy relative to the Hamiltonian $H_N$
than in the ground state.
We do not have an exact formula for the ground state or the ground state energy.
But we can obtain variational upper bounds on the ground state energy.

Aside from this idea, we use two main tools for this part of the argument.
The first main tool is a graphical representation for the equilibrium expectation.
This is the one introduced in \Cref{S:func}, and used for example in equation \cref{E:Pot}.
\begin{lemma}
\label{lem:FK}
For any graph $\mathcal{G} =  (\V,\E)$, and for any configurations $\sigma,\tau \in \Sigma_{\V}$,
\begin{equation}
\label{eq:Eform2}
\langle \Psi_{\V}(\tau), e^{-\beta H_{\G,\Delta}} \Psi_{\V}(\sigma)\rangle\,
=\, e^{\beta |\mathcal{E}|} \E_1\Bigg[\sum_{\sigma(\cdot) \in \Sigma_{\mathcal{V},\beta}(\omega)} \mathbf{1}_{\{\sigma(-\beta/2)=\sigma\}} \mathbf{1}_{\{\sigma(\beta/2)=\tau\}}
\exp\bigg(\frac{\Delta-1}{4}\, \int_{-\beta/2}^{\beta/2} U_{\G}(\sigma(t))\, dt\bigg)\Bigg]\, ,
\end{equation}
where $\E_1$ is the expectation defined just before \Cref{prop:Ueltschi}, and where $U_{\G} : \Sigma_{\V} \to \R$
is defined as the usual ferromagnetic Ising energy
$U_{\G}(\sigma) = -\sum_{\{\bi,\bj\} \in \E} \s_{\bi} \s_{\bj}$.
\end{lemma}
This is the same formula we have used before
in equation \cref{E:Pot} in \Cref{prop:Ueltschi}.
We will give a proof, shortly, for completeness.

The reason \Cref{lem:FK} is useful is that for $\Delta<1$, the exponential factor is actually
the Gibbs weight factor for the antiferromagnetic Ising model.
The projector $\mathbf{Q}_L$ is onto the ground states of the ferromagnetic Ising model
on the block $\B_L$.
Therefore, on this block, relative to the potential energy $U_{\G}$,
the configuration that $\mathbf{Q}_L$ imposes has the highest possible energy,
which means that its Gibbs probability is low.

This still leaves the difficulty that the block $\B_L$ is not as large as $\T_N$.
Moreover, the effect of imposing $\mathbf{Q}_L$ at one time does not last for all time $\beta$
because the graphical representation inherent in $\E_1$ does allow the energy of $U_{\G}$
to change.
The second main tool is reflection positivity and chessboard estimates which allows us to disseminate
the event imposed by $\mathbf{Q}_L$ in space, as well as the generalized 
H\"older's inequality to re-impose the event after a period $\delta T$ of time.
These allow us to disseminate the event to overcome this difficulty.

But the argument is still involved at this point.
Every Hamiltonian satisfies the generalized H\"older's inequality which allows to disseminate in time.
But only certain Hamiltonians satisfy reflection positivity.
For the XXZ model, this requires $\Delta\leq 0$.
Moreover, the dissemination in space is antiferromagnetic in nature.
The hopping allowed by the process related to $\E_1$
allows mixing of two adjacent blocks, one of which is of type $\uparrow$/$+$ and the other of which is $\downarrow$/$-$.
But this is essentially a boundary effect if $\delta T$ is a sufficiently small fraction of the linear size of the block $L$.
We can show this by analysis of the Poisson process related to $\E_1$.

With these guiding principles, we will now enter the details of the proof.
Before doing that, let us -quickly prove Lemma \ref{lem:FK}.

\begin{proof}
This can be proved by the Trotter product formula or by taking derivatives. It is a typical Feynman-Kac formula for perturbing 
the generator of a Markov process by a potential. 
Defining the operator $A_{\Delta,\beta}$ such that
$\langle \Psi_{\V}(\tau), A_{\Delta,\beta} \Psi_{\V}(\s)\rangle$ gives the right-hand-side of \cref{eq:Eform2}, we already know from \cref{eq:Eform} that 
$\frac{d}{d\beta} A_{1,\beta} = -H_{\G,1} A_{1,\beta}$.
So, by the rules of differentiation $\frac{d}{d\beta} \langle \Psi_{\V}(\tau), A_{\beta,\Delta} \Psi_{\V}(\s)\rangle$
equals $\langle \Psi_{\V}(\tau), [-H_{\G,1}]A_{\beta,\Delta} \Psi_{\V}(\s)\rangle$ plus $-\frac{1}{4}\, (\Delta-1) U(\tau) \langle \Psi_{\V}(\tau), A_{\beta,\Delta} \Psi_{\V}(\s)\rangle$.
Then the lemma follows because
$H_{\G,\Delta} - H_{\G,1} = -(\Delta-1)\sum_{\{\bi,\bj\}\in\mathcal{E}}S_{\bi}^z S_{\bj}^z$.
I.e., $(H_{\G,\Delta}-H_{\G,1})\Psi_{\V}(\tau) = -\frac{1}{4}\, (\Delta-1) U_{\G}(\tau) \Psi_{\V}(\tau)$.
\end{proof}

\subsection{Reduction to Estimating the Cost of a ``Universal Contour''}


Recall that we always assume $N$ is even. 
%
Recall that the vertex set of $\T_N$ is 
the same as the vertex set $\B_N$ \Cref{eq:BNdef}, even though the edge set for $\T_N$
contains extra edges beyond those in $\B_N$ \Cref{eq:TNdef}.
Let us define the configuration $\tau^{(N,L)}$ and the rank-1 projection $\widehat{\mathbf{Q}}_{N,L}$, as follows:
\begin{equation}
\label{eq:QLNdef}
\tau^{(L,N)}_{\bi}\, =\, (-1)^{\lfloor{(2i_1-1)/(2L)}\rfloor+\dots+\lfloor{(2i_d-1)/(2L)}\rfloor}\, ,\quad \text{ and }\
\widehat{\mathbf{Q}}_{N,L}\, 
=\, \ket{\Psi_N(\tau^{(L,N)})} \bra{\Psi_N(\tau^{(L,N)})}\, .
\end{equation}
The operator $\widehat{\mathbf{Q}}_{N,L}$ is the ``universal contour.''
A schematic picture of it is
given in \Cref{fig:unicon}. 
Then, using {\em reflection positivity}, we may prove the following lemma.

\begin{lemma}
\label{lem:ChessLoss}
Suppose that $\Delta\leq 0$. Let $L$ be fixed, satisfying $L\leq N/2$.
Then
\begin{equation}
\label{ineq:QhatQ}
\EFP_L(N,\beta)\, =\, 
\langle \mathbf{Q}_{L}\rangle_{N, \Delta,\beta}\,
\leq\, \left(\langle \widehat{\mathbf{Q}}_{N,L} \rangle_{N, \Delta,\beta}\right)^{1/K}\, ,
\end{equation}
where $K = 2^{d (\log_2(N/L)+1)}$.
\end{lemma}
The proof of this result is relegated to the appendix as it is a standard application of the chessboard estimates method in \cite{FrohLieb}. For the readers convenience, we give a short account of reflection positivity for quantum spin systems there; see \Cref{S:RP}.
\begin{figure}
\centerline{
\begin{tikzpicture}[xscale=0.75, yscale=0.75, very thick]
\foreach \x in {-3.25,-3,...,3.5}
{\foreach \y in {-3.25,-3,...,3.5}
{
\fill (\x,\y) circle (0.2mm);}}
\foreach \x in {-0.25,0,0.25,0.5}
{\foreach \y in {-0.25,0,0.25,0.5}
{
\fill[white] (\x,\y) circle (0.72mm);
\filldraw[thick,fill=gray] (\x,\y) circle (0.8mm);
}}
\end{tikzpicture}
\hspace{1.5cm}
\begin{tikzpicture}[xscale=0.75, yscale=0.75, very thick]
\foreach \x in {-3.25,-3,...,3.5}
{\foreach \y in {-3.25,-3,...,3.5}
{
\fill (\x,\y) circle (0.2mm);}}
\foreach \x in {-1.75,-1.5,-1.25,-1,0.25,0.5,0.75,1,2.25,2.5,2.75,3}
{\foreach \y in {-1.75,-1.5,-1.25,-1,0.25,0.5,0.75,1,2.25,2.5,2.75,3}
{
\fill[white] (\x,\y) circle (0.72mm);
\filldraw[thick,fill=gray] (\x,\y) circle (0.8mm);
}}
\begin{scope}[xshift=-1cm]
\foreach \x in {-1.75,-1.5,-1.25,-1,0.25,0.5,0.75,1,2.25,2.5,2.75,3}
{\foreach \y in {-1.75,-1.5,-1.25,-1,0.25,0.5,0.75,1,2.25,2.5,2.75,3}
{
\fill[white] (\x,\y) circle (0.72mm);
\filldraw[gray!50!black,thick,fill=white] (\x,\y) circle (0.8mm);
}}
\end{scope}
\begin{scope}[yshift=-1cm]
\foreach \x in {-1.75,-1.5,-1.25,-1,0.25,0.5,0.75,1,2.25,2.5,2.75,3}
{\foreach \y in {-1.75,-1.5,-1.25,-1,0.25,0.5,0.75,1,2.25,2.5,2.75,3}
{
\fill[white] (\x,\y) circle (0.72mm);
\filldraw[gray!50!black,thick,fill=white] (\x,\y) circle (0.8mm);
}}
\end{scope}
\begin{scope}[xshift=-1cm,yshift=-1cm]
\foreach \x in {-1.75,-1.5,-1.25,-1,0.25,0.5,0.75,1,2.25,2.5,2.75,3}
{\foreach \y in {-1.75,-1.5,-1.25,-1,0.25,0.5,0.75,1,2.25,2.5,2.75,3}
{
\fill[white] (\x,\y) circle (0.72mm);
\filldraw[thick,fill=gray] (\x,\y) circle (0.8mm);
}}
\end{scope}
\foreach \x in {3.25,3.5}
{\foreach \y in {-2.75,-2.5,-2.25,-2,-0.75,-0.5,-0.25,0,1.25,1.5,1.75,2,3.25,3.5}
{
\fill[white] (\x,\y) circle (0.72mm);
\filldraw[thick,fill=gray] (\x,\y) circle (0.8mm);
\fill[white,xshift=0.25cm,yshift=0.25cm] (-\x,-\y) circle (0.72mm);
\filldraw[thick,fill=gray,xshift=0.25cm,yshift=0.25cm] (-\x,-\y) circle (0.8mm);
\fill[white,yshift=0.25cm] (\x,-\y) circle (0.72mm);
\draw[thick,gray!50!black,yshift=0.25cm] (\x,-\y) circle (0.8mm);
\fill[white,xshift=0.25cm] (-\x,\y) circle (0.72mm);
\draw[thick,gray!50!black,xshift=0.25cm] (-\x,\y) circle (0.8mm);
}}
\foreach \y in {3.25,3.5}
{\foreach \x in {-2.75,-2.5,-2.25,-2,-0.75,-0.5,-0.25,0,1.25,1.5,1.75,2,3.25,3.5}
{
\fill[white] (\x,\y) circle (0.72mm);
\filldraw[thick,fill=gray] (\x,\y) circle (0.8mm);
\fill[white,xshift=0.25cm,yshift=0.25cm] (-\x,-\y) circle (0.72mm);
\filldraw[thick,fill=gray,xshift=0.25cm,yshift=0.25cm] (-\x,-\y) circle (0.8mm);
\fill[white,yshift=0.25cm] (\x,-\y) circle (0.72mm);
\draw[thick,gray!50!black,yshift=0.25cm] (\x,-\y) circle (0.8mm);
\fill[white,xshift=0.25cm] (-\x,\y) circle (0.72mm);
\draw[thick,gray!50!black,xshift=0.25cm] (-\x,\y) circle (0.8mm);
}}
\end{tikzpicture}
}
\caption{On the left is a schematic view of $\mathbf{Q}_L$, for $d=2$ dimensions, when $L=4$ and $N=28$.
In the right the ``universal contour,'' $\widehat{\mathbf{Q}}_{L,N}$ is shown.
A black circle depicts a $+$ spin projector, and a white circle depicts a $-$ spin projector.
A small dot represents a site without a projector.
\label{fig:unicon}
}
\end{figure}
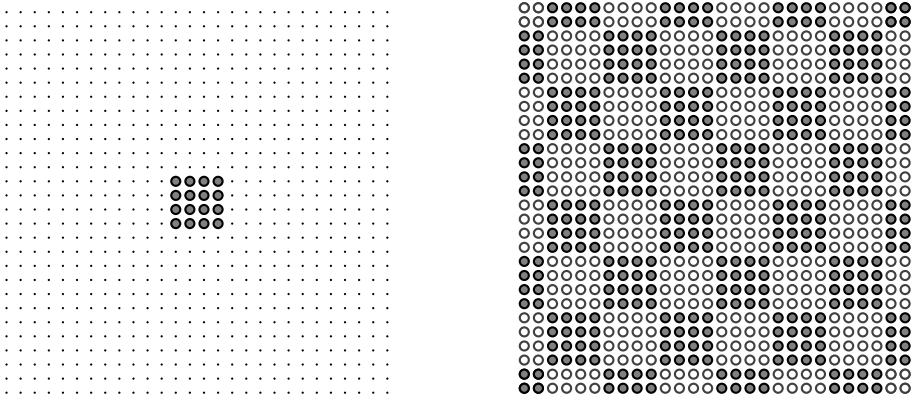

Given \Cref{lem:ChessLoss} our task now is to bound $\langle \widehat{\mathbf{Q}}_{N, L} \rangle_{N,\Delta,\beta}$ from above. 
We note that \Cref{lem:ChessLoss} applies for $\Delta\leq 0$ so that an upper bound on $\langle \widehat{\mathbf{Q}}_{N, L} \rangle_{N,\Delta,\beta}$
leads to an upper bound on $\langle \mathbf{Q}_{L} \rangle_{N,\Delta,\beta}$ for $\Delta \leq 0$.
But we may actually obtain an upper bound on $\langle \widehat{\mathbf{Q}}_{N, L} \rangle_{N,\Delta,\beta}$
for a wider range of $\Delta$'s.
We demonstrate this since the argument is the same.
\begin{theorem}
\label{thm:Block}
For each $\Delta<1$, there exists an $L_0 \in \{1,2,\dots\}$ and $C, c>0$ such that for all $L\geq L_0$
and all $N\geq 32 d L$, we have
$$
\langle \widehat{\mathbf{Q}}_{N,L} \rangle_{N, \Delta,\beta}\, \leq\, C e^{-c N^d \min(L, \beta)}\, .
$$
\end{theorem}
Combined with \Cref{lem:ChessLoss}, this theorem implies \cref{E:Ud2} of \Cref{thm:Main1}.
Therefore, proving this theorem is our main goal for the rest of this section.

We note that \Cref{lem:ChessLoss} has allowed us to disseminate the projector $\mathbf{Q}_L$, which projected just on spins inside $\B_L$,
to the projector $\widehat{\mathbf{Q}}_{N,L}$ which restricts to a specified spin configuration on all of $\T_N$.
This is what we described in the outline at the beginning of this section.
We have also obtained the desired graphical representation in \Cref{lem:FK}.
The next step is to disseminate in time.

\subsection{Generalized H\"older's inequality}
\label{sec:Block}

\begin{proposition} For any positive integer $n$,
\begin{equation}
\label{E:quantime}
\langle \widehat{\mathbf{Q}}_{N, L} \rangle_{N, \Delta,\beta}\,
\leq\, \left(\frac{\tr[(\widehat{\mathbf{Q}}_{N,L} e^{-\beta H/(2n)})^{2n}]}{Z_{N,\Delta}(\beta)}\right)^{1/(2n)}\, .
\end{equation}
\end{proposition}
This is a direct consequence of the generalized H\"{o}lder inequality for operators with $A = \widehat{\mathbf{Q}}_{N ,L}$ and this choice of $n$, see Theorem \ref{thm:Holder} of \Cref{S:App}.
In that appendix we will also describe how this proposition follows.

Because we are re-imposing the rank-1 projection
$\widehat{\mathbf{Q}}_{N, L}$ every $\frac{\beta}{2n}$ units of ``time,'' we further have
$$
\tr[(\widehat{\mathbf{Q}}_{N, L} e^{-\beta H/(2n)})^{2n}]\,
=\, \left(\tr[\widehat{\mathbf{Q}}_{N, L} e^{-\beta H/(2n)}]\right)^{2n}\, 
=\, \left(\langle \Psi_{N}(\tau^{(N,L)}), e^{-\beta H/(2n)} \Psi_{N}(\tau^{(N,L)}) \rangle\right)^{2n}\, .
$$
Therefore, \cref{E:quantime} becomes
\begin{equation}
\label{eq:QtimeDecoupled}
\langle \widehat{\mathbf{Q}}_{N, L} \rangle_{N, \Delta, \beta}\,
\leq\, \frac{\tr[\widehat{\mathbf{Q}}_{N, L} e^{-\beta H/(2n)}]}{\left(Z_{N,\Delta}(\beta)\right)^{1/(2n)}}\, .
\end{equation}
Recall that $\mathcal{E}(\mathbb{T}_N)$ denotes the edge-set for $\mathbb{T}_N$.
Then we may define
\begin{align}
\label{eq:sfNum}
\mathsf{Num}_{\beta,N,L,n}\,
&=\, e^{-(\beta/(2n)) (|\mathcal{E}(\mathbb{T}_N)|/4)} \tr[\widehat{\mathbf{Q}}_{N,L} e^{-\beta H/(2n)}]\, ,\\
\label{eq:sfDen}
\mathsf{Den}_{\beta,N,L,n}\,
&=\, e^{-(\beta/(2n)) (|\mathcal{E}(\mathbb{T}_N)|/4)} \left(Z_{N,\Delta}(\beta)\right)^{1/(2n)}\, .
\end{align}
We will frequently write these as just $\mathsf{Num}$ and $\mathsf{Den}$.
By \cref{eq:QtimeDecoupled} we have the following:
\begin{corollary}
\label{cor:QtDec}
For any choice of $n$
\begin{equation}
\label{eq:QtimeDecoupled2}
\langle \widehat{\mathbf{Q}}_{N, L} \rangle_{N,\Delta,\beta}\,
\leq\, \frac{\mathsf{Num}_{\beta,N,L,n}}{\mathsf{Den}_{\beta,N,L,n}}\, .
\end{equation}
\end{corollary}
To obtain an upper bound on the left hand side of \cref{eq:QtimeDecoupled2} we need to obtain an upper bound on $\mathsf{Num}$
and a lower bound on $\mathsf{Den}$.

The point of multplying by $e^{-(\beta/(2n)) (|\mathcal{E}(\mathbb{T}_N)|/4)}$ is to cancel the multiplier in \cref{eq:Eform2}.
This is particularly useful in obtaining a variational lower bound on $\mathsf{Den}$. We do this next.

\subsection{Variational lower bound on $\mathsf{Den}$}

Optimally, we would calculate $\mathsf{Den}$ exactly.
But this is difficult. It involves calculating the partition function.
But even the ground state energy is difficult.
However, we may make a variational calculation to obtain a bound.

This is a standard exercise in quantum formalism, and Jensen's inequality. (It could also be deduced easily from basic results in statistical mechanics such as 
the Gibbs variational principle for the free energy and equilibrium state.) 
\begin{lemma}
For any graph $\G = (\V,\Ed)$, and any $\Delta \in \R$, $\beta\geq0$, we have
\begin{equation}
\label{ineq:ZlowerBd}
Z_{\G,\Delta}(\beta) e^{-\beta |\Ed|/4}\, =\, \operatorname{Tr}\big[e^{-\beta (H_{\G,\Delta}+(1/4)|\Ed|)}\big]\geq\, 1\, .
\end{equation}
\end{lemma}
\label{lem:GenLB}
\begin{proof}
We define a unit vector
\begin{equation}
\label{eq:PhiVDef}
\Phi_{\V}\, =\, 2^{-|\V|/2} \sum_{\tau \in \Sigma_{\V}}\Psi_{\V}(\tau)\, .
\end{equation}
We note that $H_{\G,1}+(1/4)|\Ed|$ equals the sum over all $\{i,j\} \in \Ed$ of $\frac{1}{2} (1 - T_{\bi,\bj})$,
where $T_{\bi,\bj} \Psi_{\V}(\tau) = \Psi_{\V}(\tau^{(\bi,\bj)})$, where $\tau^{(\bi,\bj)}$ is the configuration
obtained from $\tau$
by interchanging $\tau^{(\bi,\bj)}_{\bi} = \tau_{\bj}$ and $\tau^{(\bi,\bj)}_{\bj} = \tau_{\bi}$.
But in \cref{eq:PhiVDef}, we sum over all $\tau$'s, uniformly. Therefore $T_{\bi,\bj} \Phi_{\V} = \Phi_{\V}$ for all $\{\bi,\bj\}$.
So 
\begin{equation}
\label{eq:Delta1Ann}
\Big(H_{\G,1}+(1/4)|\Ed|\Big) \Phi_{\V}\, =\, 0\, .
\end{equation}
We also know from the proof of \Cref{lem:FK} that $(H_{\G,\Delta}-H_{\G,1})\Psi_{\V}(\tau) = -\frac{1}{4}(\Delta-1)U_{\G}(\tau)\Psi_{\V}(\tau)$.
So
$$
\langle \Phi_{\V}, (H_{\G,\Delta}-H_{\G,1}) \Phi_{\V}\rangle\,
=\,-\left(\frac{\Delta-1}{4}\right)  2^{-|\V|} \sum_{\tau \in \Sigma_{V}} U_{\G}(\tau)\, .
$$
But the uniform average of $U_{\G}(\tau)$ is the sum over $\{\bi,\bj\} \in \Ed$ of the uniform average of $\s_{\bi} \s_{\bj}$,
and for $\bi\neq \bj$ the random variable $\s_{\bi}$ is independent of the random variable $\s_{\bj}$ in the uniform probability measure.
Moreover both $\s_{\bi}$ and $\s_{\bj}$ have expectation zero in the uniform probability measure.
So $\langle \Phi_{\V}, (H_{\G,\Delta}-H_{\G,1}) \Phi_{\V}\rangle=0$.
Combining this with \cref{eq:Delta1Ann}, we conclude that $\langle \Phi_{\V},(H_{\G,\Delta}+|\Ed|)\Phi_{\V}\rangle=0$.
Finally, \begin{align*}
\tr\left[e^{-\beta(H_{N,\Delta}+(1/4)|\mathcal{E}(\mathbb{T}_N)|)}\right]\,
&\geq\, \langle \Psi,  e^{-\beta(H_{N,\Delta}+(1/4)|\mathcal{E}(\mathbb{T}_N)|)} \Psi\rangle\\
&\geq\, \exp\left(-\beta \langle \Psi,  (H_{N,\Delta}+(1/4)|\mathcal{E}(\mathbb{T}_N)|) \Psi\rangle\right)\, ,
\end{align*}
which completes the proof.
\end{proof}
In particular, taking the $(2n)$th root, we obtain the following.
\begin{corollary}
\label{cor:DenLB}
We have the bound
\begin{equation}
\label{eq:sfDenLB}
\mathsf{Den}\, 
=\, \Big(e^{-\beta |\mathcal{E}(\mathbb{T}_N)|/4} Z_{N,\Delta}(\beta)\Big)^{1/(2n)}\, 
\geq\, 1.
\end{equation}
\end{corollary}

\subsection{Upper bound on $\mathsf{Num}$ using large deviation bounds}

Using \Cref{lem:FK} and the definition in \cref{eq:sfNum} (and \cref{eq:QLNdef}), we can rewrite
\begin{equation}
\label{eq:NumFK}
\mathsf{Num}_{\beta,N,L,n}\,
=\, e^{\beta |\mathcal{E}|} \E_1\Bigg[\sum_{\sigma(\cdot) \in \Sigma_{\mathcal{V},\beta}(\omega)} \mathbf{1}_{\{\sigma(-\beta/2)=\sigma(\beta/2)=\tau^{(N,L)}\}}
\exp\bigg(\frac{\Delta-1}{4}\, \int_{-\beta/2}^{\beta/2} U_{\G}(\sigma(t))\, dt\bigg)\Bigg]\, ,
\end{equation}
where we write $U_N$ as a short hand notation for $U_{\mathbb{T}_N}$, as usual.
Note that when there are no cul-de-sac edges, specifying $\sigma(-\beta/2)$ completely specifies $\sigma(t)$ for all $t\geq -\beta/2$.
Therefore the indicators are not restricting any multplicity of choices of $\sigma(\cdot)$ as much as they are putting restrictions on $\omega$.

We will prove the following.
\begin{proposition} 
\label{prop:NumBd1}
For any $\Delta<1$ and $L\geq 24$
\begin{equation}
\label{ineq:NumBd1}
\mathsf{Num}_{\beta,N,L,n}\, \leq\, e^{-(1/64)(1-\Delta) dN^d \delta T} + e^{[(1/4)(1-\Delta) - (M \ln M - M + 1)] dN^d \delta T}\, ,
\end{equation}
where $\delta T = \beta/(2n)$ and $M = L/(1536 d^2 \delta T)$.
\end{proposition}
The idea is relatively straightforward.
For sufficiently large $L$, we have $U_N(\tau^{(N,L)}) \approx |\Ed(\T_N)|$.
This is because most edges $\bi$, $\bj$ have $\tau^{(N,L)}_{\bi} = \tau^{(N,L)}_{\bj}$.
This only fails if $\bi$ and $\bj$ span two adjoining blocks. See \Cref{fig:unicon}.
The fraction of those edges is $1/L$.
If $\Delta<1$ then $|\Ed(\T_N)|$ is interpreted as the maximum possible energy of the antiferromagnetic
Ising potential, instead of the ground state energy of the ferromagnetic Ising potential.
Therefore, in \cref{eq:NumFK} this leads to exponential suppression.

That reasoning works to bound $U_N(\sigma(t))$ for $t = \pm\beta/(4n)$.
But for $t$ between $-\beta/(4n)$ and $\beta/(4n)$, one needs to account for the hopping
of spins allowed by the stochastic process associated to $\E_1$.
Esssentially, the stochastic process allows spins to hop between neighboring blocks.
But in a short period of time, there should not be too many arrivals of $\textd \nu_{\bi\bj}^F(\omega)$.
Therefore, with high probability, the spin configuration $\sigma(t)$ will still have relatively low energy for the ferromagnetic Ising potential,
and hence will be exponentially suppressed in the antiferromagnetic Ising Gibbs state.

In the complementary event, $\textd \nu_{\bi\bj}^F(\omega)$ may allow enough overpass edges to attain the ground state
of the Ising antiferromagnet. Therefore, the un-normalized Gibbs weight exponential term in \cref{eq:NumFK}
may be exponentially large.
But the large deviation bound for this rare event is nonlinear and dominates this. 

\begin{proof}
We may decompose $\T_N$ into ``blocks.''
Let us define a ``block'' to be a maximal connected set of sites $\bi \in \T_N$
satsifying the condition that $\tau^{(N,L)}_{\bi} = \tau^{(N,L)}_{\bj}$ for all pairs of points $\{\bi,\bj\}$ in the block.
Let us define the edges on the faces of blocks 
$$
\mathcal{F}\, =\, \{\{\bi,\bj\} \in \Ed(\T_N)\, :\, \tau^{(N,L)}_{\bi} \neq \tau^{(N,L)}_{\bj}\}\, .
$$
This means that the edge $\{\bi,\bj\}$ spans two adjacent blocks.

Some blocks are ``full'' having size $L^d$. For each coordinate direction, there may also be partial blocks at a distance less than $L$
from the two faces of $\B_N$ in the $+$ and $-$ side of that coordinate direction.
Considering this, it is easy to see that 
$$
|\mathcal{F}|\, \leq\, (L^{-1} + 2N^{-1}) |\Ed(\T_N)|\, .
$$
Since $N$ is always at least $L$, this implies $|\mathcal{F}|\leq 3L^{-1} |\Ed(\T_N)|$.
Therefore
$$
|\Ed(\T_N)| - U_N(\tau^{(N,L)})\, \leq\,  2 |\mathcal{F}|\, \leq\, 6 L^{-1} |\Ed(\T_N)|\, .
$$
At times $t=\pm \beta/(4n)$ the only edges $\{\bi,\bj\}$ with $\s_{\bi}(t)=-\s_{\bj}(t)$ are the ones
in $\mathcal{F}$.
We will show that with high probability at times between $-\beta/(4n)$ and $\beta/(4n)$
most antiferromagnetic edges are close to $\mathcal{F}$.
Because of this we note the following easy bound on the number of vertices at a short distance from $\mathcal{F}$.

For any positive integer $r$, let us define $\mathcal{V}_r$ to be the set of all sites $\bi$ satisfying this condition:
$\bi$ is in a block $\Lambda \subset \T_N$, and has distance less than or equal to $r$ from $\T_N \setminus \Lambda$.
So for instance 
$$
\mathcal{V}_1\, =\, \{\bi \in \T_N\, :\, \exists \bj \in \T_N \text{ such that } \{\bi,\bj\} \in \mathcal{F}\}\, .
$$
Because $\mathcal{F}$ can be written as a disjoint union of coordinate ``planes'' in $\T_N$, it is easy to deduce from this that
the formula above generalizes in the following way for distances $r\geq 1$:
\begin{equation}
\label{ineq:Vr}
|\mathcal{V}_r|\, \leq\, 2r|\mathcal{F}|\, \leq\, 6 r L^{-1} |\Ed(\T_N)|\, =\, 6 d r L^{-1} N^d\, .
\end{equation}
With this easy bound done, we proceed with the remainder of the proof. 

For each time $t \in [-\beta/(4n),\beta/(4n)]$
define $X(t)$ to be the number of sites $\bi$ such that $\tau^{(N,L)}_{\bi} \neq \sigma_{\bi}(t)$.
Then we have a lower bound
$$
|\Ed|-6L^{-1}|\Ed|-U_N(\sigma(t))\, \leq\, 4d X(t)\, .
$$
Let $\mathcal{A}$ be the event that $\textd \nu_{\bi\bj}^F(\omega)$ is such that for all times $t \in [-\beta/(2n),\beta/(2n)]$ we have $X(t)\leq N^d/8$.
On this event we have
\begin{equation}
\label{ineq:UpEn}
\mathbf{1}_{\{\sigma(-\beta/(4n)) = \sigma(\beta/(4n))=\tau^{(N,L)}\}}
\exp\left(\frac{\Delta-1}{4}\, \int_{-\beta/(4n)}^{\beta/(4n)} U_{N}(\sigma(t))\, dt\right)\,
\leq\, e^{-(1/8)(1-\Delta)(1-12L^{-1})|\Ed|}\, .
\end{equation}
This is a good inequality for us, for the purpose of the proof.
So now we turn our attention to bounding the probability of $\mathcal{A}^{c}$.

Suppose that at some time $t$ we have that $\bi$ is a site in $\T_N \setminus \mathcal{V}_r$
such that $\s_{\bi}(t) \neq \tau^{(N,L)}_{\bi}$.
Then $\bi$ is in a block and $\s_{\bi}(t)$ is opposite to the spin of the block.
Thus in time $[-\beta/(4n),t]$ there must be a path from some neighboring block via overpass edge arrivals of $\textd \nu_{\bi\bj}^F(\omega)$
to bring this oppositely oriented spin in to site $\bi$, and in time $[t,\beta/(4n)]$ there is another path.
Since $\bi$ is at distance more than $r$ from $\mathcal{F}$, this means that there are at least $2r$ arrivals of $\textd \nu_{\bi\bj}^F(\omega)$
associated to these two paths.

Because of this we may see that there are at least $r (X(t)-|\mathcal{V}_r|)$ arrivals of $\textd \nu_{\bi\bj}^F(\omega)$ in the interval $[-\beta/(4n),\beta/(4n)]$.
We divided by 2 since a given arrival of $\textd \nu_{\bi\bj}^F(\omega)$ could contribute to two different paths for two different vertices $\bi$, $\bj$ (since an edge
has two endpoints).
On $\mathcal{A}^c$, we have that $X(t)>N^d/8$ for some time.
So on this event, choosing $r=L/(96d)$, we get that there are at least $LN^d/(1536d)$ arrivals of $\textd \nu_{\bi\bj}^F(\omega)$  in the time interval $[-\beta/(4n),\beta/(4n)]$.
We have used \cref{ineq:Vr}.

Now we recall the large deviation tail bound for a Poisson random variable $\mathcal{N}$:
$$
\mathbf{P}(\mathcal{N} \geq M \mathbf{E} \mathcal{N})\, \leq\,
e^{-(M \ln M -M + 1)\mathbf{E} \mathcal{N} }\, ,
$$
for each $M\geq 1$.
For 
$\textd \nu_{\bi\bj}^F(\omega)$  in the time interval $[-\beta/(4n),\beta/(4n)]$, the total expectation of all the arrivals is $dN^d \delta T$, where
we write $\delta T$ for the length of the time interval, $\delta T = \beta/(2n)$.
Therefore
$$
\Pr_1(\mathcal{A}^c)\, \leq\, e^{-(M\ln M-M+1) d N^d \delta T}\, ,
$$
where $M = L/(1536d^2 \delta T)$.
Combining this with \cref{ineq:UpEn} and the uniform upper bound
$$
\mathbf{1}_{\{\sigma(-\beta/(4n)) = \sigma(\beta/(4n))=\tau^{(N,L)}\}}
\exp\left(\frac{\Delta-1}{4}\, \int_{-\beta/(4n)}^{\beta/(4n)} U_{N}(\sigma(t))\, dt\right)\,
\leq\, e^{(1/4)(1-\Delta)|\E|\delta T}\,
$$
gives the result.
\end{proof}

\subsection{Completion of the proof of \Cref{thm:Block}}

We combine \Cref{cor:QtDec} with \Cref{cor:DenLB} and \Cref{prop:NumBd1}.
Choose any fixed $\epsilon>0$ such that $M=L/(1536 d^2 \epsilon \min(\beta,L))$ is large enough that
$M\ln M - M +1 \geq (17/64)(1-\Delta)$.
Then, taking $n = \lfloor \epsilon \min(\beta,L) \rfloor$, we get that
$$
\mathsf{Num}\, \leq\, 2 e^{-(1/64)(1-\Delta)dN^d \lfloor \epsilon \min(\beta,L)\rfloor}\, .
$$
So taking $c = (1/64)(1-\Delta)d \epsilon$ basically gives the result.
Or, choosing $\beta,L$ sufficiently large that $\epsilon \min(\beta,L)\geq 2$ we can remove the effect of the floor function
by taking $c$ to be half that previous value.

\section{One dimensional results using the six vertex model}
\label{sec:1d6vtx}
In one dimension, on the 1d torus $\mathcal{G} = \T_N$,
the XXZ model of \Cref{eq:PFready} for $\Delta<1$ has a ground state
which may be understood in terms of the six-vertex model.

In this section, we will give one description of the six-vertex model, and also explain
the relation to the ground state of the XXZ model.
These are well-known results and may be found in \cite{Lieb}.
We give a brief review here for the benefit of the reader.

The underlying lattice for the six vertex model is the usual square lattice on a discrete torus of horizontal side length $N$ and vertical
sidelength which we will call $T$. When we need to, we will refer to the torus as $\mathbb{T}_{N,T}$.
Both $N$ and $T$ will be assumed even. 
More precisely, the configuration space for the six-vertex model is the edge set of the torus.

We will introduce some extra notation for the graph $\mathbb{T}_{N,T}$ because the spins
live on the edges.
Let us use a special notation for the vertices of $\mathbb{T}_{N,T}$.
The vertices will be denotes as $V_{ij}$ for $i\in \{1,\dots,N\}$ and $j\in \{1,\dots,T\}$.
Let $\mathscr{E}_{N,T}$ denote the set of edges of $\mathbb{T}_{N,T}$, both horizontal and vertical: 
\begin{equation}
\label{eq:Escr}
\mathscr{E}_{N,T}\, =\, \{E^{h}_{i,j}\, :\, i=1,\dots,N\, ,\ j=1,\dots,T\}
\cup\{E^{v}_{i,j}\, :\, i=1,\dots,N\, ,\ j=1,\dots,T\}\, .
\end{equation}
The edges incident to $V_{ij}$ are $E^h_{ij}$, $E^v_{ij}$,
$E^h_{i-1,j}$ and $E^v_{i,j-1}$, where we identify $i+N\cong i$ and $j+T\cong j$ on the torus.
In \Cref{fig:First6vtx}, in the picture on the left, we show an example when $N=T=6$ along with a few labelled edges and a labelled vertex.

The six-vertex configurations are assignments of spins to $\mathscr{E}_{N,T}$
satisfying certain conditions.
We will interpret a $+$ ($-$) spin on a horizontal edge to point right (left).
As usual, a $+$ ($-$) spin on a vertical edge will point up (down).
Let $\mathcal{S}_{N,T}$ denote the set of spin configurations $\varsigma : \mathscr{E}_{N,T} \to \{+1,-1\}$
satisfying the following six-vertex conditions: for each $i \in \{1,\dots,N\}$ and $j \in \{1,\dots,T\}$,
considering the spins on the edges incident to $V_{ij}$, we have
\begin{equation}
\label{eq:sxCnd}
\varsigma(E^{h}_{i-1,j}) - \varsigma(E^{h}_{ij}) + \varsigma(E^v_{i,j-1}) - \varsigma(E^v_{ij})\,
=\, 0\, .
\end{equation}
These are the valid six-vertex configurations.
The conditions mean that at each vertex there are two spins in and two spins out.
Since there are four edges, the number of such choices is $4$-choose-$2$, which
gives rise to the name ``six'' vertex configuration.
\newcommand{\uptri}{+(0,0.1) -- +(-0.075,-0.1) -- +(0.075,-0.1) -- +(0,0.1)}
\newcommand{\dntri}{+(0,-0.1) -- +(-0.075,0.1) -- +(0.075,0.1) -- +(0,-0.1)}
\newcommand{\ritri}{+(0.1,0) -- +(-0.1,-0.075) -- +(-0.1,0.075)}
\newcommand{\letri}{+(-0.1,0) -- +(0.1,-0.075) -- +(0.1,0.075)}

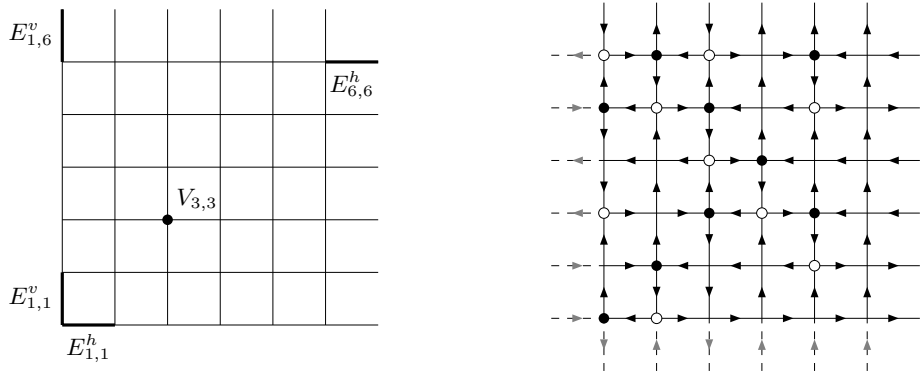
\begin{figure}[]
\begin{center}
\begin{tikzpicture}[xscale=0.7,yscale=0.7,very thin]
\draw (0,0) -- (6,0) (0,1) -- (6,1) (0,2) -- (6,2) (0,3) -- (6,3) (0,4) -- (6,4) (0,5) -- (6,5);
\draw (0,0) -- (0,6) (1,0) -- (1,6) (2,0) -- (2,6) (3,0) -- (3,6) (4,0) -- (4,6) (5,0) -- (5,6);
\draw[very thick] (0,0) -- (1,0) (5,5) -- (6,5) (0,5) -- (0,6) (0,0) -- (0,1);
\fill (2,2) circle (1mm);
\draw (0.5,0) node[below] {\small ${E}^h_{1,1}$};
\draw (5.5,5) node[below] {\small ${E}^h_{6,6}$};
\draw (0,5.5) node[left] {\small ${E}^v_{1,6}$};
\draw (0,0.5) node[left] {\small ${E}^v_{1,1}$};
\draw (2,2) node[above right] {\small ${V}_{3,3}$};
\end{tikzpicture}
\hspace{2cm}
\begin{tikzpicture}[xscale=0.7,yscale=0.7,very thin]
\draw (0,0) -- (6,0) (0,1) -- (6,1) (0,2) -- (6,2) (0,3) -- (6,3) (0,4) -- (6,4) (0,5) -- (6,5);
\draw (0,0) -- (0,6) (1,0) -- (1,6) (2,0) -- (2,6) (3,0) -- (3,6) (4,0) -- (4,6) (5,0) -- (5,6);
\draw[dashed,xshift=-7cm] (6,0) -- (7,0) (6,1) -- (7,1) (6,2) -- (7,2) (6,3) -- (7,3) (6,4) -- (7,4) (6,5) -- (7,5);
\draw[dashed,yshift=-7cm] (0,6) -- (0,7) (1,6) -- (1,7) (2,6) -- (2,7) (3,6) -- (3,7) (4,6) -- (4,7) (5,6) -- (5,7);
\fill[gray] (0,-0.5) \dntri (1,-0.5) \uptri (2,-0.5) \dntri (3,-0.5) \uptri (4,-0.5) \uptri (5,-0.5) \uptri;
\fill (0.5,0) \letri (1.5,0) \ritri (2.5,0) \ritri (3.5,0) \ritri (4.5,0) \ritri (5.5,0) \ritri;
\fill[gray] (-0.5,0) \ritri;
\fill (0,0) circle (1mm);
\filldraw[fill=white] (1,0) circle (1mm);
\fill (0,0.5) \uptri (1,0.5) \dntri (2,0.5) \dntri (3,0.5) \uptri (4,0.5) \uptri (5,0.5) \uptri;
\fill (0.5,1) \ritri (1.5,1) \letri (2.5,1) \letri (3.5,1) \letri (4.5,1) \ritri (5.5,1) \ritri;
\fill[gray] (-0.5,1) \ritri;
\fill (0,1.5) \uptri (1,1.5) \uptri (2,1.5) \dntri (3,1.5) \uptri (4,1.5) \dntri (5,1.5) \uptri;
\fill (1,1) circle (1mm);
\filldraw[fill=white] (4,1) circle (1mm);
\fill (0.5,2) \ritri (1.5,2) \ritri (2.5,2) \letri (3.5,2) \ritri (4.5,2) \letri (5.5,2) \letri;
\fill[gray] (-0.5,2) \letri;
\fill (0,2.5) \dntri (1,2.5) \uptri (2,2.5) \uptri (3,2.5) \dntri (4,2.5) \uptri (5,2.5) \uptri;
\filldraw[fill=white] (0,2) circle (1mm) (3,2) circle (1mm);
\fill (2,2) circle (1mm) (4,2) circle (1mm);
\fill (0.5,3) \letri (1.5,3) \letri (2.5,3) \ritri (3.5,3) \letri (4.5,3) \letri (5.5,3) \letri;
\fill[gray] (-0.5,3) \letri;
\fill (0,3.5) \dntri (1,3.5) \uptri (2,3.5) \dntri (3,3.5) \uptri (4,3.5) \uptri (5,3.5) \uptri;
\filldraw[fill=white] (2,3) circle (1mm);
\fill (3,3) circle (1mm);
\fill (0.5,4) \letri (1.5,4) \ritri (2.5,4) \letri (3.5,4) \letri (4.5,4) \ritri (5.5,4) \ritri;
\fill[gray] (-0.5,4) \ritri;
\fill (0,4.5) \uptri (1,4.5) \dntri (2,4.5) \uptri (3,4.5) \uptri (4,4.5) \dntri (5,4.5) \uptri;
\fill (0,4) circle (1mm) (2,4) circle (1mm);
\filldraw[fill=white] (1,4) circle (1mm) (4,4) circle (1mm);
\fill (0.5,5) \ritri (1.5,5) \letri (2.5,5) \ritri (3.5,5) \ritri (4.5,5) \letri (5.5,5) \letri;
\fill[gray] (-0.5,5) \letri;
\fill (0,5.5) \dntri (1,5.5) \uptri (2,5.5) \dntri (3,5.5) \uptri (4,5.5) \uptri (5,5.5) \uptri;
\filldraw[fill=white] (0,5) circle (1mm) (2,5) circle (1mm);
\fill (1,5) circle (1mm) (4,5) circle (1mm);
\end{tikzpicture}
\caption{
\label{fig:First6vtx}
On the left we indicate the variable names for labelling vertices and edges, which will be useful for the six vertex representation. On the right we have
an example of a valid six-vertex configuration.
}
\end{center}
\end{figure}

In \Cref{fig:First6vtx}, in the picture on the right we give an example of a valid six-vertex configuration of arrows: 
satisfying two-in, two-out at every vertex.
The dashed edges are edges which
are actually repeated on the other side of the torus due to periodic boundary conditions.
Similarly, the grey arrows are arrows which are repeated from the other side.

For the example of a six-vertex configuration in \Cref{fig:First6vtx}, we give some extra decorations,
to help to see certain features.
For each vertex $V_{ij}$ we call it a ``sink'' if 
$$
\varsigma(E^h_{i-1,j})\,
=\, \varsigma(E^v_{ij})\, 
=\, -\varsigma(E^h_{i,j})\,
=\, -\varsigma(E^v_{i,j-1})\,
=\, 1\, .
$$
In the picture on the right in \Cref{fig:First6vtx} the sinks are denoted by black (filled) circles.
A source is a vertex $V_{ij}$ such that
$\varsigma(E^h_{i-1,j})
= \varsigma(E^v_{ij})
= -\varsigma(E^h_{i,j})
= -\varsigma(E^v_{i,j-1}) = -1$.
Sources are indicated in \Cref{fig:First6vtx} by white (open) circles.

In the six vertex model there is also a weight for each $\varsigma \in \mathcal{S}_{N,T}$.
For each $i\in\{1,\dots,N\}$, 
$j \in \{1,\dots,T\}$ let us define
\begin{equation}
\label{eq:mij}
m_{ij}(\varsigma)\, =\, \mathbf{1}\{\varsigma(E^h_{i-1,j})=\varsigma(E^v_{ij})=-\varsigma(E^h_{i,j})
=-\varsigma(E^v_{i,j-1})\}\, ,
\end{equation}
which is the indicator of having either a sink or a source at $V_{ij}$.
Given a parameter $\kappa \in \R$, which will play the role of an inverse-temperature,
the weight of any $\varsigma \in \mathcal{S}_{N,T}$ is 
\begin{equation}
\label{eq:WtDf}
W_{N,T}^\kappa(\varsigma)\, 
=\, \prod_{i=1}^{N} \prod_{j=1}^{T} w_{ij}^{\kappa}(\varsigma)\, ,\qquad
w_{ij}^\kappa(\varsigma)\, =\, e^{\kappa m_{ij}(\varsigma)}\, .
\end{equation}
The relation between the XXZ model and the six-vertex model was originally noted by Lieb,
for example in \cite{Lieb}. The row-to-row transfer matrix of the six-vertex model may be viewed
as an operator on $\Hil_N$, the Hilbert space for the circle, $\T_N$ for $d=1$.
This operator has the same eigenvectors as $H^{\Delta}_N$ if $K$ and $\Delta$ satisfy a certain relation.

Let us be more precise. Given the parameter $K$, the row-to-row transfer operator for the six vertex model may be defined as 
$A_{N,\kappa} : \Hil_N \to \Hil_N$ (where again $\Hil_N$ is quantum spin Hilbert space for the graph $\mathcal{G}=\T_N$ with $d=1$).
It is defined as follows.
Let $\Sigma_N$ denote $\Sigma_{\mathcal{V}}$ for $\mathcal{V} = \T_N$.
For any $\sigma, \sigma' \in \Sigma_N$   define
\begin{equation}
\label{eq:RowToRow}
\langle \Psi_N(\sigma), A_{N,\kappa} \Psi_N(\sigma') \rangle\,
=\, \sum_{\tau \in \Sigma_N}
\prod_{i=1}^{N} \Big(\mathbf{1}\{\tau_{i-1}=\tau_{i})\}
\mathbf{1}\{\sigma'_i=\sigma_i\}
+ e^{\kappa}  \mathbf{1}\{\tau_{i-1}=\sigma'_i=-\tau_i=-\sigma_i\}\Big)\, .
\end{equation}
where $N+1\equiv 1$ because of periodic boundary conditions, as usual.
This represents the product of indicators for \cref{eq:sxCnd} for $i=1,\dots,N$ for a fixed $j$,
and the product of the weights \cref{eq:WtDf}, if we define a partial configuration
only on edges incident to the vertices $V_{ij}$ for $i=1,\dots,N$ and a fixed $j$,
where $\varsigma(E^v_{i,j-1})=\s_i$, $\varsigma(E^v_{i,j})=\s'_i$ and $\varsigma(E^h_{i,j})=\tau_i$.
See \Cref{fig:Row} for an example.
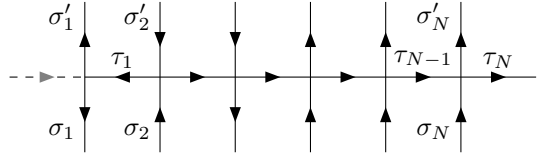
\begin{figure}
\begin{center}
\begin{tikzpicture}[xscale=1,yscale=1]
\draw[dashed] (-1,0) -- (0,0);
\draw (0,0) -- (6,0);
\foreach \x in {0,1,...,5}
{\draw (\x,-1) -- (\x,1);}
\fill (0,-0.5) \dntri (1,-0.5) \uptri (2,-0.5) \dntri (3,-0.5) \uptri (4,-0.5) \uptri (5,-0.5) \uptri;
\fill (0.5,0) \letri (1.5,0) \ritri (2.5,0) \ritri (3.5,0) \ritri (4.5,0) \ritri (5.5,0) \ritri;
\fill[gray] (-0.5,0) \ritri;
\fill (0,0.5) \uptri (1,0.5) \dntri (2,0.5) \dntri (3,0.5) \uptri (4,0.5) \uptri (5,0.5) \uptri;
\draw (0,-0.5) node[below left] {$\sigma_1$};
\draw (0,0.5) node[above left] {$\sigma_1'$};
\draw (1,-0.5) node[below left] {$\sigma_2$};
\draw (1,0.5) node[above left] {$\sigma_2'$};
\draw (0.5,0) node[above] {$\tau_1$};
\draw (5,-0.5) node[below left] {$\sigma_N$};
\draw (5,0.5) node[above left] {$\sigma_N'$};
\draw (4.5,0) node[above] {$\tau_{N-1}$};
\draw (5.5,0) node[above] {$\tau_N$};
\end{tikzpicture}
\end{center}
\caption{\label{fig:Row}
An example of the spins $\sigma$, $\sigma'$ and $\tau$ involved in \cref{eq:RowToRow}.}
\end{figure}
Therefore, for example, defining the partition function for the six vertex model on $\T_{N,T}$
with inverse temperature $\kappa$, we have
\begin{equation}
\label{eq:Ztr}
Z_{N,T}(\kappa)\, =\, \sum_{\varsigma \in \mathcal{S}_{N,T}} W_{N,T}^{\kappa}(\varsigma)\,
=\, \operatorname{Tr}[(A_{N,\kappa})^T]\, .
\end{equation}
The following is to Sutherland, see \cite{Sutherland}:
\begin{proposition}
\label{prop:Sutherland}
For each $N$, the operators $A_{N,\kappa}$ and $H_{N,\Delta}$ (for $d=1$) commute if 
\begin{equation}
\label{eq:kappaDelta}
\Delta\, =\, 1 - \frac{1}{2}\, e^{2\kappa}\, .
\end{equation}
\end{proposition}
A simple but crucial fact is that $A_{N,\kappa}$ commutes with $\mathbf{M}_M$ for each
$M \in \{-\frac{1}{2}\, N,\dots,\frac{1}{2}\, N\}$, just like $H_{N,\Delta}$.
One way to see this is the following  lemma which we anyway need later to prove the upper bounds of \Cref{E:Ld3}.
\begin{lemma}
\label{lem:alternate}
Suppose that $\varsigma \in \mathcal{S}_{N,T}$ is any valid six-vertex configuration.
For any $j \in \{1,\dots,T\}$, the sinks and sources on the vertices $V_{ij}$, $i=1,\dots,N$
must alternate. 
Similarly, for any $i \in \{1,\dots,N\}$, the sinks and sources on the vertices $V_{ij}$,
$j=1,\dots,T$ must alternate.
\end{lemma}
\begin{proof}
Suppose that $V_{ij}$ is a source. Then on the edges $E^h_{ij}$ and $E^{h}_{i-1,j}$,
the arrows are flowing away from $V_{ij}$. This means $\varsigma(E^h_{ij})=+1$
and $\varsigma(E^h_{ij})$.
Similarly, we will have $\varsigma(E^h_{i+1,j})=+1$, and so on, until we come to a vertex
$V_{i+k,j}$ which is a sink.
Similarly $\varsigma(E^h_{i-2,j})=-1$ and so on until we come to a vertex $V_{i-k,j}$ which is a sink.
If we start with a source there is a symmetric argument.
This shows that sinks and sources must alternate in rows. The argument for columns is exactly
symmetric.
\end{proof}
\begin{corollary}
\label{cor:equal}
For any $\varsigma \in \mathcal{S}_{N,T}$, looking at any row or column, 
the number of sinks must equal the number of sources.
\end{corollary}
\begin{proof}
This follows from Lemma \ref{lem:alternate} and periodic boundary conditions.
\end{proof}
\begin{corollary}
For any $N$ and any $M \in \{-\frac{1}{2}\, N,\dots, \frac{1}{2}\, N\}$, the operator $A_{N,\kappa}$
commutes with $\mathbf{M}_M$.
\end{corollary}
\begin{proof}
Suppose that $\sigma, \sigma' \in \Sigma_N$ are two spin configurations and $\Psi_{\sigma}$ is in the 
range of $\mathbf{M}_M$.
Suppose that there is some $\tau \in \Sigma_N$ such that 
$$
\prod_{i=1}^{N} \Big(\mathbf{1}\{\tau_{i-1}=\tau_{i})\}
\mathbf{1}\{\sigma'_i=\sigma_i\}
+ e^{\kappa}  \mathbf{1}\{\tau_{i-1}=\sigma'_i=-\tau_i=-\sigma_i\}\Big)\, 
>\, 0\, .
$$
Then note that the sinks and sources of $\tau$ must be equal.
For each sink at a position $i$ we have $\sigma_i = - \sigma'_i=-1$.
If there is a source, then $\sigma_i=-\sigma'_i=1$.
Otherwise $\sigma_i=\sigma_i'$.
Thus, since the number of sinks equals the number of sources
by \Cref{cor:equal}, $\sum_{i=1}^{N} \sigma_i=\sum_{i=1}^N \sigma'_i$.
\end{proof}
The next result appears in \cite{Lieb} as a consequence of Bethe Ansatz computations  
and the Perron-Frobenius theorem.
\begin{proposition}
\label{prop:Lieb}
Suppose \cref{eq:kappaDelta} is satisfied.  Then
for any $N$ and
for any $M \in \{-\frac{1}{2}\, N,\dots,\frac{1}{2}\, N\}$ consider the restiction of  $H_{N,\Delta}$ to $\Ran(\mathbf{M}_M)$.  The eigenspace with minimal eigenvalue is one
dimensional and the unit eigenvector of this space,  $\Psi^{(\Delta)}_{N,M}$, can be chosen so that  
\[
\langle \Psi^{(\Delta)}_{N,M},\Psi_{N}(\sigma) \rangle >  0\quad \forall \Psi_{N}(\sigma)\in \Ran(\mathbf{M}_M).
\]
Further, if we  restrict $A_{\Delta,\kappa}$ to $\Ran(\mathbf{M}_M)$, its eigenspace of largest eigenvalue is spanned by the same $\Psi^{(\Delta)}_{N,M}$.\end{proposition}

We will discuss both of the results \Cref{prop:Sutherland} and \Cref{prop:Lieb} further in \Cref{sec:6vtx}.
For us, the following is the most important implication.
\begin{corollary}
Given $N$ and $M \in \{-\frac{1}{2}\, N,\dots,\frac{1}{2}\, N\}$, 
suppose $\sigma \in \Sigma_N$ is such that $\Psi_N(\sigma)$  is in the range of $\mathbf{M}_M$.
Then for any $\Delta$ and $\kappa$ satisfying \cref{eq:kappaDelta}
\begin{equation}
\label{eq:components}
|\langle \Psi_{N,M}^{(\Delta)}, \Psi_N(\sigma) \rangle|^2\,
=\, \lim_{T \to \infty} 
\frac{\sum_{\varsigma \in \mathcal{S}_{N,T}} W_{N,T}^{\kappa}(\varsigma)
\mathbf{1}\{\forall i \in \{1,\dots,N\}\, ,\ \varsigma(E^v_{i,1})=\s_i\}}
{\sum_{\varsigma \in \mathcal{S}_{N,T}} W_{N,T}^{\kappa}(\varsigma)
\mathbf{1}\{\sum_{i=1}^N\varsigma(E^v_{i,1})=2M\}}\, .
\end{equation}
\end{corollary}
\begin{proof}
Let $a_{N,M,\kappa}$ be the eigenvalue such that 
$A_{N,\kappa} \Psi^{(\Delta)}_{N,M} = a_{N,M,\kappa} \Psi^{(\Delta)}_{N,M}$.
Let $\Psi_{\alpha}$ and $a_{\alpha}$, for $\alpha=1,\dots,2^{N}-1$ be the remaining orthonormal
eigenvectors and eigenvalues for the real symmetric operator $A_{N,\kappa}$.
Then
\begin{align*}
\sum_{\varsigma \in \mathcal{S}_{N,T}} W_{N,T}^{\kappa}(\varsigma)
\mathbf{1}\{\forall i \in \{1,\dots,N\}\, ,\ \varsigma(E^v_{i,1})=\s_i\}\,
&=\, a_{N,M,\kappa}^T |\langle \Psi_N({\sigma})\, ,\, \Psi_{N,M}^{(\Delta)} \rangle|^2
+ \sum_{\alpha=1}^{2^N-1} a_{\alpha}^T |\langle \Psi_{N}(\sigma)\, ,\, \Psi_{\alpha} \rangle|^2\, .
\end{align*}
A similar formula holds for the trace, so we see that since $|a_{\alpha}|<a_{N,M,\kappa}$ for each $\alpha$,
we have that the right hand side of \cref{eq:components} is
\begin{align*}
\lim_{T \to \infty} \frac{|\langle \Psi_N(\sigma)\, ,\, \Psi_{N,M}^{(\Delta)} \rangle|^2
+ \sum_{\alpha=1}^{2^N-1} 
\left(\frac{a_{\alpha}}{a_{N,M,\kappa}}\right)^T 
|\langle \Psi_{N}(\sigma)\, ,\, \Psi_{\alpha} \rangle|^2}
{1 + \sum_{\alpha=1}^{2^N-1} 
\left(\frac{a_{\alpha}}{a_{N,M,\kappa}}\right)^T}\,
=\, |\langle \Psi_N({\sigma})\, ,\, \Psi_{N,M}^{(\Delta)} \rangle|^2\, ,
\end{align*}
which is evidently the same as the left hand side of \cref{eq:components}.
\end{proof}
Let us denote 
\begin{equation}
\label{eq:MathcalM}
\mathcal{M}(N,M)\, 
=\, \{\varsigma \in \mathcal{S}_{N,T}\, :\, 
\sum_{i=1}^N \varsigma(E^v_{i,1})
=2M\}\, .
\end{equation}
An immediate consequence is the following.
\begin{corollary}
\label{cor:Q1d}
Let $N$ and $M \in \{-\frac{1}{2}\, N,\dots,\frac{1}{2}\, N\}$ be fixed.  Given $\Delta<1$, suppose \cref{eq:kappaDelta} is satisfied.
Then for any $L \leq N$, 
\begin{equation}
\label{eq:EFPd1}
\lim_{\beta \to \infty}
\frac{\langle \mathbf{M}_M \mathbf{Q}_L \rangle_{N,\Delta,\beta}}
{\langle \mathbf{M}_M \rangle_{N,\Delta,\beta}}\,
=\, \lim_{T \to \infty} 
\frac{\sum_{\varsigma \in \mathcal{S}_{N,T}} W_{N,T}^{\kappa}(\varsigma)
\mathbf{1}\{\forall i \in \B_L\, ,\ \varsigma(E^v_{i,1})=+1\}
\mathbf{1}_{\mathcal{M}(N,M)}}
{\sum_{\varsigma \in \mathcal{S}_{N,T}} W_{N,T}^{\kappa}(\varsigma)
\mathbf{1}_{\mathcal{M}(N,M)}}\,.
\end{equation}
\end{corollary}
\begin{proof}
Use the fact that $\langle \mathbf{M}_M X \rangle_{N,\Delta,\beta}$
is asymptotic to $\langle \Psi^{(\Delta)}_{N,M}, X \Psi^{(\Delta)}_{N,M} \rangle$ as $\beta \to \infty$ for any operator $X$ commuting
with $\mathbf{M}_M$, 
and then in \cref{eq:components} sum over all $\sigma \in \Sigma_N$
such that $\sigma_{i} = +1$ for all $i \in \B_L$.
\end{proof}

\section{The Upper Bound in  \Cref{E:Ld3} in \Cref{thm:1d}}
\label{sec:UB6}

For the 1d problem the upper bound in \Cref{E:Ld3} is easier than the lower bound.
So we start with the proof of that.
For both bounds, a key will be to use \Cref{cor:Q1d}.
To begin with, we use \Cref{lem:alternate} to bound the number of changes of downspins
in any interval when going from one row to the next.

\begin{lemma}
\label{lem:diff1}
Suppose that $L\leq N$.
For any $\varsigma \in \mathcal{S}_{N,T}$, if 
$\sum_{i=1}^{L} \varsigma(E^v_{i,1}) = L-2k$,
then $\sum_{i=1}^{L} \varsigma(E^v_{i,2})=L-2k'$
where $k'$ is in $\{k-1,k,k+1\}$.
\end{lemma}
\begin{proof}
This is the same basic argument as in the proof of \Cref{cor:equal}.
What changes is that, when $L$ is less than $N$, the number of sinks and sources on the vertices $V_{ij}$, $i \in \{1,\dots,L\}$
need not be equal. But they can only differ by at most one, by \Cref{lem:alternate}.
Hence $k'-k$ must be in $\{-1,0,1\}$.
\end{proof}
An immediate corollary is this.
\begin{corollary}
\label{cor:density}
Suppose that for a fixed $j \in \{1,\dots,N\}$, 
we have $\sum_{i=1}^L \varsigma(E^v_{ij}) = L$.
For any $\epsilon>0$ and any integer $\ell$ such that $\ell+1\leq \epsilon L$,
we also have $\sum_{i=1}^L \sum_{j'=j-\ell}^{j+\ell} \varsigma(E^v_{i,j'}) \geq L (2\ell+1)(1-\epsilon)$.
\end{corollary}
\begin{proof}
Iterate \Cref{lem:diff1} to obtain $\sum_{i=1}^{L} \varsigma(E^v_{i,j'}) = L - 2k'$
with $|k'| \leq |j'-j|$.
Summing over $j' \in \{j-\ell,\dots,j+\ell\}$ (and using the formula for the sum of an arithmetic sequence) gives $\sum_{i=1}^{L} \sum_{j'=j-\ell}^{j+\ell} \varsigma(E^v_{i,j'}) = L(2\ell+1) - 2K$
where $K \leq \ell(\ell+1)$.
Using $\ell +1\leq \epsilon \cdot L$ and $2\ell \leq 2\ell+1$,
we see that $2K\leq \epsilon L (2\ell+1)$.
\end{proof}

To state the next corollary, let $S^{(j)}_{L,\ell}$ be the set of all valid six-vertex configurations $\varsigma$ on the set of edges incident to all vertices $V_{i,j'}$, $i \in \{1,\dots,L\}$, 
$j' \in \{j-\ell+1,\dots,j+\ell\}$, subject to the constraint
$\sum_{i=1}^L \varsigma(E^v_{i,j}) = L$.
\begin{corollary}
\label{cor:DensityPlus}
With the same hypotheses as in the last corollary, 
\begin{equation}
\label{eq:SnumBound}
\frac{\ln|S^{(j)}_{L,\ell}|}{L(2\ell+1)}\, \leq\, -\epsilon \ln \left(\epsilon\right)
- \left(1-\epsilon\right) \ln \left(1 - \epsilon\right) + \frac{\ln(2)}{L}\, .
\end{equation}
\end{corollary}

\begin{proof}
Let us first focus just on the spins on the vertical edges: $\varsigma(E^{v}_{i,j'})$ for $i \in \{1,\dots,L\}$
and $|j-j'|\leq \ell$.
By \cref{cor:density}, we know that the total number of $\downarrow$ spins among these vertical
edges is at most an integer $K$ where $K \leq \epsilon M$, where we define $M= L(2\ell+1)\epsilon$ for notational convenience.

Then the standard binomial large deviation bound gives that the total number of choices of placing
$K$ $\downarrow$ spins among $M$ vertical edges is bounded by
$\exp(-K\ln(\epsilon)-(M-K)\ln(1-\epsilon))$.
The extra $\ln(2)/L$ arises from the fact that for any row of horizontal edges with no sources or sinks (an unusual occurrence)
all spins will point in the same direction, and there are two choices $\leftarrow$ or $\rightarrow$.
\end{proof}

Equation \cref{eq:SnumBound} is a basic entropy-type bound.
The upper bound of \cref{E:Ld3} will use entropy bounds to bound the likelihood of having
a long row of vertical edges with all spins aligned.

\subsection{Chessboard estimate and    strategy for upper bound in \Cref{thm:1d}}
Fr\"ohlich, Israel, Lieb and Simon proved that the six-vertex model is a reflection positive
classical spin system, much like the classical Ising model \cite{FILS}.
See section 6.1 of that reference culminating in Theorem 6.1, there.
We state the following very minor extension of their result here.

Let us define the Boltzmann-Gibbs state, as
\begin{equation}
\label{eq:BGdef}
\langle f(\varsigma) \rangle_{N,T,\kappa}\,
=\, \frac{\sum_{\varsigma \in \mathcal{S}_{N,T}} W_{N,T}^{\kappa}(\varsigma)
f(\varsigma)}
{\sum_{\varsigma \in \mathcal{S}_{N,T}} W_{N,T}^{\kappa}(\varsigma)}\, ,
\end{equation}
and the state constrained to the zero magnetization, as
\begin{equation}
\label{eq:BGdef2}
\langle f(\varsigma) \rangle_{N,T,\kappa,0}\,
=\, \frac{\langle f(\varsigma) \mathbf{1}_{\mathcal{M}(N,0)}(\varsigma) \rangle_{N,T,\kappa}}
{\langle \mathbf{1}_{\mathcal{M}(N,0)}(\varsigma) \rangle_{N,T,\kappa}}\, .
\end{equation}

Let $\theta^{(1)}_{n,I}, \theta^{(2)}_{t,J}$ is defined by
\begin{itemize}
\item if $I$ is even then $\theta^{(1)}_{t,I}$ is translation in the horizontal direction by $nI$;
\item if $I$ is odd then $\theta^{(1)}_{t,I}$ involves translating by $nI$, then mapping spins on the horizontal and vertical edges by reflecting (in the horizontal directions) all edges, but keeping the orientation of the spins the same on the horizontal edges,
and then flipping the orientation of all spins on vertical edges; 
\end{itemize}
and $\theta^{(2)}_{t,J}$ is defined similarly by switching the notion of horizontal and vertical directions.  See \Cref{fig:reflect} for an ilustration

Then we have the following result.
\begin{theorem}
\label{thm:FILS2}
For any $n$ and $t$, let $A \subset \mathcal{S}_{N,T}$ be an event depending only on the spins at edges incident to vertices $V_{ij}$
for $1\leq i\leq n$ and $1\leq j\leq T$.
Then, for any $a$ and $b$ such that $2^an \leq N$ and $2^b t\leq T$,
\begin{equation}
\label{ineq:FILS2}
\langle \mathbf{1}_{A} \rangle_{N,T,\kappa,0}\,
\leq\, \bigg\langle \prod_{I=0}^{2^a-1} \prod_{J=0}^{2^b-1} 
\mathbf{1}_{\theta^{(1)}_{n,I}(\theta^{(2)}_{t,J}(A))}\bigg\rangle_{N,T,\kappa,0}^{1/2^{a+b}}\, ,
\end{equation}
\end{theorem}

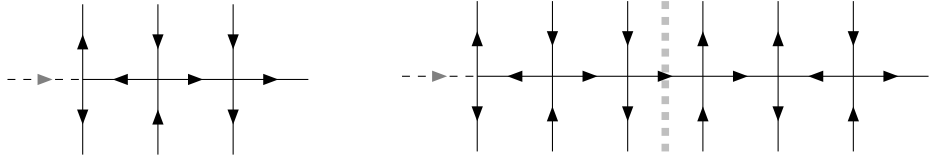
\begin{figure}
\begin{center}
\begin{tikzpicture}[xscale=1,yscale=1]
\draw[dashed] (-1,0) -- (0,0);
\draw (0,0) -- (3,0);
\foreach \x in {0,1,...,2}
{\draw (\x,-1) -- (\x,1);}
\fill (0,-0.5) \dntri (1,-0.5) \uptri (2,-0.5) \dntri;
\fill (0.5,0) \letri (1.5,0) \ritri (2.5,0) \ritri;
\fill[gray] (-0.5,0) \ritri;
\fill (0,0.5) \uptri (1,0.5) \dntri (2,0.5) \dntri;
\end{tikzpicture}
\hspace{1cm}
\begin{tikzpicture}[xscale=1,yscale=1]
\draw[line width=1mm, gray!50!white, dashed] (2.5,1) -- (2.5,-1);
\draw[dashed] (-1,0) -- (0,0);
\draw (0,0) -- (6,0);
\foreach \x in {0,1,...,5}
{\draw (\x,-1) -- (\x,1);}
\fill (0,-0.5) \dntri (1,-0.5) \uptri (2,-0.5) \dntri (3,-0.5) \uptri (4,-0.5) \dntri (5,-0.5) \uptri;
\fill (0.5,0) \letri (1.5,0) \ritri (2.5,0) \ritri (3.5,0) \ritri (4.5,0) \letri (5.5,0) \ritri;
\fill[gray] (-0.5,0) \ritri;
\fill (0,0.5) \uptri (1,0.5) \dntri (2,0.5) \dntri (3,0.5) \uptri (4,0.5) \uptri (5,0.5) \dntri;
\end{tikzpicture}
\end{center}
\caption{\label{fig:reflect}
For $n=3$ and $t=1$,
if $A$ is the set of spin configurations with the prescribed spins at the edges pictured on the left picture
then $A \cap \theta^{(1)}_{n,1}(A)$ is pictured on the right picture.
}
\end{figure}
This is a direct translation of Theorem 6.1 in \cite{FILS}, except for the fact that we
have conditioned on the event $\mathcal{M}(N,0)$.
We will comment briefly on this generalization in \Cref{sec:6vtx}.

\Cref{thm:FILS2} should be compared to  \Cref{lem:ChessLoss}.
Actually, at this point there is no ``loss'' in \Cref{thm:FILS2}.
We will provide the lossy version in the context of the actual proof of the upper bound.

\subsubsection{Proof outline for the upper bound:}
\Cref{thm:FILS2} is key, much as \Cref{lem:ChessLoss} was key to 
\Cref{E:Ud2}.
Because of some differences in the models we use a slightly different argument than before.

The beginning of the upper bound is still to consider a long string of aligned spins.
Now we consider vertical edges $E_{ij}^v$ for say $i \in \{1,\dots,L\}$, and any $j$.
We will let $\mathcal{A}_L^{(j)}$ be the event that all the spins on these $L$ edges are $\uparrow$ spins.
Using \Cref{cor:DensityPlus} we know that there are relatively few choices of spin configurations
for a box of height $\ell$, with $\ell+1\leq \epsilon L$, surrounding this row.
Using reflection positivity as in \Cref{thm:FILS2} we will disseminate the event $\mathcal{A}_L^{(j)}$, tiling
$\mathbb{T}_{N,T}$ with tiles of length $n=L$ and height $t=2\ell+1$.
More precisely, we choose $j=\ell$ so that the first tile is in the lower left corner
before periodic boundary conditions.

Because the number of choices of configurations in each tile, satisfying 
$\theta^{(1)}_{n,I}(\theta^{(2)}_{t,J}(\mathcal{A}^{(\ell)}_{L}))$, is so low,
the actual weights $w_{ij}^{\kappa}(\varsigma)$ will not be able to make up for this,
even if $\kappa$ is negative.

We will need to show that the total number of configurations times weights is small in comparison
to the unconstrained partition function.
For this purpose we will consider configurations which still have a small density of mis-aligned 
spins on vertical edges. But we will make the density significantly larger than the density $\epsilon$
associated to the events $\mathcal{A}^{(\ell)}_L$ disseminated as above.
Then we use straightforward binomial lower bounds on the partition function.

\subsection{The Partition Function: Denominator Bound}

Let us define
\begin{equation}
\label{eq:ZNTkappa0}
Z_{N,T}(\kappa,0)\, =\,
\sum_{\varsigma \in \mathcal{S}_{N,T}} W_{N,T}^{\kappa}(\varsigma)
\mathbf{1}_{\mathcal{M}(N,0)}(\varsigma)\, 
\end{equation}
and note that
\begin{equation}
\label{eq:NTkappa0}
\langle f(\varsigma) \rangle_{N,T,\kappa,0}\,
=\, \frac{
\sum_{\varsigma \in \mathcal{S}_{N,T}} W_{N,T}^{\kappa}(\varsigma)
f(\varsigma)
\mathbf{1}_{\mathcal{M}(N,0)}(\varsigma)}{Z_{N,T}(\kappa,0)}\, \, 
\end{equation}
where the LHS was defined in \Cref{eq:BGdef2}.
With $f(\sigma)$ taken from the right hand side of \cref{ineq:FILS2}, an upper bound on the left hand side of \cref{eq:NTkappa0},
results from an upper bound on the numerator of the right hand side of \cref{eq:NTkappa0},
and a lower bound on the denominator.
We are going to start with a lower bound on the denominator $Z_{N,T}(\kappa,0)$.

Let $R$ be any positive integer. Recall that $N$ and $T$ are both assumed to be even
for the torus $\mathbb{T}_{N,T}$.
We divide $N$ into $2\lfloor N/(2R) \rfloor$ intervals of length $R$.
There may be a remainder of sites, $N-2\lfloor N/(2R) \rfloor$, which is even.
We split these remaining sites into two subintervals of equal size as well.

Now we have a ``background'' configuration. This is a configuration that we will perturb.
On rows of vertical edges, on alternating tiles, we make the spins all $\uparrow$
and all $\downarrow$.
That is the background configuration.

The rows of vertical edges are also numbered by $j$.
On even rows, in each full tile, we choose one vertical edge to 
have a reversed spin.
An example is shown in \Cref{fig:Bin}.
\begin{figure}
\begin{center}
\begin{tikzpicture}[xscale=0.7,yscale=0.7,very thin]
\fill[gray,xshift=-0.5cm] (0,0) rectangle (16,1);
\fill[gray,xshift=-0.5cm] (0,2) rectangle (16,3);
\draw[very thick, xshift=-0.5cm] (0,0) -- (0,4) (2,0) -- (2,4) (4,0) -- (4,4) (6,0) -- (6,4) 
(8,0) -- (8,4) (10,0) -- (10,4) (12,0) -- (12,4) (14,0) -- (14,4) (16,0) -- (16,4);
\draw (-0.5,0) -- (15.5,0) (-0.5,1) -- (15.5,1) (-0.5,2) -- (15.5,2) (-0.5,3) -- (15.5,3);
\foreach \x in {0,1,...,15}
{
\draw (\x,0) -- (\x,4);
}
\fill (1,0.5) \dntri (5,0.5) \dntri (8,0.5) \dntri (13,0.5) \dntri;
\filldraw (1,1) circle (1mm) (5,1) circle (1mm) (8,1) circle (1mm) (13,1) circle (1mm);
\filldraw[fill=white, thick] (1,0) circle (1mm) (5,0) circle (1mm) (8,0) circle (1mm) (13,0) circle (1mm);
\draw (0,0.5) \uptri (4,0.5) \uptri (9,0.5) \uptri (12,0.5) \uptri;
\fill (3,0.5) \uptri (7,0.5) \uptri (11,0.5) \uptri (15,0.5) \uptri;
\filldraw (3,0) circle (1mm) (7,0) circle (1mm) (11,0) circle (1mm) (15,0) circle (1mm);
\filldraw[fill=white, thick] (3,1) circle (1mm) (7,1) circle (1mm) (11,1) circle (1mm) (15,1) circle (1mm);
\draw (2,0.5) \dntri (6,0.5) \dntri (10,0.5) \dntri (14,0.6) \dntri;
\fill (1,2.5) \dntri (4,2.5) \dntri (8,2.5) \dntri (13,2.5) \dntri;
\filldraw (1,3) circle (1mm) (4,3) circle (1mm) (8,3) circle (1mm) (13,3) circle (1mm);
\filldraw[fill=white, thick] (1,2) circle (1mm) (4,2) circle (1mm) (8,2) circle (1mm) (13,2) circle (1mm);
\draw (0,2.5) \uptri (5,2.5) \uptri (9,2.5) \uptri (12,2.5) \uptri;
\fill (2,2.5) \uptri (6,2.5) \uptri (11,2.5) \uptri (14,2.5) \uptri;
\filldraw (2,2) circle (1mm) (6,2) circle (1mm) (11,2) circle (1mm) (14,2) circle (1mm);
\filldraw[fill=white, thick] (2,3) circle (1mm) (6,3) circle (1mm) (11,3) circle (1mm) (14,3) circle (1mm);
\draw (3,2.5) \dntri (7,2.5) \dntri (10,2.5) \dntri (15,2.5) \dntri;

\foreach \x in {0,4,8,12}
{
\foreach \y in {1.5,3.5}
{
\draw (\x,\y) \uptri;
\draw[xshift=1cm] (\x,\y) \uptri;
}
}
\foreach \x in {2,6,10,14}
{
\foreach \y in {1.5,3.5}
{
\draw (\x,\y) \dntri;
\draw[xshift=1cm] (\x,\y) \dntri;
}
}
\end{tikzpicture}
\caption{\label{fig:Bin}
An example of the background configuration and choices of reversed spins for $R=2$.}
\end{center}
\end{figure}
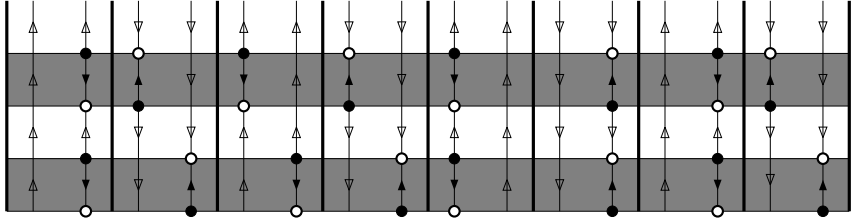
For each reversed spin, if it is $\downarrow$, then we have a source at the top vertex (a filled circle)
and a sink at the bottom vertex (an empty circle).
If instead we have an $\uparrow$, then the location of the sink and source are exchanged relative to the reversed spin.  
Now we have $2\lfloor N/(2R)\rfloor$ sinks and sources on each line of horizontal edges.
Therefore for each $\varsigma$ obtained this way, 
\[
W_{N,T}^{\kappa}(\varsigma)= \exp( 2 \kappa \lfloor N/(2R)\rfloor T).
\]

To estimate the number of configurations obtained, note that in each full tile, there is also a choice of $R$ for the location of the reversed spin. Since we are only modifying spins on even rows, the total number of $\varsigma$'s obtained in this way is at least
$\exp(\lfloor N/(2R)\rfloor T \ln(R))$.
We thus obtain the lower bound
\begin{equation}
\label{ineq:PFlb}
\frac{\ln Z_{N,T}(\kappa,0)}{NT}\,
\geq\, \frac{ \lfloor N/(2R)\rfloor}{N}\, (2\kappa + \ln(R))\, .
\end{equation}
This is not a particularly good bound in general,
but is convenient for comparing to the upper bound we will obtain in the next subsection.

\subsection{The Chessboard Estimate with Some Loss: Numerator Bound}
We begin with a more direct analogue of \Cref{lem:ChessLoss}.
Let 
\[
\mathcal{A}(L)= \{\varsigma \in \mathcal{S}_{N,T}\, :\, \forall i \in \mathbb{B}_L\, ,\ \varsigma(E_{i,1}^v)=+1\}.
\]

For integers $L\leq N/2$ and $\ell \leq T/2$, let us also define the disseminated event
$\widehat{\mathcal{A}}(N,T,L,\ell)$ as follows:
For the horizontal direction, make a decomposition of the length $N$ into
$2 \lfloor N/(2L) \rfloor$ tiles each of length $L$, as well as a number of remainder
vertices, numbering less than $2L$.
Split the remainder vertices into two partial tiles each with equal numbers of vertices.
This can be done since $N - 2 \lfloor N/(2L) \rfloor$ is even.

For the vertical direction, make a decomposition of the height $T$
into $2 \lfloor T/(2\ell) \rfloor$ tiles each of length $\ell$, as well as a number
of remainder vertices, numbering less than $2\ell$.
Split the remainder vertices into two partial tiles each with equal numbers of vertices.
This can be done since $T - 2 \lfloor T/(2\ell) \rfloor$ is even.

Disseminate the event $\mathcal{A}(L)$, viewed as occupying the first tile,
into these $2 \lfloor N/(2L) \rfloor \cdot 2 \lfloor T/(2\ell) \rfloor$ other tiles,
using reflections associated to $\theta_{n,I}^{(1)}$ and $\theta_{t,J}^{(2)}$
for $n=L$ and $t=\ell$.

On the partial tiles, continue the same pattern, except only project onto $\varsigma=+1$
or $\varsigma=-1$ for those sites in the partial tile.
\begin{lemma}
\label{lem:ChessLoss2}
Assume that $L$ is fixed with $L\leq N/2$ and $\ell$ is fixed with $\ell \leq T/2$. Then
\begin{equation}
\label{ineq:ChessLoss2}
\langle \mathbf{1}_{\mathcal{A}(L)}(\varsigma) \rangle_{N,T,\kappa,0}\,
\leq\, \left( \langle \mathbf{1}_{\widehat{\mathcal{A}}(N,T,L,\ell)} \rangle_{N,T,\kappa,0}\right)^{1/K}\, ,
\end{equation}
where $K = 2^{\log_2(N/L)+\log_2(T/\ell)+2}$.
\end{lemma}
This lemma follows from \Cref{thm:FILS2}, requiring just one more reflection in each direction, through any appropriate plane.

Now we will try to bound the right hand side of \cref{ineq:ChessLoss2}.
It will suffice to obtain a good bound on 
$$
\langle \mathbf{1}_{\widehat{\mathcal{A}}(N,T,L,\ell)} \rangle_{N,T,\kappa,0}\, .
$$
For this, we use \cref{eq:NTkappa0} with $f = \mathbf{1}_{\widehat{\mathcal{A}}(N,T,L,\ell)}$.
We then need to bound the right hand side of \cref{eq:NTkappa0}.
We already have a lower bound on the denominator, $Z_{N,T}(\kappa,0)$.
So now we just need an upper bound on the numerator:
$$
\sum_{\varsigma \in \mathcal{S}_{N,T}} W_{N,T}^{\kappa}(\varsigma) 
\mathbf{1}_{\mathcal{M}(N,0)}(\varsigma)
\mathbf{1}_{\widehat{\mathcal{A}}(N,T,L,\ell)}(\varsigma)\, .
$$
This is what we will describe next.

We choose $\ell$ such that $\ell+1 \leq \epsilon L$.
On each full tile, we know by \Cref{cor:DensityPlus} that the set $\mathcal{S}_{L,\ell,I,J}$of 
 valid spin configurations for that tile satisfies
 $$
\frac{|\mathcal{S}_{L,\ell,I,J}|}{(2\ell+1) L}\, \leq\, -\epsilon \ln(\epsilon) - (1-\epsilon) \ln(1-\epsilon)
+\frac{\ln(2)}{L}\, .
$$

Note that, in each full tile, a source or a sink must have the two incident vertical edges of opposite
spins.  By \Cref{cor:density}, the density of sources and sinks is no greater than $2\epsilon$, where we use 
to bound the number of reversed spins (relative to the majority in that tile).
So the product of the $w_{ij}^{\kappa}(\varsigma)$'s for a single full tile is no
greater than $\exp(2 \epsilon |\kappa| \ell L)$.

If one only had to worry about full tiles, then this would lead one to a bound
such as 
$$
\frac{1}{NT}\,
\ln \left(
\sum_{\varsigma \in \mathcal{S}_{N,T}} W_{N,T}^{\kappa}(\varsigma) 
\mathbf{1}_{\mathcal{M}(N,0)}(\varsigma)
\mathbf{1}_{\widehat{\mathcal{A}}(N,T,L,\ell)}(\varsigma)
\right)\, \leq\, -\epsilon \ln(\epsilon) - (1-\epsilon) \ln(1-\epsilon)
+\frac{\ln(2)}{L} + 2 \epsilon |\kappa|\, .
$$

Since the partial tiles are smaller than the full tiles, we have that there is no greater a number
of configurations satisfying the conditions than for the full tiles, and we still have
that the product of all the $w_{ij}^{\kappa}(\varsigma)$'s is no greater than $\exp(2\epsilon |\kappa| \ell L)$.
Finally, the number of partial tiles is no greater than the number of full tiles.
So we may obtain:
\begin{corollary}
With the set-up as above
\begin{equation}
\label{ineq:UB2}
\begin{split}
\frac{1}{NT}\,
\ln \bigg(
\sum_{\varsigma \in \mathcal{S}_{N,T}} W_{N,T}^{\kappa}(\varsigma) 
\mathbf{1}_{\mathcal{M}(N,0)}(\varsigma)
\mathbf{1}_{\widehat{\mathcal{A}}(N,T,L,\ell)}(\varsigma)
\bigg)\, 
&\leq\, -2\epsilon \ln(\epsilon) - 2(1-\epsilon) \ln(1-\epsilon)
+ 4 \epsilon |\kappa|\\
&\hspace{-2cm}
+\frac{2\ln(2)}{L} + \left(\frac{2L-1}{N} + \frac{2\ell-1}{T}\right)(|\kappa|+2\ln(2))\, .
\end{split}
\end{equation}
\end{corollary}

\begin{proof}
This follows from the argument above and the realization that imposing correct boundary
conditions at the edge of each  tile (instead of considering each  tile as having free boundary
conditions as above) only further reduces the total number of valid configurations.
\end{proof}

\subsection{Completion of the proof of upper bound for  \cref{thm:1d}}

We will use \cref{ineq:PFlb}. Let us denote $\eta = 1/(8R)$.
Then we may rewrite it as 
\begin{equation}
\label{ineq:PFlb2}
\frac{\ln Z_{N,T}(\kappa,0)}{NT}\, \geq\, \frac{\lfloor 4\eta N \rfloor}{N}\, (-2|\kappa| - \ln(\eta) - 3\ln(2))\, .
\end{equation}
Note that for $N$ sufficiently large, we have $\lfloor 4\eta N \rfloor/N \geq 2\eta$.
Also, for $\eta$ sufficiently small, we have $-\ln(\eta) - 2|\kappa| -3\ln(2)$ is positive.
So, combining \cref{ineq:PFlb2} with \cref{eq:NTkappa0} and \cref{ineq:UB2}, we have
\begin{equation}
\label{ineq:UppBd}
\begin{split}
\frac{1}{NT}\, \ln \langle \mathbf{1}_{\widehat{\mathcal{A}}(N,T,L,\ell)}\rangle_{N,T,\kappa,0}\,
&\leq\, - 2\epsilon \ln(\epsilon) + 2\eta \ln(\eta) \\
&\qquad + 4(\epsilon+\eta) |\kappa|
-2(1-\epsilon) \ln(1-\epsilon) + 3 \eta \ln(2)\\
&\qquad
+\frac{2\ln(2)}{L}\, .
\end{split}
\end{equation}
We have tried to arrange the terms on the right hand side as: most important first, on the first line,
second most important on the second line, and least important on the last line.
The term on the last line goes to $0$ as $L \to \infty$.

Let us fix $\epsilon>0$ and choose $R$ such that $\eta = 2\epsilon$.
Then the first line is $-2\epsilon |\ln(\epsilon)| + 4\epsilon \ln(2)$.
We see that by choosing  $\epsilon$ small enough, we can arrange that the left hand
side of \cref{ineq:UppBd} is bounded by a strictly negative quantity $-\alpha$ for some $\alpha=\alpha(\epsilon)>0.$

We have
$$
\ln \langle \mathbf{1}_{\widehat{\mathcal{A}}(N,T,L,\ell)}\rangle_{N,T,\kappa,0}\,
\leq\, -\alpha NT,\, 
$$
so that
$$
\frac{1}{K}\, \ln \langle \mathbf{1}_{\widehat{\mathcal{A}}(N,T,L,\ell)}\rangle_{N,T,\kappa,0}\,
\leq\, -\frac{\alpha N T}{K}\, .
$$
Since $K \leq 4 NT/(\ell L)$, we have 
$$
\frac{1}{K}\, \ln \langle \mathbf{1}_{\widehat{\mathcal{A}}(N,T,L,\ell)}\rangle_{N,T,\kappa,0}\,
\leq\, -\frac{\alpha \ell L}{4}\, .
$$
Finally, we chose $\ell$ just so that $\ell+1 \leq \epsilon L$, and $\epsilon>0$ is small but fixed. We now also assume $\ell \geq \frac{1}{2} \epsilon L$ and with this choice obtain
$$
\frac{1}{K}\, \ln \langle \mathbf{1}_{\widehat{\mathcal{A}}(N,T,L,\ell)}\rangle_{N,T,\kappa,0}\,
\leq\, -\frac{\alpha \epsilon}{8}\, L^2\, .
$$
Putting this together with \Cref{lem:ChessLoss2} and \Cref{cor:Q1d}
(and the definition of $\mathcal{A}(L)$ as the event that $\varsigma(E^v_{i,1})=+1$
for all $i \in \mathbb{B}_L$), we have the desired result.

\section{The Lower bound  in \Cref{thm:1d}}
\label{sec:LBthm2}
The lower bound argument uses the six-vertex configuration again, and 
 \Cref{cor:Q1d}.
The argument from here on is more combinatorial.
The idea is to try
show that given a typical configuration  we may perturb it in a square in such a way
as to obtain a large interval of all aligned spins.
Because configurations are discrete in both space and time,
we merely need to construct a configuration.
In essence entropy is cut-off below in a discrete system.
Note that a ``typical configuration,'' here will mean that we sample a configuration
according to the Gibbs state, and we show that with a probability $p$ close to 1 this works.
In fact, this part of the argument will also use reflection positivity a bit, too.
We will comment on this a bit, later.

\subsection{Osculating paths}
We switch our perspective on six-vertex configurations, now using an equivalent formulation in terms of osculating paths 
This is well know and easily constructed:
Start with a reference configuration such as: all vertical edges $\uparrow$ and all horizontal edges $\rightarrow$.
For any given edge, color the interior of the edge black if it differs from the reference configuration.
We want connect up the edges in a natural way to make paths.  The main rule is that the paths are not allowed to cross at a vertex. We round corners whenever two edges meet.  Where four edges meet, we declare the west and north edges part of the same path and the east and south edges part of a separate path (so all paths go "up and right" only).

\begin{figure}
\begin{center}
\begin{tikzpicture}[xscale=0.7,yscale=0.7]
\draw (0,1.4) node[] {$1$};
\draw (-0.7,0) -- (0.7,0) (0,-0.7) -- (0,0.7);
\fill (-0.35,0) \ritri (0.35,0) \ritri; \fill (0,-0.35) \uptri (0,0.35) \uptri;
\begin{scope}[yshift=-1.85cm]
\draw (-0.7,0) -- (0.7,0) (0,-0.7) -- (0,0.7);
\draw[rounded corners, line width=2pt,gray!75!black,nearly opaque] (-0.7,0) -- (0,0) -- (0,0.7) (0,-0.7) -- (0,0) -- (0.7,0);
\end{scope}
\begin{scope}[xshift=3cm]
\draw (0,1.4) node[] {$2$};
\draw (-0.7,0) -- (0.7,0) (0,-0.7) -- (0,0.7);
\fill (-0.35,0) \letri (0.35,0) \letri; \fill (0,-0.35) \uptri (0,0.35) \uptri;
\begin{scope}[yshift=-1.85cm]
\draw (-0.7,0) -- (0.7,0) (0,-0.7) -- (0,0.7);
\draw[rounded corners, line width=2pt,gray!75!black,nearly opaque] (0,-0.7) -- (0,0.7);
\draw[rounded corners, line width=2pt] (-0.7,0) -- (0.7,0);
\end{scope}
\end{scope}
\begin{scope}[xshift=6cm]
\draw (0,1.4) node[] {$3$};
\draw (-0.7,0) -- (0.7,0) (0,-0.7) -- (0,0.7);
\fill (-0.35,0) \ritri (0.35,0) \ritri; \fill (0,-0.35) \dntri (0,0.35) \dntri;
\begin{scope}[yshift=-1.85cm]
\draw (-0.7,0) -- (0.7,0) (0,-0.7) -- (0,0.7);
\draw[rounded corners, line width=2pt,gray!75!black,nearly opaque] (-0.7,0) -- (0.7,0);
\draw[rounded corners, line width=2pt] (0,-0.7) -- (0,0.7);
\end{scope}
\end{scope}
\begin{scope}[xshift=9cm]
\draw (0,1.4) node[] {$4$};
\draw (-0.7,0) -- (0.7,0) (0,-0.7) -- (0,0.7);
\fill (-0.35,0) \letri (0.35,0) \letri; \fill (0,-0.35) \dntri (0,0.35) \dntri;
\begin{scope}[yshift=-1.85cm]
\draw (-0.7,0) -- (0.7,0) (0,-0.7) -- (0,0.7);
\draw[rounded corners, line width=2pt] (0,-0.7) -- (0,0) -- (0.7,0) (-0.7,0) -- (0,0) -- (0,0.7);
\end{scope}
\end{scope}
\begin{scope}[xshift=12cm]
\draw (0,1.4) node[] {$5$};
\draw (-0.7,0) -- (0.7,0) (0,-0.7) -- (0,0.7);
\fill (-0.35,0) \letri (0.35,0) \ritri; \fill (0,-0.35) \uptri (0,0.35) \dntri;
\begin{scope}[yshift=-1.85cm]
\draw (-0.7,0) -- (0.7,0) (0,-0.7) -- (0,0.7);
\draw[rounded corners, line width=2pt,gray!75!black,nearly opaque] (0,-0.7) -- (0,0) -- (0.7,0);
\draw[rounded corners, line width=2pt] (-0.7,0) -- (0,0) -- (0,0.7);
\end{scope}
\end{scope}
\begin{scope}[xshift=15cm]
\draw (0,1.4) node[] {$6$};
\draw (-0.7,0) -- (0.7,0) (0,-0.7) -- (0,0.7);
\fill (-0.35,0) \ritri (0.35,0) \letri; \fill (0,-0.35) \dntri (0,0.35) \uptri;
\begin{scope}[yshift=-1.85cm]
\draw (-0.7,0) -- (0.7,0) (0,-0.7) -- (0,0.7);
\draw[rounded corners, line width=2pt,gray!75!black,nearly opaque] (-0.7,0) -- (0,0) -- (0,0.7);
\draw[rounded corners, line width=2pt] (0,-0.7) -- (0,0) -- (0.7,0);
\end{scope}
\end{scope}
\end{tikzpicture}
\caption{
\label{fig:trans}
This is a translation of the 6 valid vertex configurations, and 6 vertex configuriations for the osculating paths.
We enumerate them for later reference.
}
\end{center}
\end{figure}
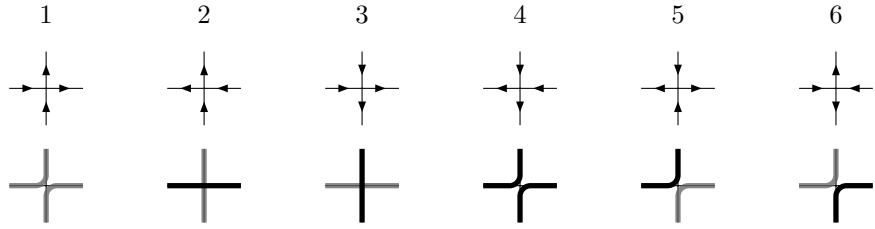
In \Cref{fig:trans}, we 
show the translation from the 6 valid types of vertices in the six-vertex model and the 6 types of valid vertices for an osculating path configuration.

Let us abbreviate ``opc'' for ``osculating path configuration.''
If we use as a reference configuration $\downarrow$ on vertical edges and $\leftarrow$ on horizontal edges we obtain a different opc, which we refer to as a gray opc.

For later reference, we note that the black and the gray opc's from a fixed six vertex configuration occupy complementary edges.
Two paths of different colors may cross at a vertex in this joint picture
The black opc has all of the information, but there are times when the picture of the gray opc is helpful.

\begin{figure}
\begin{center}
\begin{tikzpicture}[xscale=0.65,yscale=0.65,very thin]
\draw (0,1) -- (6,1) (0,2) -- (6,2) (0,3) -- (6,3) (0,4) -- (6,4) (0,5) -- (6,5);
\draw  (1,0) -- (1,6) (2,0) -- (2,6) (3,0) -- (3,6) (4,0) -- (4,6) (5,0) -- (5,6);
\fill (1,0.5) \dntri (2,0.5) \dntri (3,0.5) \uptri (4,0.5) \uptri (5,0.5) \uptri;
\fill (0.5,1) \ritri (1.5,1) \letri (2.5,1) \letri (3.5,1) \letri (4.5,1) \ritri (5.5,1) \ritri;
\fill (1,1.5) \uptri (2,1.5) \dntri (3,1.5) \uptri (4,1.5) \dntri (5,1.5) \uptri;
\fill (1,1) circle (1mm);
\filldraw[fill=white] (4,1) circle (1mm);
\fill (0.5,2) \ritri (1.5,2) \ritri (2.5,2) \letri (3.5,2) \ritri (4.5,2) \letri (5.5,2) \letri;
\fill (1,2.5) \uptri (2,2.5) \uptri (3,2.5) \dntri (4,2.5) \uptri (5,2.5) \uptri;
\filldraw[fill=white] (3,2) circle (1mm);
\fill (2,2) circle (1mm) (4,2) circle (1mm);
\fill (0.5,3) \letri (1.5,3) \letri (2.5,3) \ritri (3.5,3) \letri (4.5,3) \letri (5.5,3) \letri;
\fill (1,3.5) \uptri (2,3.5) \dntri (3,3.5) \uptri (4,3.5) \uptri (5,3.5) \uptri;
\filldraw[fill=white] (2,3) circle (1mm);
\fill (3,3) circle (1mm);
\fill (0.5,4) \letri (1.5,4) \ritri (2.5,4) \letri (3.5,4) \letri (4.5,4) \ritri (5.5,4) \ritri;
\fill (1,4.5) \dntri (2,4.5) \uptri (3,4.5) \uptri (4,4.5) \dntri (5,4.5) \uptri;
\fill (2,4) circle (1mm);
\filldraw[fill=white] (1,4) circle (1mm) (4,4) circle (1mm);
\fill (0.5,5) \ritri (1.5,5) \letri (2.5,5) \ritri (3.5,5) \ritri (4.5,5) \letri (5.5,5) \letri;
\fill (1,5.5) \uptri (2,5.5) \dntri (3,5.5) \uptri (4,5.5) \uptri (5,5.5) \uptri;
\filldraw[fill=white] (2,5) circle (1mm);
\fill (1,5) circle (1mm) (4,5) circle (1mm);
\end{tikzpicture}
\hspace{1cm}
\begin{tikzpicture}[xscale=0.65,yscale=0.65,very thin]
\draw (0,1) -- (6,1) (0,2) -- (6,2) (0,3) -- (6,3) (0,4) -- (6,4) (0,5) -- (6,5);
\draw  (1,0) -- (1,6) (2,0) -- (2,6) (3,0) -- (3,6) (4,0) -- (4,6) (5,0) -- (5,6);
\fill (1,0.5) \dntri (2,0.5) \dntri (3,0.5) \uptri (4,0.5) \uptri (5,0.5) \uptri;
\fill (0.5,1) \ritri (1.5,1) \letri (2.5,1) \letri (3.5,1) \letri (4.5,1) \ritri (5.5,1) \ritri;
\fill (1,1.5) \uptri (2,1.5) \dntri (3,1.5) \uptri (4,1.5) \dntri (5,1.5) \uptri;
\fill (1,1) circle (1mm);
\filldraw[fill=white] (4,1) circle (1mm);
\fill (0.5,2) \ritri (1.5,2) \ritri (2.5,2) \letri (3.5,2) \ritri (4.5,2) \letri (5.5,2) \letri;
\fill (1,2.5) \uptri (2,2.5) \uptri (3,2.5) \dntri (4,2.5) \uptri (5,2.5) \uptri;
\filldraw[fill=white] (3,2) circle (1mm);
\fill (2,2) circle (1mm) (4,2) circle (1mm);
\fill (0.5,3) \letri (1.5,3) \letri (2.5,3) \ritri (3.5,3) \letri (4.5,3) \letri (5.5,3) \letri;
\fill (1,3.5) \uptri (2,3.5) \dntri (3,3.5) \uptri (4,3.5) \uptri (5,3.5) \uptri;
\filldraw[fill=white] (2,3) circle (1mm);
\fill (3,3) circle (1mm);
\fill (0.5,4) \letri (1.5,4) \ritri (2.5,4) \letri (3.5,4) \letri (4.5,4) \ritri (5.5,4) \ritri;
\fill (1,4.5) \dntri (2,4.5) \uptri (3,4.5) \uptri (4,4.5) \dntri (5,4.5) \uptri;
\fill (2,4) circle (1mm);
\filldraw[fill=white] (1,4) circle (1mm) (4,4) circle (1mm);
\fill (0.5,5) \ritri (1.5,5) \letri (2.5,5) \ritri (3.5,5) \ritri (4.5,5) \letri (5.5,5) \letri;
\fill (1,5.5) \uptri (2,5.5) \dntri (3,5.5) \uptri (4,5.5) \uptri (5,5.5) \uptri;
\filldraw[fill=white] (2,5) circle (1mm);
\fill (1,5) circle (1mm) (4,5) circle (1mm);
\draw[black, semitransparent, rounded corners, line width=2pt] (1,0) -- (1,1) -- (2,1) -- (2,2) -- (3,2) -- (3,3) -- (6,3)
(2,0) -- (2,1) -- (4,1) -- (4,2) -- (6,2)
(0,3) -- (2,3) -- (2,4) -- (4,4) -- (4,5) -- (6,5)
(0,4) -- (1,4) -- (1,5) -- (2,5) -- (2,6);
\end{tikzpicture}
\hspace{1cm}
\begin{tikzpicture}[xscale=0.65,yscale=0.65,very thin]
\draw (0,1) -- (6,1) (0,2) -- (6,2) (0,3) -- (6,3) (0,4) -- (6,4) (0,5) -- (6,5);
\draw (1,0) -- (1,6) (2,0) -- (2,6) (3,0) -- (3,6) (4,0) -- (4,6) (5,0) -- (5,6);
\draw[rounded corners, line width=2pt,gray!75!black,nearly opaque] (5,0) -- (5,1) -- (6,1)
(4,0) -- (4,1) -- (5,1) -- (5,4) -- (6,4)
(3,0) -- (3,2) -- (4,2) -- (4,4) -- (5,4) -- (5,6) 
(0,1) -- (1,1) -- (1,2) -- (2,2) -- (2,3) -- (3,3) -- (3,5) -- (4,5) -- (4,6)
(0,2) -- (1,2) -- (1,4) -- (2,4) -- (2,5) -- (3,5) -- (3,6)
(0,5) -- (1,5) -- (1,6) ;
\draw[rounded corners, line width=2pt]  (1,0) -- (1,1) -- (2,1) -- (2,2) -- (3,2) -- (3,3) -- (6,3)
(2,0) -- (2,1) -- (4,1) -- (4,2) -- (6,2)
(0,3) -- (2,3) -- (2,4) -- (4,4) -- (4,5) -- (6,5)
(0,4) -- (1,4) -- (1,5) -- (2,5) -- (2,6);
\end{tikzpicture}
\caption{
\label{fig:Osc}
This picture is an example of a valid 6-vertex configuration and its mapping to an osculating
path configuration (opc). (The middle picture is an intermediate step to guide the eye.)
}
\end{center}
\end{figure}
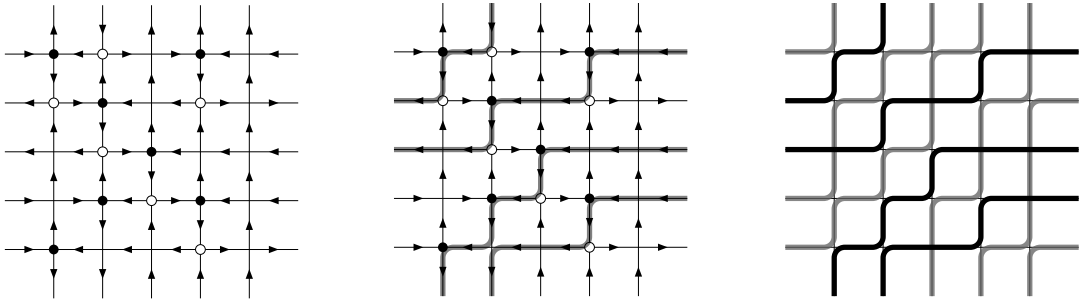
We have shown how to convert a six-vertex configuration into an osculating paths configuration in a particular example in \Cref{fig:Osc}.

\subsection{Local ``isotopy'' moves}
Consider a fixed $L \times R$ block inside $\mathbb{T}_{N,T}$.
(Later we will choose $R$ to be a large, but fixed, multiple of $L$.)
We call this block $\Gamma_{L,R} \subset \mathbb{T}_{N,T}$.Suppose that we have a valid opc $x$ on $\Gamma_{L,R}$ with a prescribed pattern of black and gray edges on $\partial \Gamma_{L,R}$.
We will construct a different $y$ on $\Gamma_{L,R}$, with the same pattern of black and gray edges on $\partial \Gamma_{L,R}$
by performing a sequence of ``local isotopy moves.''

These are based on ``flippable plaquettes.''
\begin{figure}
\begin{center}
\begin{tikzpicture}
\draw (0,0) rectangle (1,1);
\draw[rounded corners, line width=2pt, gray!75!black,nearly opaque] (1,1) -- (0,1) -- (0,0);
\draw[rounded corners, line width=2pt] (1,1) -- (1,0) -- (0,0);
\draw (3,0) rectangle (4,1);
\draw[rounded corners, line width=2pt, gray!75!black,nearly opaque] (4,1) -- (4,0) -- (3,0);
\draw[rounded corners, line width=2pt] (4,1) -- (3,1) -- (3,0);
\end{tikzpicture}
\caption{\label{fig:corn}
We show two possible squares which may be moved by a local move. We call the square on the left a corner of type $C_-$ and the square on the right a corner
of type $C_+$.
}
\end{center}
\end{figure}
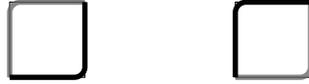
In \Cref{fig:corn}, we have shown two different squares. These are the two possible configurations for a ``flippable plaquette.''
\begin{definition}
\label{def:pMove}
We define a $+$ move to be one which switches a type $C_-$ corner square to a type $C_+$ one.
We define a $-$ move to be the opposite.
\end{definition}

Given an opc $x$ on $\Gamma_{L,R}$, let us define its height,  $h(x)$, as follows.
Recall that an opc $x$ consists black paths (and/or gray paths).
For a given black path, consider all the vertices to the lower right of the path.
The height is equal to the sum, over all black paths, of the number of these vertices.
The main point in defining $h(x)$ is that it increases after $+$ moves (by $1$ with each one).

Given $x$, we perform $+$ moves to each $C_-$-type square, one at a time, until there are no more $C_-$ squares.
It is a fact that even if there were some choices to be made of which $C_-$ square to turn into a $C_+$ square at some stage,
the final configuration $x_{\max}$ is unique.
$x_{\max}$ is sometimes called the highest opc relative to the boundary conditions on $\partial \Gamma_{L,R}$.

\begin{figure}
\begin{center}
\begin{tikzpicture}[xscale=0.65,yscale=0.65,very thin]
\draw (0,1) -- (6,1) (0,2) -- (6,2) (0,3) -- (6,3) (0,4) -- (6,4) (0,5) -- (6,5);
\draw (1,0) -- (1,6) (2,0) -- (2,6) (3,0) -- (3,6) (4,0) -- (4,6) (5,0) -- (5,6);
\draw[rounded corners, line width=2pt,gray!75!black,nearly opaque] (5,0) -- (5,1) -- (6,1)
(4,0) -- (4,1) -- (5,1) -- (5,4) -- (6,4)
(3,0) -- (3,2) -- (4,2) -- (4,4) -- (5,4) -- (5,6) 
(0,1) -- (1,1) -- (1,2) -- (2,2) -- (2,3) -- (3,3) -- (3,5) -- (4,5) -- (4,6)
(0,2) -- (1,2) -- (1,4) -- (2,4) -- (2,5) -- (3,5) -- (3,6)
(0,5) -- (1,5) -- (1,6) ;
\draw[rounded corners, line width=2pt]  (1,0) -- (1,1) -- (2,1) -- (2,2) -- (3,2) -- (3,3) -- (6,3)
(2,0) -- (2,1) -- (4,1) -- (4,2) -- (6,2)
(0,3) -- (2,3) -- (2,4) -- (4,4) -- (4,5) -- (6,5)
(0,4) -- (1,4) -- (1,5) -- (2,5) -- (2,6);
\end{tikzpicture}
\hspace{1cm}
\begin{tikzpicture}[xscale=0.65,yscale=0.65,very thin]
\draw (0,1) -- (6,1) (0,2) -- (6,2) (0,3) -- (6,3) (0,4) -- (6,4) (0,5) -- (6,5);
\draw (1,0) -- (1,6) (2,0) -- (2,6) (3,0) -- (3,6) (4,0) -- (4,6) (5,0) -- (5,6);
\draw[rounded corners, line width=2pt,gray!75!black,nearly opaque] (5,0) -- (5,1) -- (6,1)
(4,0) -- (4,1) -- (5,1) -- (5,4) -- (6,4)
(3,0) -- (3,1) -- (4,1) -- (4,4) -- (5,4) -- (5,6) 
(0,1) -- (3,1) -- (3,4) -- (4,4) -- (4,6)
(0,2) -- (2,2) -- (2,4) -- (3,4) -- (3,6)
(0,5) -- (1,5) -- (1,6) ;
\draw[rounded corners, line width=2pt] (2,0) -- (2,2) -- (6,2)
(1,0) -- (1,3) -- (6,3)
(0,3) -- (1,3) -- (1,4) -- (2,4) -- (2,5) -- (6,5)
(0,4) -- (1,4) -- (1,5) -- (2,5) -- (2,6);
\end{tikzpicture}
\caption{
\label{fig:Highest}
On the left, this picture is the same configuration in \Cref{fig:Osc} again. Now on the right we have
shown the configuration which is the highest opc with the same boundary conditions.
}
\end{center}
\end{figure}
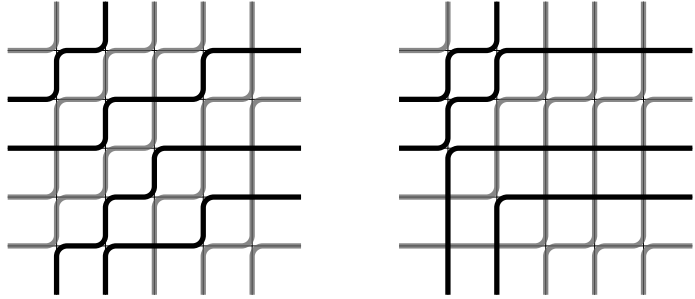

We claim that once this has been done, the resulting configuration will have a long interval of aligned spins, in a row of vertical edges,
unless the boundary edge pattern is {\em unusual}, which happens with a low probability. So the unusual event of a bad boundary edge pattern
will be controllable, {\em a priori}, in terms of the Boltzmann-Gibbs measure.

\subsection{Blockades}
We assume that $\Gamma_{L,R}$ consists of all vertices $V_{ij}$ for $1\leq i\leq L$ and $1\leq j\leq R$.
The boundary of $\Gamma_{L,R}$ are the edges $E_{0,j}^h$ and $E_{L,j}^h$ for $j \in \{1,\dots,R\}$
and $E_{i,0}^v$ and $E_{i,R}^v$ for $i \in \{1,\dots,L\}$.

\begin{lemma}
\label{lem:blockades}
Assume that $x$ is a highest opc in $\Gamma_{L,R}$. Then the following is true in reference to $x$.
Suppose that some vertex $V_{ij}$ in $\Gamma_{L,R}$ is either of type $4$ or $5$ from \Cref{fig:trans}, and suppose $i>1$ and $j<R$.
Then either all vertices $V_{1,j},\dots,V_{i-1,j}$ are of type 4 or all vertices $V_{i,j+1},\dots,V_{i,R}$ are of type 4,
or possibly both.
\end{lemma}

\begin{proof}
This argument proceeds by induction. Let us consider the event that a vertex $V_{ij}$ is of type in $\{4,5\}$ by filling in $E_{i-1,j}^h$ and $E_{ij}^v$ colored
black, but not filling in the other two edges $E_{i,j}^h$ and $E_{i,j-1}^v$ with either color, black or gray, since we do not know which color it may be filled in.
See \Cref{fig:trans}
As the first step of the induction, we note that if $V_{ij}$ has type in $\{4,5\}$, and assuming that $x$ is a highest opc, either $V_{i-1,j}$ or $V_{i,j+1}$ has type 4, or possibly both (here the reader is encouraged to check the possible cases on a piece of paper).

Once $V_{i-1,j}$ or $V_{i,j+1}$ has type 4, the above observation also applies to this vertex. We conclude that from any vertex $V_{ij}$ of type $4$ or $5$ there is a sequence of types $4$ vertices (except possibly for the first) $V_{ij}=V_{i(1),j(1)},\dots,V_{i(K),j(K)}$ such that for each $k\in\{1,\dots,K-1\}$,
$i(k+1)=i(k)$ and $j(k+1)=j(k)+1$, or else $i(k+1)=i(k-1)$ and $j(k+1)=j(k)$.
The number $K$ is the first time that this path hits either the left face of $\partial \Gamma_{L, R}$, so that $i(K)=1$, or the top face so that $j(K)=R$.

Since each of these vertices is of type 4, all 4 edges incident to each vertex in the sequence is colored black.
Thus this connected chain of vertices creates a blockade against any gray path crossing it.
Because gray paths are also up-right paths, this means that the gray paths must stay out of the rectangle joining $V_{i(K),j(K)}$ to $V_{i,j}$.
As $V_{i(K),j(K)}$ is either on the left face or the top face, this lemma is proved.
\end{proof}

\subsection{Conclusion of the proof}

We have drawn a highest opc for a given boundary configuration on a rectangle in \cref{fig:Highest2}. It is wider than tall, for the ease of drawing the figure.
It may help the reader to keep this picture in mind. Later, we will refer to the features of it which we describe, now.

Suppose that $\rho \in (1,\infty)$ is fixed and that $R$ is even and 
\begin{equation}
\label{eq:rhoR}
\frac{1}{2}\, \rho\, \leq\, \frac{R}{L}\, \leq\, 2\rho\, .
\end{equation}
We consider the rectangle $\Gamma_{2L,R}$.

Let $\mathcal{A}(\Gamma_{2L,R})$ be the event consisting of the set of all $\varsigma \in \mathcal{S}_{N,T}$
such that, taking the boundary conditions of $\varsigma$ on $\partial\Gamma_{2L,R}$,
the highest opc on $\Gamma_{2L,R}$ with these boundary conditions has a 
sequence of $L$ consecutive aligned spins at the edges in some horizontal row of vertical edges
somewhere in $\Gamma_{2L,R}$.
We will prove the following.
\begin{lemma}
\label{lem:Frcng}
Let $\rho>0$ be fixed, $N,T$ be even integers with $N\geq 2L$, $T\geq R$ for $R$ an even integer
satifying \cref{eq:rhoR}.  
Fix a rectangle $\Gamma_{2L,R}$ inside $\mathbb{T}_{N,T}$.

For $\rho$ sufficiently large,
there are constants $c_0,C_0 \in (0,\infty)$, depending on $\rho$ and $\kappa$,
such that
\begin{equation}
\label{eq:AevntBd}
1 - \left\langle \mathbf{1}_{\mathcal{A}(\Gamma_{2L,R})}(\varsigma) \right\rangle_{N,T,\kappa,0}\, \leq\, C_0 e^{-c_0 L^2}\, .
\end{equation}
\end{lemma}

\begin{proof}{Proof of \Cref{lem:Frcng}}

We will choose $\rho$ precisely later. But but the reader should think it is large and fixed, as $L \to \infty$.
Without loss of generality, we will assume that $\Gamma_{2L,R}$ consists of the vertices
$V_{ij}$ for $i=1,\dots,2L$ and $j=1,\dots,R$.
Then  the edges in $\partial \Gamma_{2L,R}$ are 
$E^h_{0,j}$ and $E^h_{2L,j}$ for $j=1,\dots,R$
and $E^v_{i,0}$ and $E^v_{i,R}$ for $i=1,\dots,2L$.

We are going to identify several "good" events $G_i$ which are disjoint subsets of
$\mathcal{A}(\Gamma_{2L,R})$.
We will also identify several events $B_i$ which cover the complement of the good events, and whose probability will be small.

{\underline{\em Step 1:}}
Consider the bottom row of vertical edges, $\{E^v_{i,0}\, :\, i=1,\dots,2L\}$, which comprises the bottom face in $\partial \Gamma_{2L,R}$.
Let us consider the second half of these edges: $E^v_{i,0}$ for $i \in \{L+1,\dots,2L\}$.
$$
G_1\, =\, \{\varsigma \in \mathcal{S}_{N,T}\, :\, \text{all edges $E^v_{i,0}$ are gray for $i=L+1,\dots,2L$}\}\, .
$$
Note $G_1 \subseteq \mathcal{A}(\Gamma_{2L,R})$.

On $\varsigma \in G_1^c$ there is at least one $i \in \{L+1,\dots,2L\}$ such that
$E^v_{i,0}$ is colored black.
Let $i_1$ be the smallest such $i$.
and consider the opc path $\gamma$ connected to this black edge.
We define
$$
B_1\, =\, \{\varsigma \in \mathcal{S}_{N,T}\, :\, \text{$\gamma$ passes through
$E^v_{i_1,R/2}$}\}\, .
$$
Note that on $B_1$ there are no black paths in the bottom half of $\Gamma_{2L,R}$ which cross the thinner direction of the rectangle, $2L$.

{\underline{\em Step 2:}}
Let us now consider $\varsigma \in [G_1\cup B_1]^c$ and $\gamma$ as in Step 1.
Then there is some $j_1 \in \{1,\dots,R/2\}$, such that $\gamma$ passes through
$E^{v}_{i_1,j_1}$, and $\gamma$ turns right there.  We observe that $V_{i_1,j_1}$ is necessarily either of type 4 or 6 (cf. \Cref{fig:trans}).
We let
$$
G_2\, =\,  \{\varsigma \in \mathcal{S}_{N,T}\, :\, \text{$E^v_{i,j_1}$ is gray for each $i\in\{L+1,\dots,2L\}$}\}\, .
$$
Once again,
$G_2 \subseteq \mathcal{A}(\Gamma_{2L,R})$.

{\underline{\em Step 3:}}
Finally consider the case $\varsigma \in [G_1 \cup B_1 \cup G_2]^c$ with and $\gamma$ as in Steps 1,2.
In this situation there is, for some $i_2 \in \{L+1,\dots,2L\}$
and $j_2 \in \{1,\dots,j_1\}$, a vertex $V_{i_2,j_2}$ which is of type 4 or 5.

To see this, for a contradiction consider the case observe that it does not happen.  Then $\gamma$ must turn right at $V_{i_1,j_1}$ and then go straight to $V_{2L,j_1}$.
Otherwise the path would have to turn up at some $V_{i_2,j_1}$ (since it starts going right at $V_{i_1,j_1}$)
and then this vertex would be of type 4 or 5.
This forces all edges $E^v_{i,j_1}$ for $i \in \{i_1,\dots,2L\}$ to be gray, since the straight horizontal segment of $\gamma$ blocks any other black edges from crossing it.
Note, by our reductions, that all edges $E^v_{i,j_1}$ for $i \in \{L+1,\dots,i_1-1\}$ are also gray.  This means $\varsigma \in G_1 \cup B_1 \cup G_2$, a contradiction.

Returning to the main argument, consider the vertex $V_{i_2,j_2}$.
Then because of \Cref{lem:blockades} one of two final events $G_3$ or $B_2$ must occur.  Here
$$
G_3\, =\,  \{\varsigma \in \mathcal{S}_{N,T}\, :\, \text{$V_{i,j_2}$ is of type 4, for each $i\in\{1,\dots,i_2-1\}$}\}\, ,
$$
and
$$
B_2\, =\,  \{\varsigma \in \mathcal{S}_{N,T}\, :\, \text{$V_{i_2,j}$ is of type 4, for each $j\in\{j_2+1,\dots,R\}$}\}\, .
$$
In $G_3$ we have that $E^v_{i,j_2}$ is colored black for all $i \in \{1,\dots,i_2-1\}$.
But since $i_2 \in \{L+1,\dots,2L\}$, this means that 
$\mathcal{I}_L = \{E^v_{i,j_2}\, :\, i=1,\dots,L\}$
has all edges black.
So $G_3 \subseteq \mathcal{A}(\Gamma_{2L,R})$.

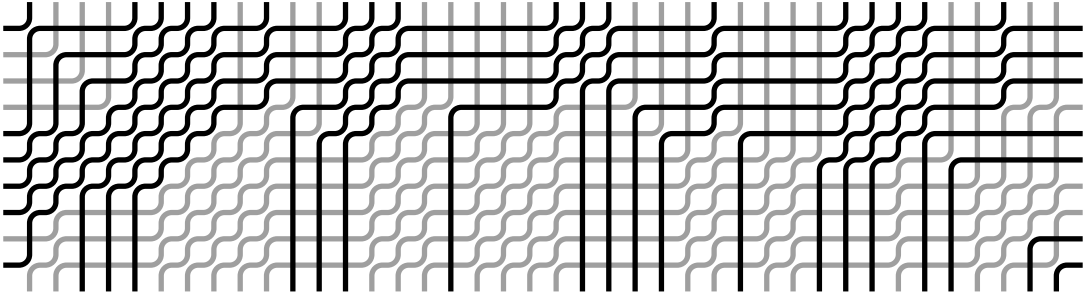
\begin{figure}
\begin{center}
\begin{tikzpicture}[xscale=0.35,yscale=0.35,very thin]
\draw[rounded corners, line width=2pt,gray,nearly opaque] 
(0,9) -- (2,9) -- (2,11)
(0,8) -- (3,8) -- (3,11)
(0,7) -- (4,7) -- (4,11)
(0,2) -- (2,2) -- (2,3) -- (6,3) -- (6,4) -- (7,4) -- (7,5) -- (8,5) -- (8,6) -- (9,6) -- (9,11) 
(1,0) -- (1,1) -- (2,1) -- (2,2) -- (6,2) -- (6,3) -- (7,3) -- (7,4) -- (8,4) -- (8,5) -- (9,5) -- (9,6) -- (10,6) -- (10,7) -- (11,7) -- (11,11)
(2,0) -- (2,1) -- (6,1) -- (6,2) -- (7,2) -- (7,3) -- (8,3) -- (8,4) -- (9,4) -- (9,5) -- (10,5) -- (10,6) -- (12,6) -- (12,11)
(6,0) -- (6,1) -- (7,1) -- (7,2) -- (8,2) -- (8,3) -- (9,3) -- (9,4) -- (10,4) -- (10,5) -- (14,5) -- (14,6) -- (15,6) -- (15,7) -- (16,7) -- (16,11)
(7,0) -- (7,1) -- (8,1) -- (8,2) -- (9,2) -- (9,3) -- (10,3) -- (10,4) -- (14,4) -- (14,5) -- (15,5) -- (15,6) -- (16,6) -- (16,7) -- (17,7) -- (17,11)
(8,0) -- (8,1) -- (9,1) -- (9,2) -- (10,2) -- (10,3) -- (14,3) -- (14,4) -- (15,4) -- (15,5) -- (16,5) -- (16,6) -- (18,6) -- (18,11)
(9,0) -- (9,1) -- (10,1) -- (10,2) -- (14,2) -- (14,3) -- (15,3) -- (15,4) -- (16,4) -- (16,5) -- (18,5) -- (18,6) -- (19,6) -- (19,11)
(10,0) -- (10,1) -- (14,1) -- (14,2) -- (15,2) -- (15,3) -- (16,3) -- (16,4) -- (18,4) -- (18,5) -- (19,5) -- (19,6) -- (20,6) -- (20,11)
(14,0) -- (14,1) -- (15,1) -- (15,2) -- (16,2) -- (16,3) -- (18,3) -- (18,4) -- (19,4) -- (19,5) -- (20,5) -- (20,6) -- (21,6) -- (21,7) -- (24,7) -- (24,11)
(15,0) -- (15,1) -- (16,1) -- (16,2) -- (18,2) -- (18,3) -- (19,3) -- (19,4) -- (20,4) -- (20,5) -- (21,5) -- (21,6) -- (25,6) -- (25,11)
(16,0) -- (16,1) -- (18,1) -- (18,2) -- (19,2) -- (19,3) -- (20,3) -- (20,4) -- (21,4) -- (21,5) -- (26,5) -- (26,11)
(18,0) -- (18,1) -- (19,1) -- (19,2) -- (20,2) -- (20,3) -- (21,3) -- (21,4) -- (26,4) -- (26,5) -- (27,5) -- (27,6) -- (28,6) -- (28,11)
(19,0) -- (19,1) -- (20,1) -- (20,2) -- (21,2) -- (21,3) -- (26,3) -- (26,4) -- (27,4) -- (27,5) -- (29,5) -- (29,11)
(20,0) -- (20,1) -- (21,1) -- (21,2) -- (26,2) -- (26,3) -- (27,3) -- (27,4) -- (29,4) -- (29,5) -- (30,5) -- (30,11)
(21,0) -- (21,1) -- (26,1) -- (26,2) -- (27,2) -- (27,3) -- (29,3) -- (29,4) -- (30,4) -- (30,5) -- (31,5) -- (31,11)
(26,0) -- (26,1) -- (27,1) -- (27,2) -- (29,2) -- (29,3) -- (30,3) -- (30,4) -- (34,4) -- (34,5) -- (36,5) -- (36,11)
(27,0) -- (27,1) -- (29,1) -- (29,2) -- (30,2) -- (30,3) -- (34,3) -- (34,4) -- (37,4) -- (37,11)
(29,0) -- (29,1) -- (30,1) -- (30,2) -- (34,2) -- (34,3) -- (37,3) -- (37,4) -- (38,4) -- (38,7) -- (39,7) -- (39,11) 
(30,0) -- (30,1) -- (34,1) -- (34,2) -- (37,2) -- (37,3) -- (38,3) -- (38,4) -- (39,4) -- (39,7) -- (40,7) -- (40,11)
(34,0) -- (34,1) -- (37,1) -- (37,2) -- (38,2) -- (38,3) -- (39,3) -- (39,4) -- (40,4) -- (40,7) -- (41,7)
(37,0) -- (37,1) -- (38,1) -- (38,2) -- (39,2) -- (39,3) -- (40,3) -- (40,4) -- (41,4)
(38,0) -- (38,1) -- (40,1) -- (40,3) -- (41,3)
;
\draw[rounded corners, line width=2pt]
(0,10) -- (1,10) -- (1,11)
(0,6) -- (1,6) -- (1,10) -- (5,10) -- (5,11)
(0,5) -- (1,5) -- (1,6) -- (2,6) -- (2,9) -- (5,9) -- (5,10) -- (6,10) -- (6,11)
(0,4) -- (1,4) -- (1,5) -- (2,5) -- (2,6) -- (3,6) -- (3,8) -- (5,8) -- (5,9) -- (6,9) -- (6,10) -- (7,10) -- (7,11)
(0,3) -- (1,3) -- (1,4) -- (2,4) -- (2,5) -- (3,5) -- (3,6) -- (4,6) -- (4,7) -- (5,7) -- (5,8) -- (6,8) -- (6,9) -- (7,9) -- (7,10) -- (8,10) -- (8,11)
(0,1) -- (1,1) -- (1,3) -- (2,3) -- (2,4) -- (3,4) -- (3,5) -- (4,5) -- (4,6) -- (5,6) -- (5,7) -- (6,7) -- (6,8) -- (7,8) -- (7,9) -- (8,9) -- (8,10) -- (10,10) -- (10,11)
(3,0) -- (3,4) -- (4,4) -- (4,5) -- (5,5) -- (5,6) -- (6,6) -- (6,7) -- (7,7) -- (7,8) -- (8,8) -- (8,9) -- (10,9) -- (10,10) -- (13,10) -- (13,11)
(4,0) -- (4,4) -- (5,4) -- (5,5) -- (6,5) -- (6,6) -- (7,6) -- (7,7) -- (8,7) -- (8,8) -- (10,8) -- (10,9) -- (13,9) -- (13,10) -- (14,10) -- (14,11)
(5,0) -- (5,4) -- (6,4) -- (6,5) -- (7,5) -- (7,6) -- (8,6) -- (8,7) -- (10,7) -- (10,8) -- (13,8) -- (13,9) -- (14,9) -- (14,10) -- (15,10) -- (15,11)
(11,0) -- (11,7) -- (13,7) -- (13,8) -- (14,8) -- (14,9) -- (15,9) -- (15,10) -- (21,10) -- (21,11)
(12,0) -- (12,6) -- (13,6) -- (13,7) -- (14,7) -- (14,8) -- (15,8) -- (15,9) -- (21,9) -- (21,10) -- (22,10) -- (22,11)
(13,0) -- (13,6) -- (14,6) -- (14,7) -- (15,7) -- (15,8) -- (21,8) -- (21,9) -- (22,9) -- (22,10) -- (23,10) -- (23,11)
(17,0) -- (17,7) -- (21,7) -- (21,8) -- (22,8) -- (22,9) -- (23,9) -- (23,10) -- (27,10) -- (27,11)
(22,0) -- (22,8) -- (23,8) -- (23,9) -- (27,9) -- (27,10) -- (32,10) -- (32,11)
(23,0) -- (23,8) -- (27,8) -- (27,9) -- (32,9) -- (32,10) -- (33,10) -- (33,11)
(24,0) -- (24,7) -- (27,7) -- (27,8) -- (32,8) -- (32,9) -- (33,9) -- (33,10) -- (34,10) -- (34,11)
(25,0) -- (25,6) -- (27,6) -- (27,7) -- (32,7) -- (32,8) -- (33,8) -- (33,9) -- (34,9) -- (34,10) -- (35,10) -- (35,11)
(28,0) -- (28,1) -- (28,6) -- (32,6) -- (32,7) -- (33,7) -- (33,8) -- (34,8) -- (34,9) -- (35,9) -- (35,10) -- (38,10) -- (38,11)
(31,0) -- (31,5) -- (32,5) -- (32,6) -- (33,6) -- (33,7) -- (34,7) -- (34,8) -- (35,8) -- (35,9) -- (38,9) -- (38,10) -- (41,10)
(32,0) -- (32,5) -- (33,5) -- (33,6) -- (34,6) -- (34,7) -- (35,7) -- (35,8) -- (38,8) -- (38,9) -- (41,9)
(33,0) -- (33,5) -- (34,5) -- (34,6) -- (35,6) -- (35,7) -- (38,7) -- (38,8) -- (41,8)
(35,0) -- (35,1) -- (35,6) -- (41,6)
(36,0) -- (36,5) -- (41,5)
(39,0) -- (39,2) -- (41,2)
(40,0) -- (40,1) -- (41,1)
;
\end{tikzpicture}
\caption{
\label{fig:Highest2}
A highest opc for a ``typical'' boundary configuration on a ``long, thin'' rectangle. There is no straight
segment spanning the long dimension (which would block black paths from crossing the small dimension),
nor any black blockade spanning the long dimension (which would block gray paths from crossing the small dimension). Roughly, heuristically, the length scale for straight segments as well as for blockades is set
by the smaller dimension. In the actual argument we have a tall, thin rectangle instead.}
\end{center}
\end{figure}

In \Cref{fig:Highest2} we have shown a highest OPC for a ``typical'' boundary configuration.
$B_1$ and $B_2$ do not occur in this particular example.

{\underline{\em Bounding Probabilities of Bad Events:}}

We will use the following lemma to bound the probabilities of bad events.
For any $L \leq \min\{N/2,T/2\}$, consider a prescribed interval $\mathcal{I}(L)$ of $L$ consecutive edges in $\mathbb{T}_{N,T}$:
either $L$ consecutive horizontal edges in a vertical column,  or else $L$ consecutive vertical edges in a horizontal row.
The spins on $\mathcal{I}(L)$ are all in $\{\leftarrow,\rightarrow\}$ if $\mathcal{I}(L)$ is a vertical interval of horizontal edges,
and are all in $\{\uparrow,\downarrow\}$ if $\mathcal{I}(L)$ is a horizontal interval of vertical edges.
Let $\mathcal{D}_{\mathcal{I}(L),\delta}$
be the event that these spins have a density of $\leftarrow$ spins or $\uparrow$ spins in the respective cases in
$[0,\delta)\cup(1-\delta,1]$.
\begin{lemma}
\label{lem:dens6}
For any $\kappa \in \R$,
there exists a $\delta>0$ and constants $c,C \in (0,\infty)$ such that
$$
\left\langle \mathbf{1}_{\mathcal{D}_{\mathcal{I}(L),\delta}}(\varsigma) \right\rangle_{N,T,\kappa,0}\, \leq\, C e^{-cL^2}\, .
$$
\end{lemma}

We now explain how this lemma applies in the present case.  The proof of \Cref{lem:dens6} follows the conclusion of the present proof.

On $B_1$ there is a straight line segment in $\gamma$
from $V_{i_1,0}$ to $V_{i_1,R/2}$ and it blocks any black path from crossing this line.
Then for every $j \in \{1,\dots,R/2\}$, the black path incident to $E^h_{2L,j}$ must terminate at $E^v_{i,0}$ for some $i \in \{i_1+1,\dots,2L\}$.
In particular, there are at most $L$ such paths and we conclude at most $L$ of the $R/2$ edges $E^h_{2L,j}$ are black.

On $B_2$ there is a black blockade running along the vertices
$V_{i_2,j}$ for $j \in \{j_2+1,\dots,R\}$.  As a consequence,  for every $j \in \{(R/2)+1,\dots,R\}$ such that  $E^h_{0,j}$  is a grey edge, the corresponding grey path must terminate
on the upper face of $\Gamma_{2L,R}$.
This means that there are at most $2L$ such paths and further that, among the $(R/2)$ edges $E^h_{0,j}$ for $j \in \{(R/2)+1,\dots,R\}$,
at most $2L$ of them may be gray.

By choosing $\rho$ appropriately, we may make the density $4L/R$ small enough to apply 
\Cref{lem:dens6}.
In particular, given the $\delta$ from \Cref{lem:dens6}, let us choose $\rho = 8/\delta$. Then $4/\rho = \delta/2$
so that we may definitely apply \Cref{lem:dens6}.
Then we see that the two bad scenarios may be included in a subset of events where \Cref{lem:dens6}
applies with $L$ replaced by $R/2$ or $R$, which is at least $\frac{1}{4}\, \rho L$.
So $c_0 = \frac{1}{4}\, \rho c$, and $C_0 = 3 C$.
\end{proof}

\begin{proof}[Proof of \Cref{lem:dens6}]
In \Cref{sec:UB6} we gave an argument that applied to horizontal intervals.
The argument was based on having a low density of misaligned spins.
We then showed that for entropic reasons, if the density of misaligned spins was low enough,
then it was better (in terms of Gibbs probability) to allow sets of configurations with slightly higher density.

Nothing in that proof  required the density on the initial interval to be precisely $\delta=0$.
Indeed, even though the density on the initial interval in that section was taken to be zero, in our argument we saw that, as one goes up row-by-row in the block,
the density may become $\epsilon/2$ in the middle row, if the height of the block is $\ell = \epsilon L$.
That is why we took $\ell = \epsilon L$ with $\epsilon>0$ small but fixed.
If we start with an initial density of $\delta$, then the maximum density will just be $\delta + \frac{1}{2}\epsilon$ instead of $\frac{1}{2}\epsilon$.
If we choose $\delta$ and $\epsilon$ sufficiently small but fixed, then the same argument goes through, essentially unchanged.

Next, to deal with vertical intervals instead of just horizontal intervals, we note that the 
model is both reflection positive in the vertical and horizontal directions. 
There is a difference in how we treat the horizontal and vertical dimensions.
For instance, in $\langle \cdot \rangle_{N,T,\kappa,0}$ we condition on having an equal number of $\uparrow$ and $\downarrow$ spins on each
horizontal row of vertical edges.
We do not make any conditioning for the vertical columns of horizontal edges.
But this was just to make contact with the XXZ model.
This model is reflection positive in both directions, and that is all that was used in \Cref{sec:UB6}.
(Indeed, we did use reflection positivity in both directions in \Cref{sec:UB6} in order to disseminate the initial reference block
in both directions.)

We did need the following condition: $L \leq N/2$ and $\ell = \epsilon L \leq T/2$.
The reason we needed $L\leq N/2$ is obvious: we needed to reflect the initial block at least 1 full time in the horizontal direction.
We took the height of the block to be $\ell = \epsilon L$ for some $\epsilon>0$ small, but fixed.
We may easily demand that $\epsilon \leq 1$.
Then we are okay as long as we have $L\leq \min\{N/2,T/2\}$.
This is  a condition which is unchanged under switching the vertical and horizontal coordinates.
So the same argument will work for an interval $\mathcal{I}(L)$ consisting of consecutive  edges in a vertical column
of horizontal edges.

Finally, we note that we can replace the condition of having density in $[0,\delta)$ of one type of arrow with a density
in $(1-\delta,1]$ for that same type of arrow, due to spin-flip symmetry of the six-vertex model.
\end{proof}

\subsubsection{Final argument}
We will use the subset bound.

For each $\varsigma \in \mathcal{A}(\Gamma_{2L,R})$, we may alter $\varsigma$ in the interior of $\Gamma_{2L,R}$
it by the highest opc relative to the original boundary configuration of $\varsigma$.
Let us call the new configuration $\varsigma' \in \mathcal{A}(\Gamma_{2L,R})$,
and let us define the function $\mathsf{F} : \mathcal{A}(\Gamma_{2L,R}) \to \mathcal{A}(\Gamma_{2L,R})$
by $\mathsf{F}(\varsigma) = \varsigma'$.
Then, for every $\varsigma \in \mathcal{A}(\Gamma_{2L,R})$,
\[
|\mathsf{F}^{-1}(\{\varsigma\})|\, \leq\, 2^{|\mathcal{E}(\Gamma_{2L,R})|}\leq 2^{8 \rho L^2}.
\]
Either $\varsigma'$ or $-\varsigma'$ does satisfy the condition
\begin{itemize}
\item $\forall i \in \mathbb{B}_L$, $\varsigma(E^{v}_{i+i(0),j(0)})=+1$,
\end{itemize}
for some $i(0),j(0)$ with $V_{i(0),j(0)} \in \Gamma_{2L,R}$.
Let $\mathsf{G} :\mathcal{A}(\Gamma_{2L,R}) \to \mathcal{A}(\Gamma_{2L,R})$
be the mapping, where $\mathsf{G}(\varsigma)$ is either $\varsigma' = \mathsf{F}(\varsigma)$ or $-\varsigma'$,
whichever one satisfies the condition above.
Finally, let $\mathsf{H} :\mathcal{A}(\Gamma_{2L,R}) \to \mathcal{A}(\Gamma_{2L,R})$
be the mapping where, given $\varsigma$, we first take $\mathsf{G}(\Gamma_{2L,R})$,
and then we shift  by $-i(0)$ in the $x$-direction and $1-j(0)$ in the $y$-direction.
An easy calculation shows
\[
|\mathsf{H}^{-1}(\{\varsigma\})|\, \leq\, 8 \rho L e^{8 \rho \ln(2) L^2}\, ,
\]
for every $\varsigma \in \mathcal{A}(\Gamma_{2L,R})$.

Let us define $\mathcal{A}' = \{\varsigma\, :\, \forall i \in \mathbb{B}_L$, $\varsigma(E^{v}_{i,1})=+1\}$.
Then $\mathsf{H}(\varsigma) \in \mathcal{A}'$ for every $\varsigma \in \mathcal{A}(\Gamma_{2L,R})$.
We have
\begin{multline}
\left\langle \mathbf{1}_{\mathcal{A}'}(\varsigma) \right\rangle_{N,T,\kappa,0}\,
\geq\, (8 \rho L)^{-1} e^{-8 \rho \ln(2) L^2} e^{-8 \rho |\kappa| L^2}
\left\langle \mathbf{1}_{\mathcal{A}(\Gamma_{2L,R})}(\varsigma) \right\rangle_{N,T,\kappa,0}\,
\\
\geq (8 \rho L)^{-1} e^{-8 \rho \ln(2) L^2} e^{-8 \rho |\kappa| L^2} (1 - C_0 e^{-c_0 L^2})\, .
\end{multline}
It is easily seen from this that there is some $C_1$ and $c_1$ such that
$$
\left\langle \mathbf{1}_{\mathcal{A}'}(\varsigma) \right\rangle_{N,T,\kappa,0}\,
\geq\, C_1 e^{-c_1 L^2}\, .
$$
Combined with \Cref{cor:Q1d} and the definitions in \cref{eq:BGdef} and \cref{eq:BGdef2}, this proves \cref{E:Ld3}.

\appendix
\section{Appendix}
\label{S:App}
\subsection{Trace inequalities: Generalized H\"older Inequality, first version}
\begin{theorem} \label{thm:Holder}
Suppose that $H$ is a self-adjoint operator and define $Z(\beta) = \textrm{Tr}[e^{-\beta H}]$.
Then for any $A \in \mathcal{B}(\Hil)$
$$
|Z(\beta)^{-1} \tr[A e^{-\beta H}]|\, \leq\, \left(Z(\beta)^{-1} \tr\left[(e^{-\beta H/4n}Ae^{-\beta H/2n}A^*e^{-\beta H/4n})^{n}\right]\right)^{1/2n}\, ,
$$
for each $n \in \{1,2,\dots\}$. In particular, if $A = A^*$ then by cyclicity of the trace this implies
$$
|Z(\beta)^{-1} \tr[A e^{-\beta H}]|\, \leq\, \left(Z(\beta)^{-1} \tr\left[(Ae^{-\beta H/2n})^{2n}\right]\right)^{1/2n}\, .
$$
\end{theorem}
Fr\"ohlich and Lieb proved the generalized H\"older inequality in \cite{FrohLieb}.
In the context that we are proving, they showed that, defining
a multilinear form,
$\alpha : (\mathcal{B}(\Hil))^{2n} \to \C$,
$$
\alpha(A_1,\dots,A_{2n})\, =\, \tr\left[(e^{-\beta H/4n}A_1e^{-\beta H/4n})\cdots (e^{-\beta H/4n}A_{2n}e^{-\beta H/4n})\right]\, ,
$$
then
$$
|\alpha(A_1,\dots,A_{2n})|^{2n}\,
\leq\, \prod_{k=1}^{2n} \alpha(A_k,A_k^*,A_k,A_k^*,\dots,A_k,A_k^*)\, .
$$
But we will only need the special case we have proved, which has $A_1=A$ and $A_k=I$ for all
$k=2,\dots,2n$.
Cyclicity of the trace is the property that, for all $A,B \in \mathcal{B}(\Hil)$,
\begin{equation}
\label{eq:cot}
\tr[AB]\, =\, \tr[BA]\, .
\end{equation}
Another important property is the Cauchy-Schwarz inequality. The bilinear form,
$\langle \cdot\, ,\ \cdot \rangle_{\mathrm{HS}} : \mathcal{B}(\Hil) \times \mathcal{B}(\Hil) \to \C$ defined by
$\langle A,B \rangle_{\mathrm{HS}} = \tr[A^*B]$ is a positive-definite form known as the Hilbert-Schmidt inner-product.
The Cauchy-Schwarz inequality says
\begin{equation}
\label{ineq:CS}
|\tr[A^*B]|^2\, \leq\, \tr[A^*A] \tr[B^*B]\, .
\end{equation}
For Frohlich and Lieb's multilinear form (\ref{ineq:CS}) implies
\begin{equation}
\label{ineq:CSref}
|\alpha(A_1,\dots,A_{2n})|^2\, \leq\, \alpha(A_1,\dots,A_n,A_n^*,\dots,A_1^*) \alpha(A_{2n}^*,\dots,A_{n+1}^*,A_{n+1},\dots,A_{2n})\, .
\end{equation}
This is {\em reflection positivity}. Later, we will see
that it applies in more general contexts.
Another important property of $\alpha$ is {\em cyclicity}. Defining $\tau : (\mathcal{B}(\Hil))^{2n} \to (\mathcal{B}(\Hil))^{2n}$
by $\tau(A_1,A_2,\dots,A_{2n}) = (A_2,\dots,A_{2n},A_1)$, equation (\ref{eq:cot}) gives
\begin{equation}
\label{eq:cyc}
\alpha(A_1,A_2,\dots,A_{2n})\, =\, \alpha(\tau(A_1,A_2,\dots,A_{2n}))\,
=\, \alpha(A_2,\dots,A_{2n},A_1)\, .
\end{equation}
Let us now prove Theorem \ref{thm:Holder}.
\begin{proof}
For $k=0,1,\dots,2n$, let us define a complex number $c_k$ as
$$
c_{k}\, =\, \alpha(A_{k,1},\dots,A_{k,n})\, ,\ \text{ where }\
A_{k,j}\, =\, \begin{cases} A & \text{ if $j\leq k$ and $k-j$ is even,}\\
A^* & \text{ if $j\leq k$ and $k-j$ is odd,}\\
I & \text{ if $j>k$.}
\end{cases}
$$
Then, considering the terms on the right hand side of (\ref{ineq:CSref}), and using $\tau$ defined in (\ref{eq:cyc}),
\begin{gather*}
(A_{k,1},\dots,A_{k,n},A_{k,n}^*,\dots,A_{k,1}^*)\, =\, \begin{cases} \tau^{-k}(A_{2k,1},\dots,A_{2k,2n}) & \text{ if $k\leq n$,}\\
\tau^{-k}(A_{2n,1}\dots,A_{2n,2n}) & \text{ if $n\leq k\leq 2n$,}
\end{cases}\\
(A_{k,2n}^*,\dots,A_{k,n+1}^*,A_{k,n+1},\dots,A_{k,2n})\,
=\, \begin{cases} (A_{0,1},\dots,A_{0,2n}) & \text{ if $k\leq n$,}\\
\tau^{2n-k}(A_{2k-2n,1},\dots,A_{2k-2n,2n}) & \text{ if $n\leq k\leq 2n$.}
\end{cases}
\end{gather*}
Therefore, using this with (\ref{ineq:CSref}) and (\ref{eq:cyc}),
\begin{gather}
\label{ineq:kleqn}
|c_k|^2 \leq c_{2k} c_0\, ,\text { when $k\leq n$; }\\
\label{ineq:kgeqn}
|c_k|^2 \leq c_{2n} c_{2k-2n}\, ,\ \text{ when $n\leq k\leq 2n$.}
\end{gather}
Then a maximum principle applies to $\gamma_{k} = (c_{k}/c_0)^{1/(k)}$ defined for $k=2,4,\dots,2n$.
Equation (\ref{ineq:kleqn}) implies
that $\gamma_{k} \leq \gamma_{2k}$ for all even $k$ such that $k\leq n$.
Therefore, $\max_{k=2,\dots,2n} \gamma_{2k}$ must equal $\gamma_{\kappa}$
for some even $\kappa$ such that $n\leq \kappa\leq 2n$.
Then (\ref{ineq:kgeqn}) implies
\begin{equation}
\label{ineq:writeit}
\gamma_{\kappa}^{2\kappa}\, \leq\, \gamma_{2n}^{2n} \gamma_{2\kappa-2n}^{2\kappa-2n}\, .
\end{equation}
But $\gamma_{2\kappa-2n} \leq \max_{k=2,\dots,2n} \gamma_{k} = \gamma_{\kappa}$.
So (\ref{ineq:writeit}) implies $\gamma_{\kappa}^{2\kappa} \leq \gamma_{2n}^{2n} \gamma_{\kappa}^{2\kappa-2n}$ which in turn implies
$\gamma_{2n} \geq \gamma_{\kappa}$. In other words $\gamma_{2n} \geq \gamma_{k}$ for all $k=2,4,\dots,2n$.
So finally by (\ref{ineq:kleqn}) again
$$
|c_1|\, \leq\, c_2^{1/2} c_0^{1/2}\, =\, c_0 \gamma_2\, \leq\, c_0 \gamma_{2n}\, =\, c_0 (c_{2n}/c_0)^{1/(2n)}\,
$$
which is the desideratum, writing out $c_0$, $c_1$ and $c_{2n}$.
\end{proof}
The full generality of Frohlich and Lieb's generalized H\"older's inequality will be described later.
The special case of the inequality just proved allows for a special application to reflections
in (imaginary) time.
\subsection{Reflection Positivity for Quantum Spin Systems}
\label{S:RP}
Suppose that $N$ is even.
Then we may decompose $\mathbb{T}_N^d$ into two halves:
\begin{align*}
\Lambda^+\, &=\, \big\{(x_1,\dots,x_d)\, :\, x_1,\dots,x_d \in \{0,\dots,N-1\}\, ,\ x_\nu \in \{(N/2),\dots,N-1\}\big\}\, ,\\
\Lambda^-\, &=\, \big\{(x_1,\dots,x_d)\, :\, x_1,\dots,x_d \in \{0,\dots,N-1\}\, ,\ x_\nu \in \{0,\dots,(N/2)-1\}\big\}\, .
\end{align*}
Moreover, let us define
$$
\Hil_+\, =\, \ell^2(\Omega(\Lambda^+))\quad \text{ and } \quad \Hil_-\, =\, \ell^2(\Omega(\Lambda_-))\, ,
$$
where as usual $\Omega(\Lambda)$ is the set of all function $\sigma = (\sigma_x)_{x \in \Lambda}$.
We may identify $\Hil(\mathbb{T}^d_N) = \Hil_- \otimes \Hil_+$.
For example, the simple tensor product may be understood as follows: suppose that $f_-$ and $f_+$ are functions
in $\Hil_-$ and $\Hil_+$, respectively.
Then we may define $(f_-\otimes f_+) \in \Hil(\mathbb{T}^d_N)$ as follows: given $\sigma \in \Omega(\mathbb{T}^d_N)$
define $(\sigma)^-$ and $(\sigma)^+$ to be the restrictions:
$$
(\sigma)^{\pm}\, =\, (\sigma_x)_{x \in \Lambda_{\pm}}\, .
$$
Then
$$
(f_- \otimes f_+)(\sigma)\, =\, f_-((\sigma)^-) f_+((\sigma)^+)\, .
$$
Let us define a reflection $R : \Lambda_+ \to \Lambda_-$ by
$$
R(x_1,\dots,x_{\nu-1},x_{\nu},x_{\nu+1},\dots,x_d)\,
=\, (x_1,\dots,x_{\nu-1},N-1-x_{\nu},x_{\nu+1},\dots,x_d)\, .
$$
Then we define an isomorphism $\mathcal{R} : \Omega(\Lambda_-) \to \Omega(\Lambda_+)$
as
$$
\mathcal{R}((\sigma_x)_{x \in \Lambda_-})\, =\, (\tau_x)_{x \in \Lambda_+}\, ,\qquad
\tau_x\, =\, \sigma_{R(x)}\, .
$$
Let us define a unitary transformation $F : \Hil_+ \to \Hil_-$ by
$$
Ff((\sigma_x)_{x \in \Lambda_-})\, =\, f((\tau_x)_{x \in \Lambda_+})\, ,\qquad
(\tau_x)_{x \in \Lambda_+}\, =\, \mathcal{R}((\sigma_x)_{x \in \Lambda_-}))\, .
$$
Finally, we define a $C^*$-algebra ismorphism $\mathcal{F} : \mathcal{B}(\Hil_-) \to \mathcal{B}(\Hil_+)$
by
$$
\mathcal{F}(A)\, =\, F A F^*\, .
$$
Let us define $\mathcal{A}^-$ to be the $C^*$ subalgebra of $\mathcal{B}(\Hil(\mathbb{T}^d_N))$
which is equivalent to set of all operators of the form $A \otimes I_{\Hil_+}$,
$A \in \mathcal{B}(\Hil_-)$.
Note that this is a $C^*$-subalgebra because it is closed under all the algebra operations,
as well as the adjoint: $(A\otimes I_{\Hil_+})^* = A^* \otimes I_{\Hil_+}$,
and $A^*$ is in $\mathcal{B}(\Hil_-)$ for each $A$ in $\mathcal{B}(\Hil_-)$.
We define $\mathcal{A}^+$ similarly as the set of all operators $I_{\Hil_-} \otimes A$
for $A \in \mathcal{B}(\Hil_+)$.
We use the symbol $\mathscr{F}$ for the $C^*$ algebra isomorphism between $\mathcal{A}^-$
and $\mathcal{A}^+$:
$$
\mathscr{F}(A \otimes I_{\Hil_+})\, =\, I_{\Hil_-} \otimes \mathcal{F}(A)\, .
$$
An important consideration for the further part of the definition will be the introduction of an orthonormal basis. We take the basis previously
stated as the appropriate orthonormal basis.
Given an operator $A : \Hil_{-} \to \Hil_{-}$ we will say that it is ``real'' if
$$
\langle \chi_{\{\sigma\}}, A \chi_{\{\sigma'\}} \rangle \in \R\, ,
$$
for all $\sigma,\sigma' \in \Omega(\Lambda_-)$.
Note that this does depend on the choice of orthonormal basis, and even depends on the choice of the phase for the ortho-normal basis.
This is a definition with less flexibility than one usually associates to the framework of quantum mechanics,
as defined on complex Hilbert spaces.
But it is needed.
We define $\mathcal{B}_{\R}(\Hil_-)$ to be the set of all ``real'' operators $A$ in $\mathcal{B}(\Hil_-)$.
Similarly, we say that $A \in \mathcal{B}(\Hil_+)$ is ``real'' if
$$
\langle \chi_{\{\sigma\}}, A \chi_{\{\sigma'\}} \rangle \in \R\, ,
$$
for all $\sigma,\sigma' \in \Omega(\Lambda_+)$, and we define
$\mathcal{B}_{\R}(\Hil_+)$ to be the set of all such ``real'' operators.
Finally, we define $\mathcal{A}_{\R}^-$ to be the subspace of operators in $\mathcal{A}^-$ of the form
$A \otimes I_{\mathcal{H}_+}$ for $A$ in $\mathcal{B}_{\R}(\Hil_-)$, and we define
$\mathcal{A}_{\R}^+$, similarly.
It is straightforward to check that $\mathscr{F}$ actually also maps $\mathcal{A}^-_{\R}$ to
$\mathcal{A}^+_{\R}$.
\begin{definition}
A linear functional $\alpha : \mathcal{B}(\Hil(\mathbb{T}^d_N)) \to \C$ is said to be reflection positive
if, for every operator $A \in \mathcal{A}^-_{\R}$ we have
$$
\alpha( A \mathscr{F}(A))\, \geq\, 0\, .
$$
\end{definition}
\begin{lemma}
\label{lem:RP1}
The tracial state, defined as $\langle \cdot \rangle_0 = \tr[\cdot]/\tr[I]$, is reflection positive.
\end{lemma}
\begin{proof}
For every operator $A \in \mathcal{A}^-_{\R}$, there is some $\widetilde{A} \in \mathcal{B}_{\R}(\Hil_-)$
such that $A = \widetilde{A} \otimes I_{\Hil_{+}}$.
Note that then $\tr_{\Hil_{-}}(A)$ is real, because $A$ is ``real,'' and the trace may be calculated
in the basis used for the definition of ``real'' operators.
Then we have the string of identities
\begin{align*}
\tr[A \mathscr{F}(A)]\,
&=\, \tr[\widetilde{A} \otimes \mathcal{F}(\widetilde{A})]\\
&=\, \sum_{\sigma \in \Omega(\mathbb{T}^d_N)}
\langle \chi_{\{\sigma\}}, \widetilde{A} \otimes \mathcal{F}(\widetilde{A}) \chi_{\{\sigma\}} \rangle\\
&=\, \sum_{\sigma \in \Omega(\mathbb{T}^d_N)}
\langle \chi_{\{(\sigma)^-\}}, \widetilde{A} \chi_{\{(\sigma)^-\}} \rangle_{\mathcal{H}_-} \cdot
\langle \chi_{\{(\sigma)^+\}}, \mathcal{F}(\widetilde{A}) \chi_{\{(\sigma)^+\}} \rangle_{\mathcal{H}_+}
\\
&=\, \sum_{\sigma \in \Omega(\mathbb{T}^d_N)}
\langle \chi_{\{(\sigma)^-\}}, \widetilde{A} \chi_{\{(\sigma)^-\}} \rangle_{\mathcal{H}_-} \cdot
\langle \chi_{\{(\sigma)^+\}}, \mathcal{F}(\widetilde{A}) \chi_{\{(\sigma)^+\}} \rangle_{\mathcal{H}_+}
\\
&=\, \sum_{\sigma \in \Omega(\Lambda_-)}
\langle \chi_{\{\sigma\}}, \widetilde{A} \chi_{\{\sigma)\}} \rangle_{\mathcal{H}_-}
\sum_{\sigma \in \Omega(\Lambda_+)}
\langle F^* \chi_{\{\sigma\}}, \widetilde{A} F^* \chi_{\{\sigma\}}\rangle_{\mathcal{H}^+}\, .
\end{align*}
But considering the definition of $F^*$ on $\chi_{\{\sigma\}}$, we realize that (because $R$ is a bijection)
the second factor is equal to the first. In other words,
$$
\tr[A\mathscr{F}(A)]\, =\, \left(\tr_{\mathcal{\Hil_-}}[\widetilde{A}]\right)^2\, .
$$
Because the trace of $\widetilde{A}$ is real, this is nonnegative.
\end{proof}
\begin{definition}
A linear functional $\alpha : \mathcal{B}(\Hil(\mathbb{T}^d_N)) \to \C$
is said to be generalized reflection positive if, for every $n$ and all operators $A_1,\dots,A_n \in \mathcal{A}^-_{\R}$ we have
$$
\alpha( A_1 \mathscr{F}(A_1) A_2 \mathscr{F}(A_2) \cdots A_n \mathscr{F}(A_n)\rangle\, \geq\, 0\, .
$$
\end{definition}
\begin{theorem}
The tracial state is also generalized reflection positive.
\end{theorem}
\begin{proof}
Note that $\mathcal{A}^-$
commutes with $\mathcal{A}^+$.
So we are actually trying to prove
$$
\langle A_1 A_2 \cdots A_n \mathcal{F}(A_1) \mathcal{F}(A_2) \cdots \mathcal{F}(A_n)\rangle\, \geq\, 0\, .
$$
Because $\mathcal{F}$ is a $C^*$-algebra homomorphism, this may be rewritten as
$$
\langle A_1 A_2 \cdots A_n \mathcal{F}(A_n \cdots A_2 A_1)\rangle\, \geq\, 0\, .
$$
Then this reduces to the definition of reflection positivity proved in Lemma \ref{lem:RP1}.
\end{proof}
The key theorem for obtaining reflection positive examples is as follows.
\begin{theorem}
\label{thm:RPex}
Suppose that we have a generalized reflection positive linear functional which we will denote
$\alpha_0 : \mathcal{B}(\Hil(\mathbb{T}^d_N)) \to \C$.
Suppose that for some $K$ there are operators $B,C_1,\dots,C_K\in \mathcal{A}^-_{\R}$
such that a Hamiltonian $H \in \mathcal{B}(\Hil(\mathbb{T}^d_N))$ may be written as
\begin{equation}
\label{eq:Hform}
H\, =\, B + \mathscr{F}(B) - \sum_{k=1}^{K} C_k \mathscr{F}(C_k)\, .
\end{equation}
Then, for each $\beta \in [0,\infty)$, defining the linear functional $\alpha_{\beta}$
as
$$
\alpha_{\beta}(\cdot)\,
=\, \frac{\alpha_{0}(\cdot\, e^{-\beta H})}{\alpha_0(e^{-\beta H})}\, ,
$$
this
is also generalized reflection positive.
\end{theorem}
\begin{proof}
This follows from the Trotter product formula. Given any $A_1,\dots,A_n \in \mathcal{A}^-_{\R}$,
the Trotter product formula implies
\begin{align*}
\alpha_0(A_1\mathscr{F}(A_1)\cdots A_n\mathscr{F}(A_n) e^{-\beta H})\,
&=\,
\lim_{R \to \infty}
\alpha_0\big(A_1\mathscr{F}(A_1)\cdots A_n\mathscr{F}(A_n) \\
&\hspace{-2cm} \big(e^{-(\beta/R)B} \mathscr{F}(e^{-(\beta/R) B})
(I + \frac{\beta}{R}C_1 \mathscr{F}(C_1)) \cdots
(I + \frac{\beta}{R}C_K \mathscr{F}(C_K))\big)^R \big)\, .
\end{align*}
For any fixed $R$ the quantity on the right hand side may be expanded as a nonnegative
combination of terms of the form $\alpha_0(D_1\mathscr{F}(D_1) \cdots D_L \mathscr{F}(D_L))$
for finite $L$'s and operators $D_1,\dots,D_L \in \mathcal{A}^-_{\R}$.
So, since $\alpha_0$ is generalized reflection positive, we see that this term is nonnegative as well.
The property of being nonnegative survives the limit $R \to \infty$.
Therefore,
$$
\alpha_0(A_1\mathscr{F}(A_1)\cdots A_n\mathscr{F}(A_n) e^{-\beta H})\, \geq\, 0\, .
$$
A similar calculation shows $\alpha_0(e^{-\beta H}) \geq 0$. So the theorem follows.
\end{proof}
Since the tracial state is generalized reflection positive, if $H$ satisfies the form
(\ref{eq:Hform}) then the equilibrium state $\langle \cdot \rangle_{\beta}$ is also generalized
reflection positive by this theorem.
\subsection{The XXZ model for $\Delta \leq 0$}
The XXZ Hamiltonian (\ref{eq:PFready}) with $\Delta \leq 0$ is not generalized reflection positive.
But it is unitarily equivalent to a Hamiltonian which is.
Let $U : \Hil_+ \to \Hil_+$ be the unitary transformation on $\Hil_+ = \ell^2(\Omega(\Lambda_+))$
given by
$$
Uf(\sigma)\, =\, f(-\sigma)\, ,
$$
where $-\sigma=(-\sigma_x)_{x \in \Lambda_+}$ for $\sigma = (\sigma_x)_{x \in \Lambda_+}$.
Note that $U=U^*=U^{-1}$.
Importantly, we also have
$$
US_x^{(1)}U\, =\, S_x^{(1)}\, ,\qquad
US_x^{(2)}U\, =\, -S_x^{(2)}\, ,\quad \text{and}\quad
US_x^{(3)}U\, =\, -S_x^{(3)}\, ,
$$
for each $x \in \Lambda_+$.
Let us define $\mathcal{U} = I_{\Hil_-} \otimes U$, which is a unitary transformation on $\Hil(\mathbb{T}^d_N)$ also satisfying $\mathcal{U}=\mathcal{U}^*=\mathcal{U}^{-1}$.
Then, we claim that if $x$ is in $\Lambda_-$ and $y$ is in $\Lambda_+$ then
$$
\mathcal{U}h_{xy}^{\Delta}\mathcal{U}\, =\, - S_x^{(1)} S_y^{(1)} + S_x^{(2)} S_y^{(2)} + \Delta S_x^{(3)} S_y^{(3)}\, .
$$
This is important because $S_x^{(2)}$ and $S_y^{(2)}$ are not real operators because of the presence
of a factor $i$. But $i S_x^{(2)}$ and $i S_y^{(2)}$ are real operators.
Therefore, we may write
$$
Uh_{xy}^{\Delta}U\, =\, -S_x^{(1)} S_y^{(1)} - (iS_x^{(2)}) (iS_y^{(2)}) + \Delta S_x^{(3)} S_y^{(3)}\, .
$$
If we have the condition $\Delta \leq 0$ then this means that if $y = R x$ then
$$
Uh_{xy}^{\Delta}U\, =\, -\sum_{j=1}^{3} C^{(j)}_x \mathscr{F}(C^{(j)}_x)\, ,
$$
where $C^{(1)} = S_x^{(1)}$, $C^{(2)} = i S_x^{(2)}$ and $C^{(3)} = \sqrt{-\Delta} S_x^{(3)}$.
We note that if $x$ and $y$ are both in $\Lambda_-$ then we do have $h_{xy}^{\Delta}$ in $\mathcal{A}^-_{\R}$
because even though $S_x^{(2)}$ and $S_y^{(2)}$ involve $i$, when we multiply them both
together we only get a real factor, $i^2=-1$.
Also, we observe that for $x,y \in \Lambda_+$ we have $Uh_{xy}^{\Delta}U=h_{xy}^{\Delta}$
(because the $-1$
factors are squared in $S_x^{(2)} S_y^{(2)}$ and $S_x^{(3)}S_y^{(3)}$.
Therefore, if we enumerate the pairs $\{x,y\}$ with $y=Rx$ and $\{x,y\} \in \mathcal{E}(\mathbb{T}^d_N)$
as $\{x_1,y_1\},\dots,\{x_K,y_K\}$, then we have
\begin{equation}
\label{eq:XXZform}
\mathcal{U} H^{\Delta}_{N,d} \mathcal{U}\,
=\, B + \mathscr{F}(B) - \sum_{k=1}^{K} \sum_{j=1}^{3} C^{(j)}_{x_k} \mathscr{F}(C^{(j)}_{x_k})\, ,
\end{equation}
where
$$
B\, =\, \sum_{\{x,y\} \in \mathcal{E}(\Lambda_-)} h_{xy}^{\Delta}\, .
$$
Therefore, we conclude:
\begin{corollary}
\label{cor:XXZgrp}
For $\Delta\leq 1$, the equilibrium state associated to the Hamiltonian $\mathcal{U}H^{\Delta}_{N,d}\mathcal{U}$ is generalized reflection positive.
\end{corollary}
\begin{proof}
From (\ref{eq:XXZform}), the Hamiltonian satisfies the condition (\ref{eq:Hform})
needed to apply Theorem \ref{thm:RPex}. Then by the discussion immediately following
the theorem, the corollary follows.
\end{proof}
\subsubsection{Chessboard estimate with some loss}
\label{CbL}
\begin{proof}[Proof of \Cref{lem:ChessLoss}]
This proof is similar to the proof of Theorem \ref{thm:Holder}. In the present context, Fr\"ohlich
and Lieb refer to this type of inequality as a chessboard estimate.
Consider $\mathcal{P}_{L}^+$ which projects onto all spins $+$ on $\mathbb{B}^d_L$.
Consider the face of $\mathbb{B}^d_L$ whose outward pointing normal
points in the direction $e_k$. One may reflect in the plane separating that face from the neighboring
spin sites outside the box $\mathbb{B}^d_L$.
In order to use reflection positivity, we must conjugate by $U$, introduced in the last subsection, on one of the halves. But we can actually undo this conjugation if we consider the effect this has on the operators
whose expectation we take.
Since $U \mathcal{P}^+_{L} U = \mathcal{P}^-_{L}$ and vice-versa,
the Cauchy-Schwarz inequality implies
$$
\langle \mathcal{P}^+_{L}\rangle_{N, \beta,\Delta}\,
\leq\, \left(\left\langle \prod_{x \in \mathbb{B}^d_{2L}} \left(\frac{1+\varsigma_x^{(L)} S_i^{z}}{2}\right)
\right\rangle_{N, \beta,\Delta}\right)^{(2)^{-d}}\, .
$$
More precisely, each time the Cauchy-Schwarz inequality is applied the operator on one of the two
halves is just the identity operator.
Since the identity operator has expectation 1, this does not contribute.
The effect of having to conjugate by $U$ after each reflection gives rise to the factor $\varsigma_i^{(L)}$.
Note also that one must reflect in each of the $d$ directions. This increases the cardinality
pf the box by a factor of $2^d$. The effect of taking the square-root for each application of the Cauchy-Schwarz inequality results in taking the $2^d$'th root.
One may repeat this procedure $n = \lfloor \log_2(N/L) \rfloor$ times with no essential change to the procedure. The last time (in each direction) however will result in cutting-off what is already a projection
covering more than half of the torus.
That is okay, then in the Cauchy-Schwarz inequality the other factor will be an expectation of a projection
which is always less than or equal to 1.
We bound it by 1.
This results in some loss of sharpness in our inequality, which is why we call this ``Chessboard estimates
with some loss.''
But this is not an essential point for us, because we will frequently be considering the case that $N$ is much larger than $L$. So an extra $1$ added to the already large term $\log_2(N/L)$ does not concern us.
\end{proof}

\subsection{The Six Vertex Model}
\label{sec:6vtx}

We wish to mention that Sutherland's result \Cref{prop:Sutherland} follows from the Yang-Baxter relation.
Let us not mention more about this, here.
One may also understand \Cref{prop:Sutherland} from a diagrammatic standpoint.
But the real goal, to generalize to higher dimensions, eludes us.
The reason this is so useful is that \Cref{prop:Sutherland} along with the fact that both the 6-vertex model
and the XXZ model satisfy the good signs condition of the Perron-Frobenius theorem
implies that the ground state of the XXZ model is the eigenvector of the 6-vertex model
with largest eigenvector.

The six-vertex model is reflection positive, using the result from \cite{FILS}.
We may also condition on the event $\mathcal{M}(N,0)$.
This was not treated in \cite{FILS}.
But the idea has been used before for other reflection positive models.
See the paper \cite{AizenmanEtAl} of Aizenman, et.~al., particularly Appendix B.
Given any vertical or horizontal hyperplane, we may enforce $\mathcal{M}(N,0)$
in a reflection positive way.
For a horizontal hyperplane of vertical edges, we merely restrict to configurations on the plane with the same
number of $\uparrow$'s as $\downarrow$'s.
For a vertical hyperplane of horizontal edges, we can write $\mathbf{1}_{\mathcal{M}(N,0)}(\varsigma) = \sum_{M} \mathbf{1}_{\mathcal{M}(N/2,k)}(\varsigma^{(L)}) 
\mathbf{1}_{\theta^{(1)}_{N/2,1}(\mathcal{M}(N/2,k))}(\varsigma^{(R)})$ where we decompose $\varsigma = (\varsigma^{(L)},\varsigma^{(R)})$.
This is the type of condition one needs for reflection positivity.

With this, the chessboard estimate with loss \Cref{lem:ChessLoss2} follows just like in \Cref{CbL}.

\end{document}